\documentclass[sigconf,nonacm]{acmart}

\AtBeginDocument{%
  }

\setcopyright{none}
\renewcommand\footnotetextcopyrightpermission[1]{}
\newcommand{\etal}{\textit{et al.}}
\newcommand{\sysname}{Deckstacking Resistance\xspace}
\newcommand{\sysnameAc}{DS Resistance\xspace}
\newcommand{\deckstacking}{deckstacking\xspace}
\newcommand{\deckstackingAc}{DS\xspace}
\newcommand{\sysnameTool}{NODS\xspace}

\usepackage{xcolor}
\usepackage{todonotes}
\usepackage{comment}
\usepackage{hyperref}

\definecolor{MPG-dark-green}{HTML}{016C66}
\definecolor{MPG-dark-blue}{HTML}{01445F}
\definecolor{MPG-orange}{HTML}{D06516}
\definecolor{MPG-dark-gray}{HTML}{49494B}
\definecolor{MPG-light-green}{HTML}{96BE2D}
\definecolor{MPG-light-blue}{HTML}{01819D}
\definecolor{MPG-light-gray}{HTML}{EEEEEE}

\newcommand{\subheading}[1]{\vspace{0.5 em}\noindent \textbf{#1}}

\graphicspath{ {./images/} }

\usepackage{listings}
\usepackage{xcolor}

\usepackage{graphicx}
\usepackage{subcaption} %

\usepackage[shortlabels]{enumitem} %
\usepackage{mathpartir} %
\usepackage{listings}
\usepackage{mathtools} %
\usepackage{stmaryrd} %
\usepackage{nicefrac} %
\usepackage{multirow}
\usepackage{pifont}
\usepackage{mathrsfs} %
\usepackage{relsize}
\usepackage{booktabs} 
\usepackage{wrapfig}

\usepackage{algorithm} %
\usepackage[noend]{algorithmic}

\usepackage{xifthen}
\newcommand{\ifempty}[3]{%
  \ifthenelse{\isempty{#1}}{#2}{#3}%
}

\newcommand{\codefont}{\fontsize{9}{9}\selectfont}

\definecolor{Magenta}{HTML}{FF00FF}
\definecolor{ForestGreen}{HTML}{028A0F}

\def\tokColor{\color{Magenta}}
\def\txColor{\color{MidnightBlue}}
\definecolor{MidnightBlue}{HTML}{006795}
\definecolor{Plum}{HTML}{92268F}

\usepackage{xspace}

\newcommand{\procF}{\txcode{f}}
\newcommand{\procG}{\txcode{g}}

\newcommand{\code}[1]{{\tt\codefont{#1}}}

\newcommand{\txcode}[1]{{\text{\tt\codefont{\txColor{#1}}}}}

\def\txColor{\color{MidnightBlue}}

\newcommand{\txFmt}[1]{{\txColor{\sf #1}}}

\newcommand{\tx}[2][]{\txFmt{#2}_{\txColor{#1}}}
\newcommand{\txT}[1][]{\tx[#1]{tx}} %
\newcommand{\txB}[1][]{\tx[#1]{\tau}} %
\newcommand{\TxU}[1][]{\tx[#1]{\mathbb{X}}} %

\def\tokColor{\color{Magenta}}
\newcommand{\tokFmt}[1]{{\tokColor{\tt #1}}}
\newcommand{\tok}[2][]{\tokFmt{#2}_{\tokColor{#1}}\xspace}

\newcommand{\tokT}[1][]{\tok[{#1}]{T}}    %

\def\contrColor{\color{MidnightBlue}}
\newcommand{\contrFmt}[1]{{\contrColor{\code{#1}}}}
\newcommand{\contrC}[1][]{\mathord{\contrFmt{C}_{\contrColor{#1}}}}

\lstdefinelanguage{txscript}{
    basicstyle=\tiny,
	commentstyle=\color{Gray},
	morecomment=[l]{//},
	morecomment=[s]{/*}{*/},
	classoffset=0,
        escapechar=\$,
	morekeywords={if,then,else,contract,emit,skip,require,fun,return,for},
	keywordstyle=\color{Plum}\bfseries,
	classoffset=1,
	morekeywords={&&,||,!},
	keywordstyle=\color{Plum},
	classoffset=2,
	morekeywords={},
	keywordstyle=\color{MidnightBlue}\bfseries,
}

\lstdefinelanguage{txscriptbig}{
	commentstyle=\color{Gray},
	morecomment=[l]{//},
	morecomment=[s]{/*}{*/},
	classoffset=0,
        escapechar=\$,
	morekeywords={if,then,else,contract,emit,skip,require,fun,return,for},
	keywordstyle=\color{Plum}\bfseries,
	classoffset=1,
	morekeywords={&&,||,!},
	keywordstyle=\color{Plum},
	classoffset=2,
	morekeywords={},
	keywordstyle=\color{MidnightBlue}\bfseries,
}

\newcommand{\WmvA}[1][]{W_{#1}}
\newcommand{\WmvAi}[1][]{W'_{#1}}

\newcommand{\WalU}[1][]{\mathbb{W}_{#1}} %

\newcommand{\waldistrarrow}[1]{\approx_{\$}}

\newcommand{\ContrU}[1][]{\mathbb{C}_{#1}} %

\newlength\replength
\newcommand\repfrac{.1}

\newcommand\rulewidth{.6pt}
\newcommand\tdashfill[1][\repfrac]{\cleaders\hbox to \replength{%
  \smash{\rule[\arraystretch\ht\strutbox]{\repfrac\replength}{\rulewidth}}}\hfill}

\newcommand\tdotfill[1][\repfrac]{\cleaders\hbox to \replength{%
  \smash{\raisebox{\arraystretch\dimexpr\ht\strutbox-.1ex\relax}{.}}}\hfill}

\newcommand{\true}{\mathit{true}}
\newcommand{\false}{\mathit{false}}

\newtheorem*{remark}{Remark}
\theoremstyle{plain}
\newtheorem{assumption}{Assumption}

\newcommand{\config}{\Gamma}
\newcommand{\txv}{\textit{tx}}

\newcommand{\contract}{\contrC}
\newcommand{\cstrat}{\Sigma}

\newcommand{\astrat}{\textcolor{MPG-orange}{A}}
\newcommand{\sstrat}{\textcolor{MPG-light-blue}{S}}

\newcommand{\honusers}{\textcolor{MPG-dark-green}{\textit{Hon}}}
\newcommand{\ustrat}[3]{\cstrat^{#1}_{#3}}
\newcommand{\sustrat}[1]{\cstrat^{#1}}

\newcommand{\combine}[2]{(#1,#2)}
\newcommand{\semantics}{\vdash}

\newcommand{\mempool}{\mathcal{P}}
\newcommand{\bblock}{\mathcal{B}}
\newcommand{\isSim}[1]{\Pi(#1)}

\newcommand{\run}{R}

\newcommand{\symrun}{R^{\sstrat}}
\newcommand{\arun}{\run^{\astrat}}
\newcommand{\srun}{\symrun}

\newcommand{\conf}{\Gamma}
\newcommand{\trans}[2][]{\xrightarrow[#1]{#2}}

\newcommand{\obs}[0]{xs}
\newcommand{\listtype}[1]{\mathscr{L}(#1)}

\newcommand{\eqth}{\varepsilon}

\newcommand{\mdot}{.~}

\newcommand{\sseq}{\mempool}
\newcommand{\sender}[1]{\textit{sender}(#1)}

\newcommand{\getblocknumber}{\code{now}}
\newcommand{\blocknumber}{b}
\newcommand{\contracts}{\mathcal{C}}

\newcommand{\secD}[1][\eqth]{\vDash_{#1}}
\newcommand{\secrets}{\mathcal{N}}

\newcommand{\contrCall}[1]{\contrC.#1}
\newcommand{\contrS}[1][]{\mathord{\contrFmt{S}_{\contrColor{#1}}}}

\newcommand{\DSone}{\ensuremath{D_{\text{full}}}}
\newcommand{\DStwo}{\ensuremath{D_{\text{avail}}}}
\newcommand{\DSthree}{\ensuremath{D_{\text{fixes}}}}

\newcommand{\tick}[1]{\overline{#1}}

\newcommand{\cstate}[2]{{#1}.\contrS_{#2}}
\newcommand{\cblocknumber}[1]{\delta(#1)}
\newcommand{\cvars}[1]{\mathcal{X}_{#1}}
\newcommand{\ctx}[1]{\txv_{#1}}
\newcommand{\equivfor}[1]{\approx_{#1}}
\newcommand{\varset}{X}

\newcommand{\varfixpoint}{X^*}
\newcommand{\criticalvars}{X_\textit{Ev}}
\newcommand{\honuser}{\textcolor{MPG-dark-green}{\textit{u}}}
\newcommand{\ssteps}[2]{\xrightarrow[#1]{#2}^*}
\newcommand{\txlist}{\vec{\txv}}
\newcommand{\obslist}{\vec{\obs}}

\newcommand{\perm}{\pi}
\newcommand{\equalupto}[1]{=_{\backslash {#1}}}
\newcommand{\equaluptotx}[1]{\equalupto{\{ #1 \}}}
\newcommand{\itattacker}[4]{\textit{step}(#1, #2, #3, #4)}
\newcommand{\itattackerN}{\textit{step}}
\newcommand{\getAttackrun}[5][\confPred]{\textit{aRun}^{#1}(#2, #3, #4, #5)}
\newcommand{\getAttackrunN}{\textit{aRun}}
\newcommand{\bblocks}{\vec{\bblock}}
\newcommand{\blocks}{\mathcal{L}}

\newcommand{\ablockspre}{\blocks^{\astrat}_0}

\newcommand{\runforblocks}[2]{}
\newcommand{\initconf}{\conf_{\textit{init}}}
\newcommand{\arunpre}{\arun_{\textit{pre}}}
\newcommand{\arunpost}{\arun_{\textit{post}}}
\newcommand{\arunround}[1]{\arun_{#1}}
\newcommand{\srunpre}{\srun_{\textit{pre}}}
\newcommand{\srunpost}{\srun_{\textit{post}}}
\newcommand{\srunround}[1]{\srun_{#1}}
\newcommand{\aconf}{\conf_{\astrat}}
\newcommand{\sconf}{\conf_{\sstrat}}
\newcommand{\filter}[2]{{#1} \downarrow_{#2}}

\newcommand{\simsim}{\sim_{\sstrat}}
\newcommand{\gensstrat}[1][\confPred]{\sstrat^{#1}_{\astrat}}

\newcommand{\obsone}[2]{\langle #1: #2\rangle}
\newcommand{\eventvar}{x}

\newcommand{\obsc}[1]{\obsone{\eventvar}{#1}}

\newcommand{\blockvar}{\text{\lstinline|block.number|}}

\newcommand{\singlesubst}[2]{{#1} \rightarrow {#2}}
\newcommand{\nil}{\epsilon}
\newcommand{\applysubst}[2]{{#1}[{#2}]}
\newcommand{\applysubstblock}[2]{{#1}[{#2}]_{\bblock}}
\newcommand{\multisubst}[1]{#1}
\newcommand{\simsubst}[2]{#1 \rightarrow #2}

\newcommand{\getTerms}[1]{#1 \downarrow^{\mathcal{T}}}
\newcommand{\sterms}{\mathcal{T}}

\newcommand{\concat}[2]{{#1} \cdot {#2}}

\newcommand{\length}[1]{|#1|}
\newcommand{\getBlocks}{\textsf{blocks}} %

\newcommand{\lastBlockNumber}{\delta} %

\newcommand{\rstrat}{\Sigma_r}
\newcommand{\ablock}[1]{\bblock^{\astrat}_{0,#1}}

\newcommand{\amempool}[1]{\mempool^{\astrat}_{#1}}

\newcommand{\roundstrat}[1]{#1^{\rightarrow}}
\newcommand{\getTxs}{\textsf{TXs}}

\newcommand{\memcount}[1][]{\mempool_{#1}}
\newcommand{\acblock}[1][]{\bblock^{\astrat}_{#1}}

\newcommand{\hobs}{\obs_{\textsf{m}}}
\newcommand{\aobs}{\obs_{\astrat}}
\newcommand{\round}{\rho}

\newcommand{\htxv}{\txv_{\honuser}}
\newcommand{\atxv}{\txv_{\astrat}}

\newcommand{\obscV}[2]{\langle #1: #2 \rangle}

\newcommand{\tsMap}{\text{\lstinline|timestamps|}}

\newcommand{\token}{\texttt{T}}

\newcommand{\testFunc}{\text{\lstinline|test|}}
\newcommand{\triggerFunc}{\text{\lstinline|trigger|}}
\newcommand{\triggeredEvent}{\text{\lstinline|Triggered|}}
\newcommand{\shaFunc}{\text{\lstinline|sha256|}}

\newcommand{\intConds}{interaction conditions}
\newcommand{\intCond}{interaction condition}

\newcommand{\symExRel}[1]{\textit{SE}^{\indexvar}_{#1}} %
\newcommand{\symExReli}[2]{\textit{SE}^{#2}_{#1}} %
\newcommand{\seSub}{\theta} %
\newcommand{\seCond}{\varphi} %
\newcommand{\seRes}[2]{\langle #1 , #2 \rangle} %
\newcommand{\seAss}{\mathcal{V}} %
\newcommand{\confToAss}[2]{\seAss^{\indexvar}_{\contract}(#1, #2)} %
\newcommand{\confToAssi}[3]{\seAss^{#3}_{\contract}(#1, #2)} %
\newcommand{\fargs}{\overrightarrow{v}} %
\newcommand{\assComp}[2]{#1 \circ {#2}} %
\newcommand{\seExp}{e} %
\newcommand{\evAss}[2]{#1(#2)} %

\newcommand{\expset}{\mathcal{E}^{\indexvar}_{\funvar}} %
\newcommand{\expseti}[2]{\mathcal{E}_{#1}^{#2}} %
\newcommand{\seCvars}[1]{\mathcal{X}^{\indexvar}_{#1,\funvar}} %
\newcommand{\seCvarsi}[3]{\mathcal{X}^{#3}_{#1, #2}} %
\newcommand{\satis}[2]{#1 \models #2} %
\newcommand{\ddeps}[2]{\textit{DDeps}^{\indexvar}_{#1}(#2)} %
\newcommand{\cdeps}[2]{\textit{CDeps}^{\indexvar}_{#1}(#2)} %
\newcommand{\ddepsi}[3]{\textit{DDeps}^{#3}(#2 \xleftarrow{\scalebox{0.6}{#1}}\cdot)} %
\newcommand{\cdepsi}[3]{\textit{CDeps}^{#3}_{#1}(#2)} %
\newcommand{\expvars}[1]{\textit{Vars}(#1)} %
\newcommand{\expvarsi}[2]{\textit{Vars}^{#2}(#1)} %
\newcommand{\seCvar}{x} %
\newcommand{\condvars}[1]{\textit{Vars}(#1)} %
\newcommand{\sementails}{\vDash} %
\newcommand{\pathcond}[3]{\textit{PathCond}^{\indexvar}_{#1,#2}(#3)} %
\newcommand{\pathcondi}[4]{\textit{PathCond}^{#4}(#3 \xleftarrow{\scalebox{0.6}{#1}}#2)} %
\newcommand{\nowritecond}[2]{\textit{NoWrite}^{\indexvar}(#2\xleftarrow{\scalebox{0.6}{#1}}\cdot)} %
\newcommand{\nowritecondi}[3]{\textit{NoWrite}^{#3}(#2 \xleftarrow{\scalebox{0.6}{#1}}\cdot)} %
\newcommand{\nowriteset}{\overline{W}} %
\newcommand{\roundPred}{P_{r}} %
\newcommand{\roundinv}[3]{\textit{RoundInv}(#2, #3)} %
\newcommand{\execinv}[3]{\textit{Inv}^{\indexvar}(#1, #2, #3)} %
\newcommand{\roundinvi}[4]{\textit{RoundInv}^{#4}(#2, #3)} %
\newcommand{\execinvi}[4]{\textit{Inv}^{#4}(#1, #2, #3)} %
\newcommand{\inv}{\phi_i} %
\newcommand{\invpre}{\phi_{p}} %
\newcommand{\predIsSender}[1]{\sendervari{\userindexvar} = {#1}}
\newcommand{\predIsNoSender}[1]{\sendervari{\nuserindexvar} \neq {#1}}
\newcommand{\roundvar}{r} %
\newcommand{\procTick}{\txcode{tick}}
\newcommand{\writeseti}[5]{\textit{Writes}^{}_{}(#2, #3, #4)} %
\newcommand{\writeset}[3]{\writeseti{}{#1}{#2}{#3}{}}

\newcommand{\frCondPre}{\Phi^{\textit{fr}}_\textit{pre}} %
\newcommand{\frCondInv}{\Phi^{\textit{fr}}_\textit{inv}} %
\newcommand{\brCondPre}{\Phi^{\textit{br}}_\textit{pre}} %
\newcommand{\brCondInv}{\Phi^{\textit{br}}_\textit{inv}} %
\newcommand{\frWriteCond}[4]{\textit{FR}^{\honuser}(#3 \xleftarrow{\scalebox{0.6}{#1}} #2)}
\newcommand{\invaccumulator}{(\invpre, \inv)}
\newcommand{\frAssSet}[3]{\textbf{FR}^{\honuser}_{#1}(#2, #3)}
\newcommand{\frAssSetPre}[3]{\textbf{pFR}^{\honuser}_{#1}(#2, #3)}

\newcommand{\brAssWrite}[3]{\textbf{BR}^{\honuser}_{#1,#2}(#3)}
\newcommand{\brAssWritePre}[3]{\textbf{pBR}^{\honuser}_{#1,#2}(#3)}

\newcommand{\brWriteCond}[4]{\textit{BR}^{\honuser}({#2 {\perp} [#3 \xleftarrow{\scalebox{0.6}{#1}}]})} %

\newcommand{\cFuncs}[1]{\mathcal{F}_{#1}} %
\newcommand{\blocknumvar}{\texttt{b}} %
\newcommand{\sendervar}{\text{\lstinline|sender|}} %
\newcommand{\sendervari}[1]{\sendervar^{#1}}
\newcommand{\ownervar}{\text{\lstinline|owner|}} %

\newcommand{\dSet}{D} %
\newcommand{\cSet}{C} %
\newcommand{\vSet}{Z} %
\newcommand{\doneflag}{\textit{done}} %
\newcommand{\seSubst}[2]{[#1 \rightarrow #2]} %
\newcommand{\preCond}{\varphi_{\textit{pre}}}

\newcommand{\gvars}[1]{\mathcal{G}_{#1}} %
\newcommand{\lvars}[2]{\mathcal{L}_{#1}^{#2}} %
\newcommand{\indexvar}{i} %
\newcommand{\userindexvar}{\honuser} %
\newcommand{\nuserindexvar}{\neg\honuser} %
\newcommand{\funvar}{\procF}
\newcommand{\bigforall}[1]{\mathop{\mathlarger{\mathlarger{\mathlarger{\forall}}}}_{#1}}
\newcommand{\setforall}[2]{\bigforall{#1 \in #2}}
\newcommand{\setforallshort}[1]{\bigforall{#1}\mdot}
\newcommand{\locvar}{\ell} %
\newcommand{\condVar}{\phi}
\newcommand{\confPred}{\condVar}
\newcommand{\getRoundInv}{\textsf{get}_{\textit{RI}}}
\newcommand{\mergeRoundInv}{\textsf{merge}_{\textit{RI}}}
\newcommand{\select}{\textbf{select}}

\newcommand{\depclosed}[1]{\textit{closed}^{\textit{deps}}(#1)}
\newcommand{\WFset}{W_{\honuser}}
\newcommand{\NWFset}{\varfixpoint \setminus \WFset}
\newcommand{\WGset}{W_{\astrat}}
\newcommand{\NWGset}{\varfixpoint \setminus \WGset}
\newcommand{\NWcap}{\textit{NW}}

\newcommand{\relVars}[1]{\textit{RelVars}(#1)}

\newcommand{\frCond}[3]{\textit{FR}^{\honuser}(#1, #2, #3)}
\newcommand{\frCondR}[3]{\textit{FR}^{\honuser}_{\rho}(#1, #2, #3)}
\newcommand{\brCond}[3]{\textit{BR}^{\honuser}(#1, #2, #3)}
\newcommand{\brCondR}[3]{\textit{BR}^{\honuser}_{\rho}(#1, #2, #3)}

\newcommand{\cmark}{\ding{51}}%
\newcommand{\xmark}{\ding{55}}%
\newcommand{\lmark}{$\lightning$}%

\newcommand{\lastconfig}[1]{\config_{#1}}

\definecolor{verylightgray}{rgb}{.97,.97,.97}

\lstdefinelanguage{Solidity}{
	keywords=[1]{anonymous, assembly, assert, balance, break, call, callcode, case, catch, class, constant, continue, constructor, contract, debugger, default, delegatecall, delete, do, else, emit, event, experimental, export, external, false, finally, for, function, gas, if, implements, import, in, indexed, instanceof, interface, internal, is, length, library, log0, log1, log2, log3, log4, memory, modifier, new, payable, pragma, private, protected, public, pure, push, require, return, returns, revert, selfdestruct, send, solidity, storage, struct, suicide, super, switch, then, this, throw, transfer, true, try, typeof, using, value, view, while, with, addmod, ecrecover, keccak256, mulmod, ripemd160, sha256, sha3}, %
	keywordstyle=[1]\color{blue}\bfseries,
	keywords=[2]{address, bool, byte, bytes, bytes1, bytes2, bytes3, bytes4, bytes5, bytes6, bytes7, bytes8, bytes9, bytes10, bytes11, bytes12, bytes13, bytes14, bytes15, bytes16, bytes17, bytes18, bytes19, bytes20, bytes21, bytes22, bytes23, bytes24, bytes25, bytes26, bytes27, bytes28, bytes29, bytes30, bytes31, bytes32, enum, int, int8, int16, int24, int32, int40, int48, int56, int64, int72, int80, int88, int96, int104, int112, int120, int128, int136, int144, int152, int160, int168, int176, int184, int192, int200, int208, int216, int224, int232, int240, int248, int256, mapping, string, uint, uint8, uint16, uint24, uint32, uint40, uint48, uint56, uint64, uint72, uint80, uint88, uint96, uint104, uint112, uint120, uint128, uint136, uint144, uint152, uint160, uint168, uint176, uint184, uint192, uint200, uint208, uint216, uint224, uint232, uint240, uint248, uint256, var, void, ether, finney, szabo, wei, days, hours, minutes, seconds, weeks, years},	%
	keywordstyle=[2]\color{teal}\bfseries,
	keywords=[3]{block, blockhash, coinbase, difficulty, gaslimit, number, timestamp, round, msg, data, gas, sender, sig, value, now, tx, gasprice, origin},	%
	keywordstyle=[3]\color{violet}\bfseries,
	identifierstyle=\color{black},
	sensitive=true,
	comment=[l]{//},
	morecomment=[s]{/*}{*/},
	commentstyle=\color{gray}\ttfamily,
	stringstyle=\color{red}\ttfamily,
	morestring=[b]',
	morestring=[b]",
    mathescape= true,
    frame=single   
 }

\begin{document}

\title{On Identifying Sound Conditions for Frontrunning Resistance}

\author{Sebastian Holler}
\email{sebastian.holler@mpi-sp.org}
\affiliation{%
  \institution{MPI-SP}
  \city{Bochum}
  \country{Germany}
}

\author{Anna Piscitelli}
\email{anna.piscitelli@rub.de}
\affiliation{%
  \institution{Ruhr University Bochum}
  \city{Bochum}
  \country{Germany}
}

\author{Jannik Albrecht}
\email{jannik.albrecht@ruhr-uni-bochum.de}
\affiliation{%
  \institution{Ruhr University Bochum}
  \city{Bochum}
  \country{Germany}
}

\author{Stephan D{\"u}bler}
\email{stephan.duebler@mpi-sp.org}
\affiliation{%
  \institution{MPI-SP}
  \city{Bochum}
  \country{Germany}
}

\author{Ghassan Karame}
\email{ghassan@karame.org}
\affiliation{%
  \institution{Ruhr University Bochum}
  \city{Bochum}
  \country{Germany}
}

\author{Clara Schneidewind}
\email{clara.schneidewind@mpi-sp.org}
\affiliation{%
  \institution{MPI-SP}
  \city{Bochum}
  \country{Germany}
}

\begin{abstract}

Blockchains enable decentralized applications through smart contra\-cts---interactive programs executed 
through consensus.
However, the inherently asynchronous nature of blockchain transaction ordering introduces a class of vulnerabilities known as \emph{frontrunning attacks}, which have caused millions of dollars in losses in major blockchains, such as Ethereum.
Frontrunning attacks arise because users interact with smart contracts through transactions, which are added to the blockchain by designated nodes called \emph{miners}. Miners can exploit their ability to reorder, delay, or insert transactions to gain an advantage over honest users, effectively frontrunning them. 

Yet, to date, the field lacks a rigorous definition of what it even means for a contract to resist such attacks. Worse, we show that existing dynamic detection approaches are fundamentally inadequate: in a large-scale study comprising 287 smart contract audits, 55\% of the 393 reported vulnerabilities identified by leading smart contract auditors fall outside the scope of state-of-the-art detection criteria. 
 To address this gap, we propose the first formal definition of frontrunning vulnerability for smart contracts. Our definition captures a key insight: resistance to frontrunning is not an intrinsic property of a contract alone, but depends critically on how honest users interact with it. Grounded in this observation, we develop a sound algorithm for synthesizing secure interaction conditions, alongside a prototype implementation that we apply to audited real-world contracts—revealing previously undiscovered vulnerabilities in two Ethereum contracts.

\end{abstract}



\maketitle

\section{Introduction}
Blockchains allow mutually mistrusting entities to jointly maintain a fairly evolving and tamper-resistant ledger of transactions. 
To extend the ledger, a designated set of nodes, so-called \emph{miners}, verifies
user transactions from a public \emph{mempool}
and include them in blocks that are later appended to the ledger. This mechanism not only supports financial transactions but also the execution of smart contracts; these are (stateful) programs that can manage the transfer of funds according to some contract code that is invoked on the ledger. 
Smart contracts give rise to a rich environment for deploying decentralized financial applications (often referred to as DeFi), encompassing lending platforms, currency exchanges, or decentralized autonomous organizations.

Although smart contracts enable powerful deployment of such applications, they differ significantly from conventional program execution models, creating opportunities for so-called frontrunning attacks. In particular, since the results of contract invocations only come into effect once the corresponding transaction is appended to the ledger, miners have an unprecedented advantage in inspecting transactions in the mempool and manipulating their execution order~\cite{DBLP:conf/fc/EskandariMC19,daian2019flash,babel2023clockwork,DBLP:conf/fc/McCorryHM18, torres2021frontrunner, DBLP:conf/sp/ZhouQCLG21}.
Many real-world incidents show that miners can easily take advantage of slippage opportunities in currency exchanges in a targeted manner~\cite{babel2023clockwork,zhou2021high} or prevent unwanted user actions on demand by outrunning them with specially-crafted transactions when observing their publication~\cite{DBLP:conf/fc/EskandariMC19,DBLP:conf/uss/BreindenbachDTJ18}. For instance, \mbox{Qin \etal~\cite{DBLP:conf/sp/QinZG22}} report losses of \$540.54M over 32 months (Dec 2018–Aug 2021) solely due to arbitrage attacks. 

Interestingly, despite the prevalence of frontrunning attacks, as of today, there exists no precise formal understanding of when a smart contract is susceptible to such attacks. 
Several works observed a connection between frontrunning attacks and concurrency bugs, resulting in the notion of \emph{transaction order dependence} (short TOD). Intuitively, TOD deems contracts vulnerable whose concurrent execution does not safely commute. 
While TOD is a prerequisite for the existence of frontrunning attacks, almost all realistic smart contracts exhibit benign instances of TOD, making TOD an insufficient criterion for identifying frontrunning vulnerabilities. Several practical works~\cite{zhang2024nyx,bose2022sailfish} tried to narrow down TOD to identify exploitable instances, but all such attempts are of purely heuristic nature and turn out to be neither sound nor complete.

Complementing these attempts towards the static identification of frontrunning vulnerabilities, there exist monitoring approaches that aim to dynamically detect frontrunning attacks. 
Here, the commonly used criterion to identify frontrunning attacks is the notion of \emph{maximal-extractable value}~\cite{daian2019flash, babel2023clockwork} (short \emph{MEV}). MEV aims to quantify the maximum financial profit that a frontrunning adversary could make when being confronted with a specific set of user transactions to be scheduled. 
If an attacker has a positive MEV, this indicates an attacker's ability to carry out a profitable frontrunning attack.
This approach, thus, has the advantage of identifying concrete attacks, but disregards a large class of frontrunning attacks that do not yield an immediate financial gain for the attacker.
Indeed, in a large-scale analysis of 287 smart contract audits conducted by top auditing companies, we found that 55\% of the 393 %
reported vulnerabilities do not yield an immediate financial gain for the attacker and, thus, cannot be captured by MEV (among those, 61 (28\%) are critical or major vulnerabilities).
\emph{These findings reveal a significant gap in 
current approaches used to characterize a contract's susceptibility to frontrunning attacks.}

\begin{figure}[tp]
    \centering
    \includegraphics[width=0.7\columnwidth]{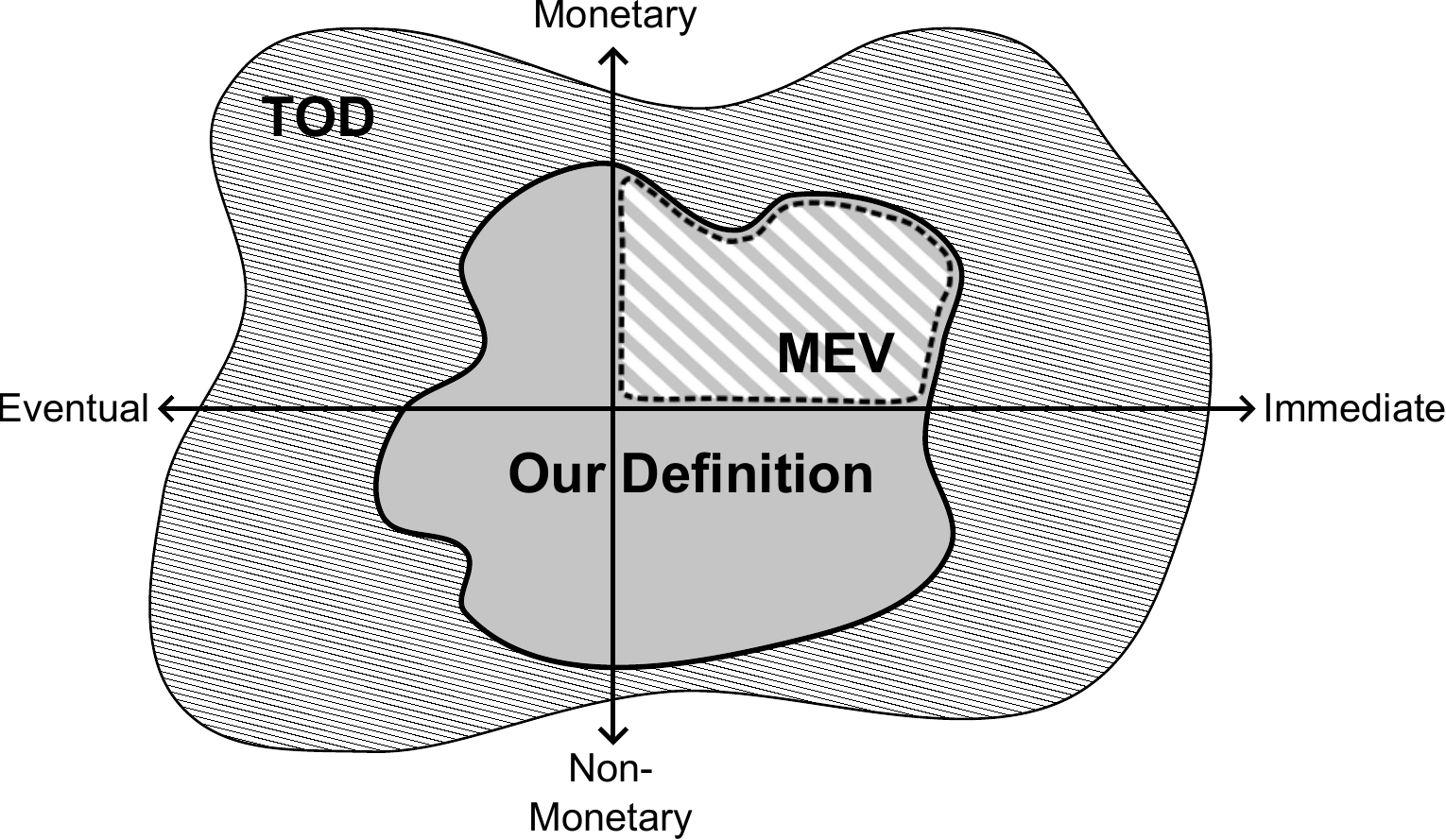}
    \caption{
        Our definition establishes a strong tradeoff between the narrow focus of MEV on monetary and immediate attacks and the too broad notion of TOD.
        The outer dashed area represents possible benign instances, while the inner gray region indicates where frontrunning attacks are possible. 
    }
    \label{fig:MAS_illustration}
\end{figure}

Motivated by these findings, we set forth in this work to develop the first definition that precisely characterizes when a smart contract is robust against frontrunning attacks.
We base our definition on the key observation that frontrunning attacks constitute a malicious interference of an attacker with the way that an honest user interacts with a contract. Consequently, a contract's resistance to frontrunning attacks cannot be solely determined by the contract code alone, but needs to consider the nature of user interactions. 
Following this insight, we propose a simulation-based definition of frontrunning resistance that is formulated w.r.t. a contract and the interaction strategy of an honest user.
Based on this definition, we then develop an algorithm to synthesize secure \emph{interaction conditions} that formulate requirements on the user strategy to provably evade frontrunning attacks.
In summary, we make the following contributions:
\begin{description}[leftmargin = 0.5 cm]
\item[Empirical Analysis of Frontrunning Attacks:] We study a real-world data set of frontrunning vulnerabilities reported in professional smart contract audits and analyze shortcomings of existing notions to capture these attacks.
In particular, we show that existing frameworks cannot capture a significant fraction of reported frontrunning instances by top auditors (cf.~Section~\ref{sec:problem-statement}).
\item[Novel Definition:] 
We introduce a new formal definition that captures frontrunning resistance as an attacker's inability to interfere in a targeted manner with the interactions between honest users and smart contracts (cf.~Section~\ref{sec:formal}). 
This definition addresses the shortcomings of the coarse notion of TOD by formally characterizing those scenarios where an attacker with mining capabilities can exploit TOD to purposefully cause interferences with honest user transactions.
At the same time, the definition is (as opposed to MEV) general enough to also capture such frontrunning attacks that do not result in immediate monetary attacker gains (cf.~Figure~\ref{fig:MAS_illustration}).

\item[Synthesis of Interaction Conditions:] We devise an algorithm to automatically synthesize secure \emph{interaction conditions} based on the code of a smart contract. We show that user strategies that respect those conditions are provably robust against the interference of a frontrunning attacker
 (cf. Section~\ref{section:conditions}).
\item[Tool Implementation and Evaluation:] We implement our approach in a prototype tool and evaluate it on a benchmark of real-world smart contracts (cf. Section~\ref{sec:evaluation}). In the course of the evaluation, we identified two fixes implemented in response to audit reports to remain susceptible to frontrunning attacks.
We responsibly disclosed these findings to the contract owners and the audit companies. 
We additionally disclosed our results to the Enterprise Ethereum Alliance,
resulting in an adjustment of
 their vulnerability classification in the EEA EthTrust Security Levels Specification~\cite{nevile_2023}. 
All the code and data of the prototype is available on GitHub~\cite{deckstacking2026nods}. 
\end{description}

\section{Critical Gaps in Frontrunning Detection} \label{sec:problem-statement}

\subsection{Background}

Frontrunning attacks pose a significant threat to the fairness and security of smart-contract-based applications.
These attacks are enabled due to the specific execution model of smart contracts on public blockchains, which allows adversaries to observe and manipulate pending user transactions before they are finalized on the blockchain.
While the considerations in this paper apply to various blockchain platforms, we focus in the following on Ethereum as it is the most widely used platform for smart contracts.

\subheading{Ethereum Smart Contracts.}
Ethereum smart contracts are written in a high-level language before being compiled into a bytecode representation that is run on the Ethereum blockchain.
For simplicity, all examples in this paper are given in the high-level language Solidity.
Solidity is a Java-like programming language featuring contracts (instead of classes).
We show an example of a Solidity smart contract implementing a \textit{hash puzzle} in Figure~\ref{fig:sc:hash-puzzle}. 
This smart contract is initialized with a \emph{challenge}, i.e., a hash value, and allows accounts to propose a \emph{solution} to this challenge by calling the \lstinline|submit| function. 
When submitting a valid \emph{solution} that is the preimage of the \emph{challenge} for the first time, the calling account (accessed via \lstinline|msg.sender|) is rewarded. 
This is enforced by the \lstinline|require(...)| statement in line \lstinline|5|, which checks that the challenge condition \lstinline|sha256(solution) == challenge| is met and reverts the transaction execution otherwise.

\subheading{Frontrunning Attacks.}
The \lstinline|HashPuzzle| contract (cf.~Figure~\ref{fig:sc:hash-puzzle}) is vulnerable to a frontrunning attack: 
Suppose a user finds the solution to the \emph{challenge} and submits a corresponding \lstinline|submit(solution)| transaction to redeem the reward.
Transactions submitted for inclusion in the blockchain enter the \emph{mempool}. 
Miners collect the transactions of network users and group them into blocks that extend the blockchain. 
Hence, an adversarial miner with access to the mempool learns the \emph{solution} in the honest user transaction before this transaction gets included in the blockchain.
Instead, the miner can craft their own \lstinline|submit(solution)| transaction and place it in front of the honest user transaction when creating a block. 
Thus, the adversary will get rewarded for solving the puzzle while the original user who submitted the \emph{solution} loses their rightful reward.

\begin{figure}[tp]
\scriptsize
    \begin{lstlisting}[basicstyle=\scriptsize\ttfamily, numberstyle=\tiny]
contract HashPuzzle {
  bytes challenge;
  
  function submit(bytes solution) {
    require(sha256(solution) == challenge);
    payable(msg.sender).transfer(this.balance);   } // ...
    \end{lstlisting}
\caption{Hash-puzzle contract vulnerable to frontrunning.}
\label{fig:sc:hash-puzzle}
\end{figure}

\begin{figure}[tp]
\scriptsize
    \begin{lstlisting}[basicstyle=\scriptsize\ttfamily, numberstyle=\tiny]
contract ProposeVoting is Voting {
  mapping (Id => Proposal) proposals;

  function proposeCandidate(Id id, Proposal prop){
    require(proposals[id] == null);
    proposals[id] = prop;   } // ...
\end{lstlisting}
\caption{Voting contract vulnerable to denial-of-service frontrunning attacks.}
\label{fig:sc:voting-dos}
\end{figure}

\subsection{Why Static and Dynamic Frontrunning Detection Falls Short} \label{ps-sota}
Despite the absence of a precise formal definition of frontrunning vulnerabilities, various approaches have been proposed to detect such vulnerabilities in smart contracts, both statically (based on the contract code) and dynamically (monitoring the blockchain and the mempool).
In the following, we analyze the shortcomings of these approaches, starting with dynamic analysis, as these are particularly evident in the case of frontrunning.

\subsubsection{Dynamic Analysis for Frontrunning Attacks.}
A prominent line of work towards dynamic frontrunning detection centers around the notion of \emph{Maximal-Extractable Value} (MEV)~\cite{daian2019flash}.
MEV quantifies the monetary gain for a frontrunning attacker who inspects the current mempool and the blockchain state. 
In practice, MEV is the basis for the implementation of so-called \emph{MEV bots}~\cite{babel2023clockwork,zhang2023your}, network nodes that monitor the mempool to identify profitable arbitrage and sandwiching opportunities and then bribe miners to enforce their determined order of execution.
Formally, MEV is defined as the maximum personal gain that a blockchain participant can obtain by reordering future transactions at a given point in time, provided a concrete mempool~\cite{babel2023clockwork}.

While MEV is effective in characterizing specific profitable attack opportunities, it is only designed to detect attacks with the following characteristics: 
(1) the adversary's goal needs to be of \emph{monetary} nature and measurable as an increase in the attacker's wealth, and (2) attacks need to be \emph{immediate}, meaning that an attacker can mount the attack using only the information from the current mempool to create a (limited) sequence of blocks that meets the attack goal.

\subheading{(Non-)Monetary Attacks.}
MEV identifies quantifiable attacks by measuring value changes in user-owned tokens.
However, not all frontrunning attacks incur direct financial gains for an attacker. 
For example, adversaries can have (personal) incentives to harm specific users even without causing state changes that provide them with benefits.
Namely, frontrunning can be used to provoke user- and smart-contract-specific denial-of-service attacks.
Such attacks can lead to significant harm, including monetary loss or permanent disincentives for the use of a service.
Figure~\ref{fig:sc:voting-dos} showcases the candidate proposal functionality of an e-voting contract that is vulnerable to such a DoS-Frontrunning attack.
Here, an attacker can prevent a new candidate from being proposed by frontrunning every \lstinline|proposeCandidate| call with their own call that uses the same identifier.
The attacker call will reserve the candidate identifier \lstinline|id| for a dummy proposal, thereby causing the intentionally user-proposed candidate to be dropped.
Any re-submission of the candidate with a new identifier could be frontrun again by an attacker monitoring the mempool, continuously preventing specific candidate proposals.

\subheading{Immediate vs. Eventual Attacks.}
The definition of MEV is inherently \emph{immediate} in that it identifies attacks that
(1) can be completed within the transaction scheduling window, under the control of the attacker, and 
(2) use transactions from the current mempool or such transactions that the attacker can craft independently.
However, some frontrunning attacks may have effects that are only observable after the attacker's scheduling window, e.g., because they incur a state change that becomes exploitable only in later phases of contract execution.
For instance, consider a voting contract in which a frontrunning attacker manipulates the voting settings (e.g., the maximum number of votes a user can cast) to their advantage during the voting phase. The possible effects of such an attack (e.g., winning the election and receiving a payout) may only materialize after the voting phase ends, which could fall outside the attacker's scheduling window.
Such non-immediate attacks lie outside the scope of MEV, and we refer to them as \emph{eventual} attacks.

\subsubsection{Static Analysis for Frontrunning Detection.}
Approaches to the static analysis of frontrunning vulnerabilities center on the notion of \emph{transaction order dependence} (TOD), which was first introduced and discussed as a potential source of vulnerabilities in~\cite{luu2016making}. 
A smart contract exhibits TOD if the execution order of different contract-invoking transactions can lead to different contract states.
In practice, most smart contracts exhibit some form of TOD, since contracts usually maintain state that can be updated by multiple users. 
TOD coarsely overapproximates the existence of frontrunning vulnerabilities by flagging all contracts that could have conflicting transactions---even if the concurrent execution of those transactions can be avoided when users respect the intended contract usage patterns~\cite{bose2022sailfish}.
To account for this imprecision, 
current works refer to the original notion of TOD as \emph{generalized TOD} (G-TOD) or event ordering (EO) bugs 
and try to apply heuristics to identify exploitable fragments of TOD, e.g., by focusing on concurrency bugs that directly affect the flow of Ether~\cite{bose2022sailfish,zhang2024nyx}. These heuristics are not formally defined but only evaluated empirically using tools 
aimed at identifying vulnerabilities that are immediately exploitable for monetary profit---the most recent tools being Sailfish \cite{bose2022sailfish} and Nyx \cite{zhang2024nyx}.
To detect TOD bugs,~Sailfish identifies read-write hazards (conflicting read/write accesses to the same variables) in smart contracts that indicate G-TOD bugs.
G-TOD bug candidates are then narrowed down to (potential) TOD bugs by applying heuristics to decide whether they can be exploited by a frontrunning attacker to influence a money transfer.
Similarly,~\cite{zhang2024nyx} presents the tool Nyx, which focuses on detecting \emph{exploitable} \mbox{(G-)TOD} bugs. 
A (G-)TOD bug is considered exploitable if reordering the invocations of two contract functions increases the profit of one specific user (the attacker) and decreases that of another (the victim).

Nyx \cite{zhang2024nyx} and Sailfish \cite{bose2022sailfish} implement this static contract analysis to detect arbitrage and sandwiching opportunities before contract deployment.
In this way, they purposefully target monetary frontrunning vulnerabilities to increase analysis precision.
However, we observe that the analysis underlying those tools additionally disregards eventual (monetary) frontrunning attacks.

Nyx and Sailfish both implement a symbolic validation phase, which verifies that an attack trace candidate (indicating a possible (G-)TOD bug) is indeed executable.
Such attack trace candidates consist of two potentially interfering function calls ($\textit{f}_1, \textit{f}_2$), out of which at least one ($\textit{f}_2$) needs to have a direct financial implication (e.g., trigger a money transfer).
The validation step then checks that $\textit{f}_1$ and $\textit{f}_2$ can be executed consecutively (enabling the critical interference). 
However, triggering an eventual attack requires the non-consecutive execution of $\textit{f}_1$ and $\textit{f}_2$, e.g., executions in different blocks or such that are separated by other function invocations.

\begin{figure}[tp]
    \centering
    \small
    \scalebox{0.99}{\begin{tabular}{|l|c|c|c|}%
    \hline%
    \textbf{Auditor}
    &\textbf{\#\DSone}
    &\textbf{\#\DStwo}
    &\textbf{\#\DSthree}\\%
    \hline
    \textit{Trail of Bits}&81   &2 &2\\%
    \hline%
    \textit{ConsenSys Diligence}&79 &11   &11\\%
    \hline%
    \textit{Nethermind}&11  &4  &4\\%
    \hline%
    \textit{Runtime Verification}&20 &0  &0\\%
    \hline%
    \textit{HashEx}&35  &0  &0\\%
    \hline%
    \textit{OpenZeppelin}&65    &6  &6\\%
    \hline%
    \textit{Hacken}&37  &0  &0\\%
    \hline%
    \textit{Quantstamp}&65  &1  &1\\%
    \hline%
    \hline%
    \textit{\textbf{Total}}&393&24&24\\%
    \hline%
    \end{tabular}}
    \captionof{table}{Overview of our datasets comprised of 287 publicly available smart contract audits by 8 auditors. }
    \label{tab:datasets}
\end{figure}

\subsubsection{Evaluating Frontrunning in the Wild.}\label{sec:eval_MEV}
To verify the prevalence of the aforementioned attack types, we \emph{manually} reviewed 287 publicly available professional smart contract audits
based on the auditors
listed by the Ethereum Foundation~\cite{ethereumAuditors}.  %

In particular, in our dataset $\DSone$, we chose the top eight smart contract audit companies with the least number of exploited projects according to the ``rekt.news'' leaderboard~\cite{rektleaderboard}. 
Subsequently, we analyzed all audits that were publicly accessible and filtered out those that were not related to frontrunning vulnerabilities.
Additionally, we curated datasets $\DStwo$ and $\DSthree$ with executable smart contracts including all those audits where complete Solidity source code is available for both the vulnerability and the fix ($ \DSone \supseteq \DStwo \sim \DSthree$), i.e., for each vulnerable audit in $\DStwo$ there is an implemented fix in $\DSthree$, and the fixes were seeking to completely repair the described frontrunning vulnerability.
Based on $\DSone$, there were a total of 393 frontrunning vulnerabilities identified by the auditors in question within the analyzed 287 contracts. 
Of these, 24 vulnerabilities fall in the category of $\DStwo$ and $\DSthree$.
From the vulnerability descriptions in the audits, we collected attack traces and the exploitable function pairs involved in the attack for further analysis. 
Additionally, we tagged those vulnerabilities that are non-monetary or eventual.
Table~\ref{tab:datasets} provides an overview of our three datasets. %

\subheading{Evaluating Static Frontrunning Analysis.}
To validate our observation on static frontrunning analysis, we ran Nyx and Sailfish on each vulnerability in $\DStwo$ using the flattened version of the contract, with all contract dependencies included in a single Solidity file that can be statically compiled.
We checked whether the exploitable function pair reported in the audit is flagged as vulnerable by the respective tool.

\vspace{0.5 em}
\noindent \fbox{
 \parbox{0.945\columnwidth}{
Unfortunately, as already hinted in \cite{zhang2024nyx}, Sailfish's real-world performance is limited, resulting in execution errors on 67\% of the contracts in the dataset.
In the remaining portion, Sailfish detected only two vulnerabilities.
Nyx failed in considerably fewer cases (21\%), but detected only {two} out of all frontrunning vulnerabilities in $\DStwo$.
In particular, Nyx and Sailfish did not detect any eventual-only attacks, and both failed to detect 89\% of non-monetary-only frontrunning attacks.}
}
\vspace{0.5 em}

To cross-validate the evaluation results, we ran both tools on $\DSthree$ to understand their performance on secure contracts. 
Interestingly, of the 24 fixed contracts, most cases in which a tool correctly labeled a fixed contract as ``safe'' were those in which the same tool had already misclassified the corresponding vulnerable contract as secure (four instances for Sailfish and 16 for Nyx). Across both tools, only a single contract was correctly classified at both stages: identified as ``vulnerable'' before the fix and ``safe'' after it, and this was achieved by Sailfish alone.

\subheading{Evaluating Dynamic Frontrunning Analysis.}
Lastly, for all 393 frontrunning vulnerabilities in $\DSone$, we manually applied the MEV definition 
to the attack traces reported in the audits.
We aimed to 
verify whether the adversary could effectively earn monetary rewards by exploiting the vulnerability. 
In case MEV was not applicable to the particular vulnerability, we analyzed whether this was due to \emph{non-monetary}, \emph{eventual}, or \emph{other} possible limitations of the MEV definition.
Our MEV analysis results on dataset $\DSone$ confirm that the limitations of qualifying frontrunning attacks using MEV are not only theoretical but also affect a significant number of real-world contracts. 

\vspace{0.5 em}
\noindent \fbox{
 \parbox{0.945\columnwidth}{
In particular, 
(a) among the 393 vulnerabilities, a total of 217 (55.2\%) could not be captured by MEV, 
(b) 18.6\% (21.9\%) of the vulnerabilities could not be captured by MEV because they are eventual (non-monetary) only, while 14.8\% of the vulnerabilities could not be captured by MEV because they are both eventual and non-monetary, and
(c) no other limitations (besides eventual and non-monetary vulnerabilities) hindered MEV's applicability.}}
\vspace{0.5 em}

A summary of these results (and the corresponding severity rating by the respective auditors) is shown in Figure~\ref{fig:summary-of-results}.

\begin{figure}[t]
    \centering
    \includegraphics[width=0.8\linewidth]{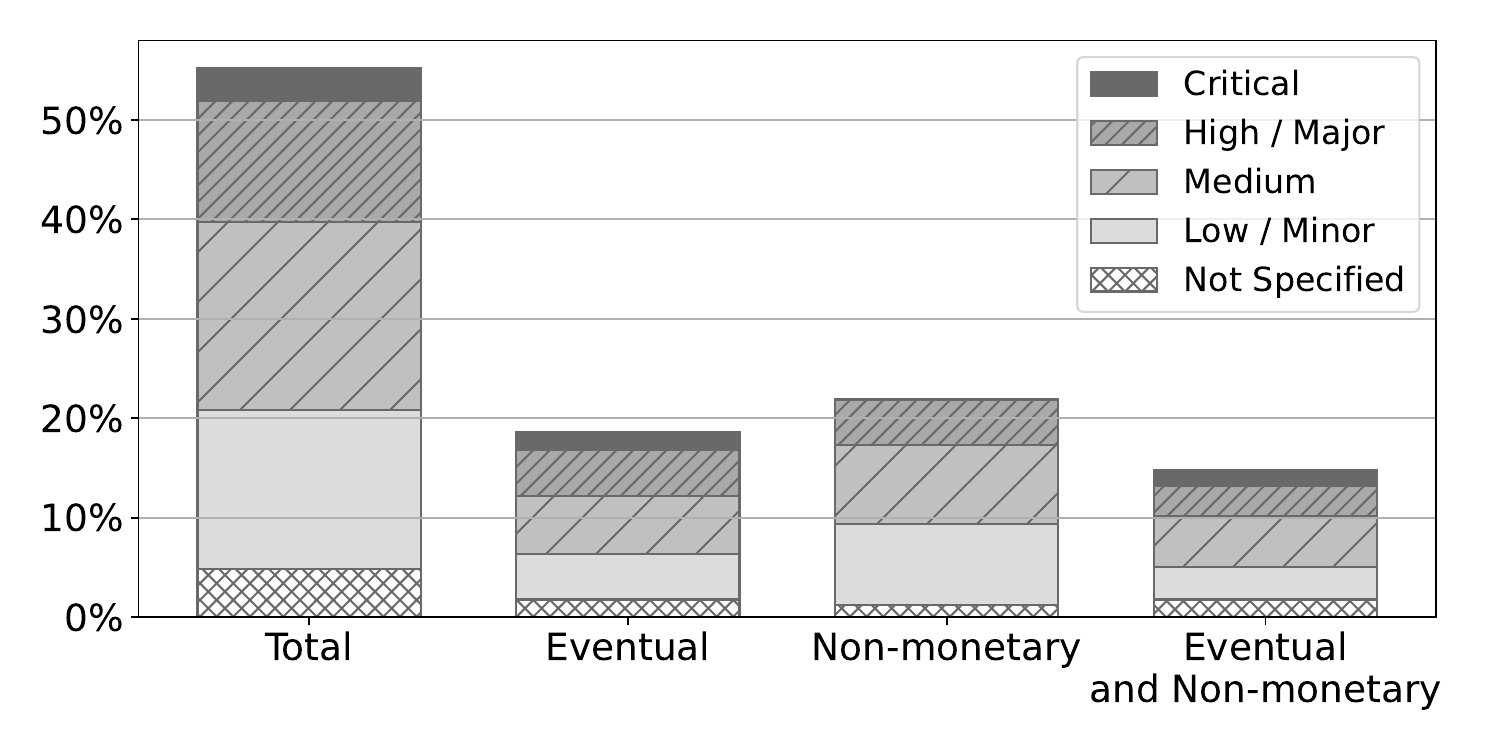}
    \captionof{figure}{Summary of the vulnerabilities not captured by MEV due to their eventual or non-monetary nature (in \DSone). %
    }
    \label{fig:summary-of-results}
\end{figure}

\section{Defining \sysname}\label{sec:formal} \label{sec:deckstacking_resistance}
Informed by our empirical analysis, we make the following observation:
Smart contracts usually play the role of a trusted party that mediates a protocol between multiple mutually distrusting users. Consequently, frontrunning attacks constitute an interference of an attacker with the interaction of a user with this trusted party.
Following this intuition, we will set forth a simulation-based security notion similar to security definitions for cryptographic protocols to characterize that user-contract interactions should be ideal.
We will call this notion \sysname{} (short \sysnameAc) accounting for the fact that attackers with miner-like capabilities can not only harm user-contract interactions by literally \emph{frontrunning} honest user transactions but can also manipulate transaction execution in multiple ways, similar to how a cheating user in a card game may peek into the card deck, drop cards, insert cards, or reorder the entire deck in their favor (so `stacking the deck').

\subsection{Blockchain Model} \label{subsec:blockchain-model}
We base our definition on a formal model of smart contract execution.
Our model combines and extends concepts from~\cite{DBLP:conf/ccs/BartolettiZ18}, which gives a computationally sound model for executing Bitcoin-based smart contracts, 
and~\cite{bartoletti2023theoretical}, which defines a core calculus for Solidity. 
We lift the blockchain execution model from~\cite{DBLP:conf/ccs/BartolettiZ18} to support block-based scheduling and transaction inclusion time guarantees.

\vspace{0.5 em}\noindent \textbf{Modeling Transactions. }
To model the blockchain state, we assume a set of blockchain users $\honusers$ 
and a set of contracts 
where every contract $\contrC$ has a contract state $\contrS$ that maps every variable that occurs in $\contrC$ to its current value.
We model the evolution of blockchain states with a transition system on configurations $\conf$ where a configuration describes the current snapshot of the blockchain state.
Configurations can be advanced with transactions $\txv$
whose execution impacts the states of users and contracts.
To capture the effects of executing a transaction $\txv$ (which may involve multiple asset transfers or the logging of different events), each transition $\conf_0  \trans[\obs]{\txv} \conf_1$ emits a list of observables $\obs$ 
and we call a sequence of (valid) state transitions
\noindent
{\small
$
\conf_0  \trans[\obs_0]{\txv_0} \conf_1 \trans[\obs_1]{\txv_1} \cdots \trans[\obs_{n-1}]{\txv_{n-1}} \conf_n
$}
a \emph{run} (also denoted by $\run$).
Since our transaction framework is modular to the exact smart contract execution semantics, it can be instantiated for different formal smart contract languages such as~\cite{hildenbrandt2018kevm,jiao2020semantic,marmsoler2021denotational}.  
Additional details about the blockchain model and accompanying formalizations for the following technical subsections on the definition of \sysname are given in Appendix~\ref{appendix:model}.

\subsection{Overview \& Key Insights}\label{sec:dsr_definition:overview}

In cryptography, simulation-based security definitions characterize the security of a cryptographic protocol by comparing real-world interactions with this protocol in the presence of a real-world attacker with interactions in an idealized world. 
In such an idealized world, the ideal workings of the protocol are characterized by explicitly describing what information an attacker (then usually called the \emph{simulator}) may learn and in which way they may interfere with the (secure) protocol execution.
To prove the security of a protocol with respect to such a simulation-based security notion, one needs to construct for each real-world adversary a corresponding ideal-world simulator that produces an output in the ideal world that is indistinguishable from the one that would have been produced in the real world. 
In this way, it is guaranteed that all executions possible in the real world can be mapped to some behavior in the ideal world (which is deemed unproblematic by definition).

We adopt this real-ideal paradigm to define \sysnameAc of a contract $\contract$ by contrasting the executions of $\contract$ in a real world where the attacker is assumed to be a malicious block generator (a.k.a.\ miner) with full access to the mempool and full reordering capabilities with executions of $\contract$ in an ideal world where the simulator has restricted knowledge and capabilities to append blocks.
We write $\astrat$ to refer to the \textcolor{MPG-orange}{real-world} block generator (also called \mbox{{\textcolor{MPG-orange}{attacker}} $\astrat$}), and $\sstrat$ for the \textcolor{MPG-light-blue}{ideal-world} block generator (\mbox{{\textcolor{MPG-light-blue}{simulator}} $\sstrat$}).

In our model, a \textcolor{MPG-orange}{real-world} block generator $\astrat$ is a function that maps a 
mempool $\mempool = [tx^\mempool_{1}, \dots, tx^\mempool_{m}]$ of unprocessed transactions
and a run $\run_0$ to a valid sequence $[\txv_0, \txv_1,\dots, \txv_{n-1}]$ of blockchain transactions (a \textit{block}).
Transactions within the block are either chosen from $\mempool$ or crafted by the block generator (for better readability, we denote such transactions in the following with $\txv_{i}^{\astrat}$ or $\txv_{i}^{\sstrat}
$).
For example, when querying a block generator with $\astrat(\run_0, [ tx^{\mempool}_1 ])$, both $[ \txv^{\mempool}_1 ]$ and $[ \txv^{\astrat}, \txv^{\mempool}_1 ]$ could be valid transaction sequences that extend the blockchain execution $\run_0$.
In contrast, an \textcolor{MPG-light-blue}{ideal-world} block generator $\sstrat$ is a function that maps the past blockchain execution $\run_0$ directly to a \emph{block template} without any such access to the mempool. 
The block template here is a transaction sequence with gaps (wildcards) consisting of simulator-generated transactions.
Intuitively, wildcards serve as placeholders for transactions from the mempool. 
For instance, a simulator $\sstrat(\run_0)$ can return the template block $[ \txv^{\sstrat}, \_ ]$ indicating that the first transaction will be its own $\txv^{\sstrat}$ followed by some transaction from the mempool.
After the simulator has crafted a block template including its own transactions, transactions from the mempool are assigned to the template's wildcards, and the block is published.

\subheading{Illustrative Example.}\label{sec:deckstacking_resistance:ideal}
Consider the example of DoS-frontrunning depicted in Figure~\ref{fig:sc:voting-dos}. 
To carry out the attack, 
the attacker observes the call \mbox{\lstinline|proposeCandidate(Id, Candidate)|$^{\honuser}$} by honest user $\honuser$ in the mempool $\mempool_1$ and outruns it with their own transaction, copying the single-use candidate identifier \lstinline|Id|.

{    
    \centering
    \includegraphics[width=0.6\linewidth]{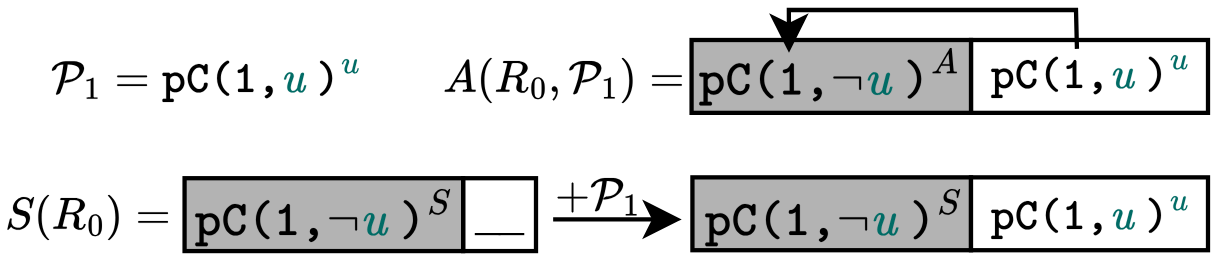}
    
}%
The attacker invalidates the user's proposal by scheduling their own \lstinline|proposeCandidate| transaction with the same identifier (\mbox{here \lstinline|1|}) and a different candidate $\neg \honuser$ at the beginning of their block, causing it to be executed before the honest user's transaction. 
As a result, the honest user's transaction will fail instead of resulting in a valid candidate proposal. 
The same behavior can be simulated in the ideal world by a simulator that proposes a candidate with \mbox{identifier \lstinline|1|} (regardless of the mempool) within its block template.

However, our \sysnameAc definition will require that for each real-world attacker, there is a single ideal-world simulator that can simulate an attacker for \emph{all possible user behaviors}.
This means, in particular, that a real-world attacker may not gain any advantage from observing user actions in the mempool since an ideal-world simulator could not mimic this \emph{adaptability} of the real-world attacker.
Consider, for instance, that the honest user chooses a different identifier \lstinline|Id| for their \lstinline|proposeCandidate| call:

{    
    \centering
    \includegraphics[width=0.66\linewidth]{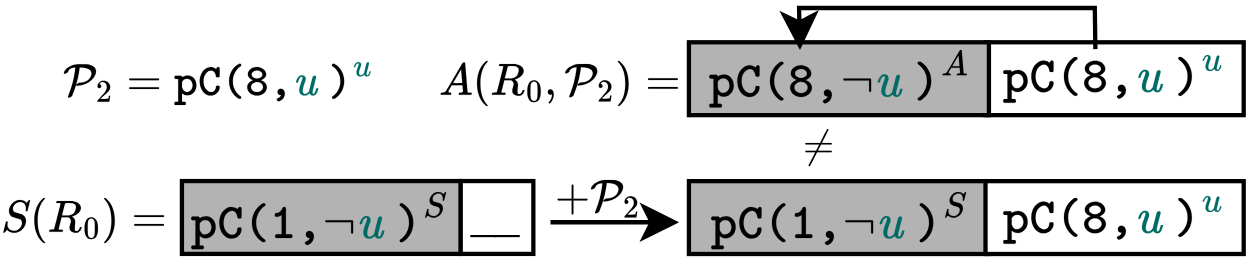}

}

In this case, the input to the simulator (consisting only of the past blockchain execution) would still be unchanged, making them output the same block template as before. 
In contrast, the real-world attacker would exploit the transaction observed in the mempool $\mempool_2$ and copy the new identifier.
In fact, we cannot find a single simulator that simulates the real-world attacker for all valid user behaviors (proposing different identifiers) for the same initial blockchain \mbox{run $\run_0$}.
Therefore, the smart contract from Figure~\ref{fig:sc:voting-dos}
is not classified to satisfy \sysnameAc. 
Intuitively, we define \sysname~as, 
\textit{for every malicious real-world block generator (\mbox{{\textcolor{MPG-orange}{attacker}} $\astrat$}), there exists a non-adaptive ideal-world block generator (\mbox{{\textcolor{MPG-light-blue}{simulator}} $\sstrat$}) such that every valid smart contract execution of the {\textcolor{MPG-dark-green}{honest user}} in the real world, can be simulated in the ideal world such that the executions agree 
on their observable behavior.
}

\subsection{Real-World Blockchain Execution}
The \sysnameAc of a contract $\contract$ does not only rely on the contract code but also on the behavior of honest users $\honusers$ interacting \mbox{with $\contract$}. 
E.g., in the illustrative example, honest users that refrain from calling the \lstinline|proposeCandidate| function, are clearly not subject to the described attack.
One can think of \sysnameAc as the guarantee given to users $\honusers$ on the behavior of $\contract$ when they interact \mbox{with $\contract$}.
For example, many frontrunning mitigation techniques rely on honest users complying with restrictions, such as submitting certain transactions only in well-defined time windows or not leaking authentication secrets to the attacker. 
Thus
we define \sysnameAc with respect to a set $\sustrat{\contract}$ of honest user strategies that describe the legitimate behavior of the users $\honusers$  when interacting with $\contract$.

\vspace{0.5 em}\noindent \textbf{Modeling Cryptographic Interactions. } \label{model:crypto}
To simplify presentation and reasoning, we will in the following treat the security of these building blocks in a \emph{symbolic} fashion, meaning that we will assume ideal security (e.g., assuming that it is impossible to find the preimage of a hash function without having learned it before).
This is a common approach to abstract the security of cryptographic primitives when modeling and analyzing cryptographic protocols (see~\cite{dolev1983security,blanchet2016modeling,meier2013tamarin}). 
    We incorporate this notion of cryptography by assuming a dedicated set of secrets $\secrets_{\honusers}$ of honest users and a concrete smart contract semantics that supports cryptographic operations with respect to an equational theory $\eqth$ on symbolic terms.
    Attackers may only learn and use an honest user secret after it got revealed in the mempool or the previous blockchain run.
As a consequence, transactions, as well as observables emitted in runs, may contain symbolic terms.

\vspace{0.5 em}\noindent \textbf{Modeling Mempools. } \label{model:mempool}
We define a mempool $\mempool$ as a list \linebreak \mbox{$[\txv_0, \dots, \txv_i]$} of valid transactions by honest blockchain users $\honusers$.
The choices of all honest users $\honusers$ for a contract $\contract$ are modeled by a set of symbolic user strategies $\ustrat{\contract}{\honusers}{\secrets}$.
Each strategy $\cstrat \in \ustrat{\contract}{\honusers}{\secrets}$ is a function mapping the current blockchain run $\run$ to the next mempool $\mempool_{\cstrat}$ expressing the combined strategy of all honest users while only using their shared set of symbolic secrets $\secrets$.

We define a number of correctness properties on a mempool $\mempool_{\cstrat}$ and user strategy $\cstrat$ to capture realistic blockchain behavior:
A mempool can only contain valid transactions sent by honest users, i.e., transactions must be executable and may be only composed of knowledge available to such honest users.
Additionally, the content of mempools must be persistent until inclusion in the run and, correspondingly, the honest user strategy needs to consistently publish the same transactions as long as they are valid (reflecting that a transaction once submitted to the mempool stays there).

\vspace{0.5 em}\noindent\textbf{Real-World Scheduling Strategies. } \label{model:scheduling}
We model the mining of new blockchain blocks (cf. block generators in Section \ref{sec:dsr_definition:overview}) with \emph{symbolic scheduling strategies} that are based on the current blockchain run and mempool and return a valid next block.
Formally, a symbolic scheduling strategy is a function $\astrat(\run, \sseq)_{\secrets}$, parameterized with a set of symbolic secrets $\secrets$ accessible by the function, that produces a valid transaction sequence $\bblock = [\txv_0, \dots, \txv_{n-1}, \txB] $ where each transaction $\txv_i \in \bblock$ 
when sent by an honest user is unique and stems from the mempool,
and otherwise needs to be constructible using attacker knowledge. 
Further, it is guaranteed that each pending transaction $\txv \in \mempool$ from the mempool is included in block $\bblock$ if the time of its proposal by an honest user dates back at least $k$ blocks.

Given a scheduling strategy $\astrat$ and a user strategy $\cstrat$, we can now characterize the blockchain runs $\run$ that result from the interactions of $\astrat$ and $\cstrat$ (written $\combine{\astrat}{\cstrat} \semantics \run$).
In this case, we say that \emph{$\run$ conforms to $\combine{\astrat}{\cstrat}$ in the real world}.
First, it holds that $\combine{\astrat}{\cstrat} \semantics \initconf$ where $\initconf$ here denotes the run starting in the initial blockchain configuration where no state transitions have been performed. 
Further, runs with $\combine{\astrat}{\cstrat} \semantics \run$ can be derived given that $\combine{\astrat}{\cstrat}  \semantics \run'$ holds, by appending a new block $\astrat(\run', \cstrat(\run')) = \bblock$ with $\run = \run' \trans[\textit{xs}]{\bblock}^* \conf_{\run}$.
So, a new block $\bblock$ with observables $\textit{xs}$ can be added to the blockchain run $\run$ provided that scheduler $\astrat$ selects this block based on its knowledge of $\run$ and the mempool created by $\cstrat$.
Note that, for the sake of readability, we only explicitly annotate the symbolic secret sets of strategies when relevant.

\subsection{Ideal-World Blockchain Execution} \label{model:simulator}
To clearly specify the limits of the ideal world, we define ideal-world scheduling to be a restricted version of real-world block scheduling:
Crucially, an ideal-world scheduling should be \emph{non-adaptive} to the mempool, meaning that its scheduling strategy is not influenced by the content of the mempool. 

In particular, we use the intuition of template blocks to restrict the set of symbolic scheduling strategies $\sstrat$ to be translucent towards certain mempool changes and aim to characterize the strongest ideal-world model with a non-adaptive scheduler.
Note that comparing against an ideal-world model that is too strong (e.g., one that does not allow for any transaction reordering) yields a security notion that is too restrictive (e.g., forbidding any concurrent interactions with a contract).
Similarly, an ideal-world model that is too weak (equipping the simulator with too much knowledge and capabilities) renders smart contracts that are intuitively insecure resistant to \deckstacking. 
Formally, ideal-world scheduling strategies are symbolic scheduling strategies that fulfill a stability requirement under substitution:
We define a set of substitution functions operating on a mempool $\mempool$ or blockchain block $\bblock$ that is defining a sequence of honest user transactions in $\mempool$/$\bblock$ to be substituted for other valid transactions ($\{ \txv / \txv' \}$) or removing them.
Every substitution function $\rho$ is defined such that its application $\rho(\mempool)$ on a mempool $\mempool$ returns a valid mempool $\mempool'$.
Such substitutions allow us to define the following stability criterion $\isSim{\sstrat}$ that rules out simulators that create blocks based on \textit{length}, \textit{content} and \textit{order} of the mempool $\mempool$: $\forall \rho. \sstrat(\run, \rho(\sseq)) = \rho(\sstrat(\run, \sseq))$

This effectively eliminates any adaptive \deckstackingAc capabilities of the simulating scheduler $\sstrat$ in the ideal-world.
In fact, this stability requirement on substitutions is the formalization of the template block intuition in Section~\ref{sec:dsr_definition:overview}.
Intuitively, wildcards in a template are spots that should be filled with arbitrary transactions from the mempool. 
If the mempool changes (e.g., a transaction $\txv_1$ being replaced with a transaction $\txv_2$), then this change should be reflected when filling the template (by transaction $\txv_2$ being filled into the same wildcard slot that $\txv_1$ took before).
Using this intuition, we define valid blockchain executions in the \textcolor{MPG-light-blue}{ideal-world} by accepting any simulator $\sstrat$ that is a non-adaptive real-world attacker
w.r.t. the stability requirement on substitutions.
In this case, we write $\combine{\sstrat}{\cstrat} \semantics \run$ and say that $\run$ \emph{conforms to $\combine{\sstrat}{\cstrat}$ in the ideal world} if $\sstrat$ is a scheduling strategy that satisfies $\isSim{\sstrat}$. 
Recall that a complete formalization of the stability requirement on substitutions and previous subsections is available in Appendix \ref{appendix:model}.

\subsection{Definition}

We next define when a contract $\contract$ is \deckstackingAc Resistant with respect to a set of symbolic user strategies $\ustrat{C}{\honusers}{\secrets_{\honusers}}$. 
Each strategy describes one of the legitimate behaviors of the users in $\honusers$ when interacting with $\contract$ and owning secrets $\secrets_{\honusers}$. 
To capture different forms of observables in the most general fashion, we parametrize \sysnameAc with a similarity relation $\sim$ that describes when a run $\arun$ and a run $\srun$ are considered to show the same observable behavior. 
More formally, this leads us to the following characterization:

\vspace{1 em}
\noindent \fbox{
 \parbox{0.935\columnwidth}{
\textit{\textbf{\sysname:} A contract $\contract$ satisfies \sysname w.r.t. a set of honest user strategies $\ustrat{\contract}{\honusers}{\secrets_{\honusers}}$ if }
{\small
\begin{align*}
    \forall \astrat_{\secrets}\mdot~ \exists &\sstrat_{\secrets} \textit{~with~} \isSim{\sstrat_{\secrets}} \mdot~  
    \forall \cstrat \in \ustrat{\contract}{\honusers}{\secrets_{\honusers}}\mdot~  %
    \\ 
     &\forall\arun\mdot \combine{\astrat}{\cstrat} \semantics  \arun 
     \Longrightarrow \exists\srun\mdot  
     \combine{\sstrat}{\cstrat} %
            \semantics \srun 
     \wedge \arun \sim \srun
\end{align*}
}%
\textit{where the scheduling strategies ($\astrat$ and $\sstrat$) and the user strategies do not share any secrets ($\secrets \cap \secrets_{\honusers} = \emptyset$). }
} }
\vspace{1 em}

The simulation-based definition enables us to
distinguish aimless, disadvantageous behavior (caused by asynchronous blockchain execution) from targeted malicious behavior.
For instance, considering the example from Figure~\ref{fig:sc:voting-dos}, $\sustrat{\contract}$ would contain all strategies $\cstrat_{\texttt{i}}$ invoking the \lstinline|proposeCandidate| function with identifier \lstinline|i|, so in particular $\cstrat_{\texttt{1}}$ invoking \lstinline|proposeCandidate| with identifier \lstinline|1| and $\cstrat_{\texttt{8}}$ invoking \lstinline|proposeCandidate| with identifier \lstinline|8|. 
This exemplifies the chosen order of quantification such that a \emph{single} simulator must be able to simulate every concrete user behavior and to restrain the adaptiveness of the simulator $\sstrat$ to such untargeted frontrunning.

\subheading{Discussion.}
\sysnameAc is purposefully defined such that it can be instantiated with any smart contract semantics, equational theory $\eqth$ and similarity relation $\sim$.
Indeed, the concrete choice of the similarity relation determines whether a contract is considered resistant to \deckstacking or not.
This reflects that, depending on the concrete smart contract, certain deviations from the ideal contract behavior caused by \deckstackingAc may be deemed acceptable (e.g., if such deviations can at most benefit honest users). 
Identifying security-critical behavior (for which no deviations with respect to the ideal behavior should occur) is part of the contract development process, where developers usually explicitly log critical events in smart contracts. 
This gives rise to a generic similarity relation requiring that the same sequences of critical events 
(e.g., modifications of high-integrity contract state) are exhibited in the real- and ideal-world execution.
\sysnameAc, however, is flexible enough to incorporate more fine-grained use-case-specific similarity notions, e.g., assigning and comparing (quantitative) utilities of different contract executions.

This flexibility, conversely, implies that there is not one universal instantiation of the similarity relation $\sim$ that is adequate for all scenarios.
E.g., consider a smart contract that allows a user to donate a fraction of their own assets under some miner-controlled condition (e.g., whenever the block number is even at the time of invoking the donation function). With an instantiation of $\sim$ that considers such runs similar, which provide the same payouts to all users, this contract would be classified as violating \deckstackingAc{}. This is because a deckstacking attacker could decide upon the execution/blocking of donations by placing calls to the donation function in even or odd blocks, respectively.
However, depending on the context, one may argue that the ability to deny donations does not constitute harmful behavior (e.g., as such interference does not induce direct money losses). To reflect this interpretation, one could simply adjust the $\sim$ relation to ignore financial transfers triggered by invocations of the donation function. 
This example illustrates how the similarity relation $\sim$ serves to express the degree of interference of an adaptive, \deckstacking{} attacker that is still deemed acceptable.

\newcommand{\nwcond}{\condVar_\text{NW}}

\section{Interaction Condition Generation}\label{section:conditions} %
A key insight of the previously presented definition is that \sysnameAc is a property that does not concern a smart contract alone but the interplay between a contract and an honest user. This poses the question of how a user can decide whether their intended contract interactions can be subject to \deckstacking attacks. To this end, in this section, we propose a procedure to generate provably secure \emph{interaction conditions} from the code of a smart contract $\contract$. More specifically, such interaction conditions formulate requirements on the user strategy and guarantee that $\contract$ together with any user strategy that meets these requirements satisfies \sysnameAc{}. 

The key property of a user strategy to resist \deckstacking attacks is that it only submits transactions $\htxv$ to the blockchain that have a predictable effect (meaning that a \deckstacking attacker cannot interfere with them).
Since transactions can be only guaranteed to be included within a fixed number $k$ of blocks, user strategies, to achieve predictable transaction behavior, usually proceed in rounds of length $k$: whenever a user $\honuser$ submits a transaction $\htxv$, they want to ensure that within the next $k$ blocks, 
independently of how an adaptive attacker includes $\htxv$, the execution result will be the same. 

\begin{figure}[t]
  \raggedleft
  \footnotesize
  \begin{minipage}{0.925\linewidth}
  \begin{lstlisting}[basicstyle=\scriptsize\ttfamily, numberstyle=\tiny]
contract TimelockedFeeMinted { 
  function scheduleFeeChange(uint _newFee) onlyOwner {  
    bytes32 id = sha256("newFee", _newFee);
    timestamps[id] = block.number + k;    }
  
  function executeFeeChange(uint _newFee) onlyOwner {
    bytes32 id = sha256("newFee", _newFee);
    require(timestamps[id] > 1);
    require(timestamps[id] <= block.number);
    fee = _newFee / 100;
    timestamps[id] = 1; }
  
  function mint() {
    uint minted = msg.value * (1 - fee);  // ...
    emit Minted(msg.sender, minted); 
  } // ...
  \end{lstlisting}
  \end{minipage}
      \caption{Simplified OpenZeppelin TimelockController~\cite{openzeppelin-timelock}.}
      \label{fig:sc:timlock-fix}
  \end{figure}

As an example, consider the smart contract depicted in Figure~\ref{fig:sc:timlock-fix}.
This contract allows buying (aka minting) contract-specific tokens for a \lstinline|fee|, which is determined by the contract owner, using the \lstinline|mint| function. 
Such a contract can be easily prone to \deckstacking attacks since a malicious contract owner, when observing a \lstinline|mint| transaction, can preempt this transaction with a corresponding transaction that increases the fee, causing the honest user to buy less tokens than expected. 
The presented contract circumvents this problem by making a fee change subject to a staged process, where a contract owner first needs to invoke the \lstinline|scheduleFeeChange| function to announce a fee change. 
The fee change can only be executed (using the \lstinline|executeFeeChange| function) after time $k$. 
Consequently, an honest user can ensure not to be subject to unexpected fee changes by only invoking the \lstinline|mint| function in rounds where initially no fee change has been scheduled. 
 More precisely, a condition for securely invoking the \lstinline|mint| function would be that $\forall i.~\cstate{\conf}{\contract}(\tsMap[i]) \leq 1$. %
In this case, the blockchain inclusion time ensures that by the end of the round (so after $k$ blocks), a \lstinline|mint| transaction submitted at the beginning of the round will be included.
The contract logic would ensure that the execution of \lstinline|executeFeeChange| could only impact the \lstinline|fee| in the following round, as $\cstate{\conf}{\contract}(\tsMap[i]) > \text{\lstinline|block.number|}$ would hold for the remainder of the current round, even if the attacker executed \lstinline|scheduleFeeChange|.

\subsection{Overview \& Key Insights}
Our algorithm makes this reasoning explicit by generating a pair of conditions 
$(\invpre, \inv)$ for the \lstinline|Minted| event, which ensure (a) that if the condition $\invpre$ holds at the beginning of a round, then also 
the condition $\inv$ holds throughout the whole round, and 
(b) $\inv$ is sufficiently strong to ensure that the attacker cannot interfere with the execution of \lstinline|Minted| during the round.
For the previous example, $\forall i.~\cstate{\conf}{\contract}(\tsMap[i]) \leq 1$ would be such a candidate for $\inv$, since it satisfies requirement (b) by excluding that executions of \lstinline|executeFeeChange| may change the value of \lstinline|fee|, which is the only (global) variable that may influence \lstinline|Minted|.
However, to satisfy requirement (a), the challenge also lies in
(i) finding a condition $\invpre$ that can be checked at the beginning of the round so that the user can decide based on $\invpre$ whether it is safe to execute the \lstinline|mint| function;
(ii) ensuring that $\inv$ is indeed upheld throughout the whole round, even if the attacker executes arbitrary other transactions; 
(iii) finding a formulation of $(\invpre, \inv)$ that ideally characterizes all possible situations where it is safe for the user to execute the \lstinline|mint| function.
For example, the condition $\forall i.~\cstate{\conf}{\contract}(\tsMap[i]) \leq 1$ would immediately violate requirement (ii) since after the attacker executes \lstinline|scheduleFeeChange|, it is no longer valid.
Further, the condition is not sufficiently general (violating requirement (iii)), since it does not capture that executing \lstinline|mint| would also always be safe if it could be excluded that the attacker is the owner of the contract.

To solve these challenges, our algorithm proceeds by first generating candidate conditions $(\invpre, \inv)$, which are then stepwise refined to satisfy requirements (a) and (b). We first illustrate the algorithm with the example: 

\vspace{0.5 em}\noindent \textbf{Synthesis Example. }
The algorithm starts from an event, which shall not be influenced by the attacker (here the \lstinline|Minted| event) 
and uses the condition for reaching this event (also called the \emph{path condition}) as the first candidate condition $\inv^0$. 
In our example, $\inv^0 = \top$ since the \lstinline|Minted| event can be reached unconditionally.
We next check if $\invpre^0$ set to $\inv^0$ is a reasonable candidate for a round precondition, by checking 
whether $\invpre^0$ holding at the beginning of the round (so for some value $b$ of \lstinline|block.number| such that $b \bmod{k} = 0$) implies that $\inv^0$ also holds during the round (so for values $b + i$ of \lstinline|block.number| such that $i < k$).
In this case, we say that $(\invpre^0, \inv^0)$ is a \emph{round invariant}. 
Indeed, $(\top, \top)$ trivially is a round invariant since it is fully independent of \lstinline|block.number|. 
In the next step, we refine $(\invpre^0, \inv^0)$ towards satisfying requirement (b): 
To this end, we determine the part of the global state (global contract variables and blockchain variables like \lstinline|block.number|) that may directly impact the arguments of the \lstinline|Minted| event, i.e., \lstinline|Minted|'s \emph{data dependencies}. 
In the example, the variable \lstinline|fee| is the only data dependency of \lstinline|Minted|.
The algorithm then strengthens $(\invpre^0, \inv^0)$ to enforce that \lstinline|fee| cannot be written by the attacker by 
adding the following condition to both $\invpre^0$ and $\inv^0$:

{%
\scriptsize
\begin{align*}
 \nwcond = \forall &\text{\lstinline|_newFee|}^{\nuserindexvar},~ \sendervari{\nuserindexvar}.~\sendervari{\nuserindexvar} \neq \honuser\rightarrow  
  (\sendervari{\nuserindexvar} \neq \ownervar) \\
    & \lor (\tsMap[\shaFunc(\text{"newFee"}, \text{\lstinline|_newFee|}^{\nuserindexvar})] \leq 1) \\
    & \lor  (\tsMap[\shaFunc(\text{"newFee"}, \text{\lstinline|_newFee|}^{\nuserindexvar})] > \blockvar)
\end{align*}
}%
The condition $\nwcond$ is derived from the function \lstinline|executeFeeChange| by forming the \emph{no-write condition}, so the disjunction of the path conditions of all execution paths that leave \lstinline|fee| unchanged (also denoted by $\nowritecondi{\text{\lstinline|executeFeeChange|}}{\text{\lstinline|fee|}}{}$).
$\nwcond$ further adds the precondition $\sendervari{\nuserindexvar} \neq \honuser$, which requires that the sender (indicated by the variable $\sendervari{\nuserindexvar}$) of the transaction invoking \lstinline|executeFeeChange| shall be different from the honest user (represented by variable $\userindexvar$).
Note that we use the superscript ${}^{\nuserindexvar}$ to index variables that are local to the execution of an attacker transaction (such as the function argument $\text{\lstinline|_newFee|}^{\nuserindexvar}$ and the transaction sender $\sendervari{\nuserindexvar}$).
Intuitively, $\nwcond$ states that if \lstinline|executeFeeChange| is invoked by a user different from $\userindexvar$, its execution will never write \lstinline|fee|.

Thus, the refinement leaves us with the new candidate condition $(\invpre^1, \inv^1) = (\nwcond, \nwcond)$.
Note, however, that $(\invpre^1, \inv^1)$ does not constitute a round invariant, since incrementing the block number within the round can cause the third disjunct {$\scriptsize \tsMap[\shaFunc(\dots)] > \blockvar$} to become false.
To correct this, we apply a generic strengthening to $\invpre^1$, replacing occurrences of $\blockvar$ with $\blockvar + k -1$:

{%
\scriptsize
\begin{align*}
 \nwcond' = \forall &\text{\lstinline|_newFee|}^{\nuserindexvar},~ \sendervari{\nuserindexvar}.~\sendervari{\nuserindexvar} \neq \honuser\rightarrow  
  (\sendervari{\nuserindexvar} \neq \ownervar) \\
    & \lor (\tsMap[\shaFunc(\text{"newFee"}, \text{\lstinline|_newFee|}^{\nuserindexvar})] \leq 1) \\
    & \lor  (\tsMap[\shaFunc(\text{"newFee"}, \text{\lstinline|_newFee|}^{\nuserindexvar})] > \blockvar + k - 1)
\end{align*}
}%

This strengthening ensures that if {$\scriptsize\tsMap[\shaFunc(\dots)] >$\linebreak$\blockvar + k -1$} holds at the beginning of the round, then also { $\scriptsize\tsMap[\shaFunc(\dots)] > \blockvar$} (and thus $\inv^1$) holds throughout the whole round.
After verifying that the strengthening establishes a round invariant, we can use $(\invpre^2, \inv^2) = (\nwcond', \nwcond)$ as new round invariant.

Finally, we check towards ensuring requirement (ii) whether $\nwcond'$ is indeed upheld when executing other contract functions by users different from $\honuser$, so whether $\nwcond'$ is an \emph{invariant} w.r.t. these functions. 
This check holds because $\nwcond'$ could only be invalidated by writes to the $\tsMap$ array, which may only occur in \lstinline|scheduleFeeChange| and \lstinline|executeFeeChange|. 
For \mbox{\lstinline|executeFeeChange|,} the write to $\tsMap$ is already excluded due to the \lstinline|require| statements given that $\nwcond'$ holds. In contrast, \lstinline|scheduleFeeChange| may write $\tsMap$ unconditionally. However, $\nwcond'$ ensures that even after updating any $\tsMap[i]$ to $\blockvar + k$, the condition $\tsMap[\shaFunc(\dots)] > \blockvar + k -1$ will hold.

At this point, the algorithm terminates with the round invariant $(\invpre^2, \inv^2)$, which cannot be invalidated by any attacker transaction in the same round (satisfying requirement (a)).
Note that $(\invpre^2, \inv^2)$ is constructed such that $\inv^2$ 
enforces the path condition of \lstinline|Minted| (so as long as $\inv^2$ holds, \lstinline|Minted| gets executed)
and such that it prevents concurrent writes to data dependencies of \lstinline|Minted|.
For the given example, this ensures that a concurrent attacker transaction cannot change the way \lstinline|Minted| is triggered (satisfying requirement (b)).

\begin{figure}[t]
\centering
\small
\begin{tabular}{@{}l@{~} p{6.2cm}@{}}
\toprule
\textbf{Notation} & \textbf{Description} \\
\midrule
$\frWriteCond{\procF}{\seExp}{\seCvar}{\roundPred}$ & Procedure: a round invariant $(\invpre, \inv)$ ensuring that $\procF$ deterministically assigns variable $\seCvar$ to expression $\seExp$ \\
$\cFuncs{\contract}$ & Set of functions of contract $\contract$ \\
$\gvars{\contract}$ & Global variables of contract $\contract$ \\
$\lvars{\procF}{\userindexvar}$ & Local variables of function $\procF$ (e.g., arguments) for user $\userindexvar$ \\
$\setforallshort{\lvars{\procG}{\nuserindexvar}}$ & Universal quantification over all local variables $\lvars{\procF}{\nuserindexvar}$\\
\midrule
\multicolumn{2}{c}{\textbf{Symbolic Execution Operations}} \\
$\pathcondi{\procF}{\seExp}{\seCvar}{\userindexvar}$ & Path condition where executing $\procF$ assigns $\seCvar$ to $\seExp$ \\
$\ddepsi{\procF}{\seCvar}{\userindexvar}$ & Data dependencies of variable $\seCvar$ being written in $\procF$ \\
$\nowritecondi{\procG}{y}{\userindexvar}$ & Condition ensuring that executing $\procG$ does not write~$y$ \\
$\execinvi{\procG}{\psi}{\invpre}{\userindexvar}$ & Checks whether $\invpre$ is an invariant under executions of function $\procG$ that initially satisfy predicate $\psi$ \\
\midrule
\multicolumn{2}{c}{\textbf{Round Invariant Operations}} \\
$\getRoundInv(\condVar)$ & Constructs a round invariant $(\invpre, \condVar)$ from $\condVar$ by checking $\condVar$ and possibly strengthening $\condVar$ to $\invpre$ \\
$\mergeRoundInv(r_1, r_2)$ & Conjoins components of two round invariants \\
\midrule
$\dSet$ & Initial set of global data dependencies \\
$\nowriteset = \{ (y, \procG),~\dots~$ & Set of pairs where writes to $y$ by $\procG$ have been excluded \\
$\condvars{\cdot}$ & Returns the set of variables occurring in a condition \\
\bottomrule
\end{tabular}
\captionof{table}{Legend for Section~\ref{section:conditions} and Algorithm~\ref{algo:fr}}
\label{tab:algo-legend}
\end{figure}

\begin{algorithm}[t]
    \caption{Simplified synthesis of condition $\frWriteCond{\procF}{\seExp}{\seCvar}{\roundPred}$}
    \begin{algorithmic}[1]
\REQUIRE Contract $\contract$, user $\honuser$, function $\procF \in \cFuncs{\contract}$, symbolic execution relation 
$\symExReli{\procG}{}$ for $\procG \in \cFuncs{\contract}$, 
variable $\eventvar \in \gvars{\contract}$, 
 and expression $\seExp$ in $\procF$ %
\STATE $\invaccumulator \gets \getRoundInv(\pathcondi{\procF}{\seExp}{\seCvar}{\userindexvar})$ \label{line:init}
\STATE $\nowriteset \gets \emptyset$; $\dSet \gets \ddepsi{\procF}{\seCvar}{\userindexvar} 
\setminus \lvars{\procF}{\userindexvar}$ \label{line:d-init} %
\FORALL{$y \in \dSet$, $\procG \in \cFuncs{\contract} 
        $} \label{line:ddeps-loop-begin}
            \STATE $\invaccumulator \gets 
            \mergeRoundInv(
                \invaccumulator,$ \\
                \hspace{\algorithmicindent}\hspace{\algorithmicindent}
                $\getRoundInv \left(
                    \setforallshort{\lvars{\procG}{\nuserindexvar}} \sendervari{\nuserindexvar} \neq \honuser \rightarrow \nowritecondi{\procG}{y}{\nuserindexvar}
                \right)
            )$
            \STATE $\nowriteset \gets \nowriteset \cup \{(y, \procG)\}$ \label{line:ddeps-loop-end}
    \ENDFOR
\STATE $\doneflag \gets \false$
    \WHILE{$\condvars{\invaccumulator} \times \cFuncs{\contract} \not \subseteq \nowriteset \land \doneflag = \false$} \label{line:main-loop-begin}
        \STATE $\doneflag \gets \true$
        \FORALL{$\procG \in \cFuncs{\contract}$} \label{line:g-loop-begin}
            \IF{$\execinvi{\procG}{~\predIsNoSender{\honuser}}{~\invpre}{\nuserindexvar}$} \label{line:main:algo-fr:inv-check}
                \STATE \textbf{continue}
            \ENDIF
            \STATE $\doneflag \gets \false$
            \STATE $y \stackrel{\text{choose}}{\gets}\condvars{\invaccumulator} \setminus \{ z ~|~ (z, \procG) \in \nowriteset \}$ \label{line:choose-y}
            \STATE $\invaccumulator \gets 
            \mergeRoundInv (
                    \invaccumulator,$\\
                    \hspace{\algorithmicindent}\hspace{\algorithmicindent}
                    $\getRoundInv\left(
                        \setforallshort{\lvars{\procG}{\nuserindexvar}} \sendervari{\nuserindexvar} \neq \honuser \rightarrow \nowritecondi{\procG}{y}{\nuserindexvar}
                    \right)
            )$ \label{line:merge-no-write-main-loop}
            \STATE $\nowriteset \gets \nowriteset \cup \{(y, \procG)\}$ \label{line:main-loop-end}
        \ENDFOR
    \ENDWHILE
    \RETURN $\invaccumulator$
    \end{algorithmic}
    \label{algo:fr}
\end{algorithm}

\subsection{Generation Algorithm}
Technically, we build the synthesis generation algorithm illustrated in the previous subsection upon a symbolic execution engine for smart contracts. 
The symbolic execution engine provides a logical description for the execution logic of a contract. 
More formally, for each contract function $\procG$, it describes a relation of tuples of the form $\seRes{\seSub}{\seCond}$ which indicate the state updates $\seSub$ that $\procG$ will conduct along the execution path enabled by path condition $\seCond$.
To this end, state updates map contract variables to symbolic expressions $\seExp$ (over global state variables).

Our algorithm (described in Algorithm~\ref{algo:fr}) takes as arguments the symbolic execution relations $\symExReli{\procG}{}$ for the functions $\procG$ of a contract $\contract$, a contract function $\procF$, a user variable $\userindexvar$, a contract variable $\seCvar$ and an expression $\seExp$
and generates conditions $\frWriteCond{\procF}{\seExp}{\seCvar}{\roundPred} = (\invpre, \inv)$. 
The generated condition $\frWriteCond{\procF}{\seExp}{\seCvar}{\roundPred}$ will then ensure that if the user (indicated by variable $\userindexvar$) executes a transaction $\htxv$ that invokes function $\procF$ at the beginning of a round where $\invpre$ holds, then executing $\htxv$ will deterministically assign the variable $\seCvar$ according to the symbolic expression $\seExp$.
The expression $\seExp$ here represents the concrete execution path that shall be taken when executing $\procF$.
In the example generation, we showed 
$$\frWriteCond{\text{\lstinline|mint|}}{(\text{\lstinline|msg.sender|}, \text{\lstinline|msg.value|} * (1-\text{\lstinline|fee|}))}{\text{\lstinline|Minted|}}{\roundPred}$$
i.e., the conditions for deterministically assigning the \lstinline|Minted| event according to the expression $(\text{\lstinline|msg.sender|}, \text{\lstinline|msg.value|} * (1-\text{\lstinline|fee|}))$, when executing \lstinline|mint|.
The expression $(\text{\lstinline|msg.sender|}, \text{\lstinline|msg.value|} * (1-\text{\lstinline|fee|}))$ is obtained by inlining the local assignments in the \lstinline|mint| function.

The symbolic execution relations allow us to formally define the notions of \emph{path conditions}, \emph{data dependencies}, \emph{no-write conditions}, \emph{invariants}, and \emph{round invariants} used in the previous example.
We additionally assume two functions $\getRoundInv$ and $\mergeRoundInv$ that abstract the generation of round invariants: 
The function $\getRoundInv$, provided a condition $\condVar$, generates a round invariant $(\invpre, \inv)$ such that $\condVar \equiv\inv$. 
In particular, $\getRoundInv$ may attempt to strengthen the invariant as described in the example, or default to $(\bot, \condVar)$ in case  the strengthening fails. The latter case would indicate that there is no safe interaction condition. 
The function $\mergeRoundInv$, provided two round invariants $(\invpre^1, \inv^1)$ and $(\invpre^2, \inv^2)$, creates a new round invariant 
$(\invpre, \inv)$ such that $\invpre \equiv \invpre^1 \land \invpre^2$ and $\inv \equiv \inv^1 \land \inv^2$, and, in particular, may simplify the conditions in the process.
We overview these notions and functions in Table~\ref{tab:algo-legend}.

Algorithm~\ref{algo:fr} proceeds as follows: 
It first computes a round invariant (using $\getRoundInv$) from the path condition $\pathcondi{\procF}{\seExp}{\seCvar}{\userindexvar}$ for writing $\seCvar$ to $\seExp$ using function $\procF$ (line~\ref{line:init}). 
Then, it iteratively refines $\invaccumulator$ by ensuring that for all contract functions $\procG$ and all variables $y$ that $\seExp$ directly depends on (as per its data dependencies $\ddepsi{\procF}{\seCvar}{\userindexvar}$), no transaction executing $\procG$ by a user different from $\honuser$ can write to $y$ (lines~\ref{line:ddeps-loop-begin}--\ref{line:ddeps-loop-end}).
This is realized by conjoining (using $\mergeRoundInv$) the current condition with round invariants derived from the conditions $\nowritecondi{\procG}{y}{\nuserindexvar}$, which ensure that function $\procG$ cannot write to $y$. 
Here, the universal quantification over all local variables of $\procG$ in $\nowritecondi{\procG}{y}{\nuserindexvar}$ and the premise $\sendervari{\nuserindexvar} \neq \honuser$ ensure that writes to $y$ are excluded regardless of the call parameters of $\procG$ as long as $\procG$ is invoked by a user different from $\honuser$.
Next, the algorithm enters a loop (lines~\ref{line:main-loop-begin}--\ref{line:main-loop-end}) where it checks for all contract functions $\procG$ whether the current $\invpre$ is an invariant under executions of functions $\procG$ by a user different from $\honuser$ (checked by $\execinvi{\procG}{\predIsNoSender{\honuser}}{\invpre}{\nuserindexvar}$, line~\ref{line:main:algo-fr:inv-check}).
If this is not the case, $\invaccumulator$ is further refined by selecting a variable $y$ from $(\invpre, \inv)$ that had not yet been recorded in $\nowriteset$ for function $\procG$ (line~\ref{line:choose-y}) and conjoining the corresponding condition $\nowritecondi{\procG}{y}{\nuserindexvar}$ (line~\ref{line:merge-no-write-main-loop}) to the current invariant. The set $\nowriteset$ here tracks pairs $(y, \procG)$ for which it has been excluded that $\procG$ writes $y$.
The algorithm terminates if all contract functions have been checked to maintain $\invpre$ or if 
all variables in $\invaccumulator$ got included in $\nowriteset$ for all contract functions $\procG$.
Note that also in the latter case, $\invaccumulator$ is ensured to be an invariant under all possible adversarial invocations of any $\procG$, since 
all variables in $(\invpre, \inv)$ have been shown unalterable by any function $\procG$.

\subheading{Backrunning.}
A deckstacking attacker may not only interfere with user transactions through \emph{frontrunning} (changing the behavior of honest transactions $\htxv$ by placing adversarial transactions \emph{before} them) but also through \emph{backrunning}. 
Backrunning refers to the ability of an attacker to cause divergent behavior of an attacker transaction $\atxv$ by placing it strategically \emph{after} honest user transactions.
To avoid backrunning, we devise an algorithm similar to Algorithm~\ref{algo:fr} (specified in Appendix~\ref{appendix:tool-soundness}).
This algorithm generates conditions, which ensure that the way an honest transaction $\htxv$ writes to a variable $\seCvar$ cannot change how an attacker transaction $\atxv$ writes variables from a set $\vSet$.
The resulting condition $(\invpre, \inv) =\brWriteCond{\procF}{\vSet}{\seCvar}{\roundPred}$ constitutes a round invariant enforcing that attacker transactions cannot change the values of variables in $\vSet$ based on an update to $\seCvar$ provoked by the honest user executing $\procF$. In other words: the assignments to variables in $\vSet$ are independent of how the honest user interacts with $\procF$ to write $\seCvar$.

\subsection{Soundness Proof}
\label{section:conditions:round-based}
To formally reflect the round-based nature of user strategies, we define \emph{round-based} user strategies $\roundstrat{\rstrat}$.
Round-based user strategies (identifiable by the over-arrow) schedule new transactions only at the beginning of a round and keep on submitting the transactions throughout the whole round until they are published.
They are derived from a round strategy $\rstrat$ which determines the transactions to be submitted during the round based on the round's initial configuration $\conf_r$.
In the following construction, we restrict individual users to only submit a single transaction $\htxv$ per round, i.e., only after $k$ blocks, they will submit consecutive transactions (as only then $\htxv$ got certainly executed). 
However, the results naturally extend to multiple transactions per user and round for independent transactions or blockchains that enforce transaction ordering of multiple transactions with a nonce mechanism (such as Ethereum).

\definecolor{figgreen}{RGB}{150,190,45}
\definecolor{figcyan}{RGB}{1,129,157}
\definecolor{figred}{RGB}{229,20,0}
\definecolor{figteal}{RGB}{1,108,102}
\definecolor{figblue}{RGB}{1,68,95}
\begin{figure}[t]  %
\centering 
\includegraphics[width=0.85\columnwidth]{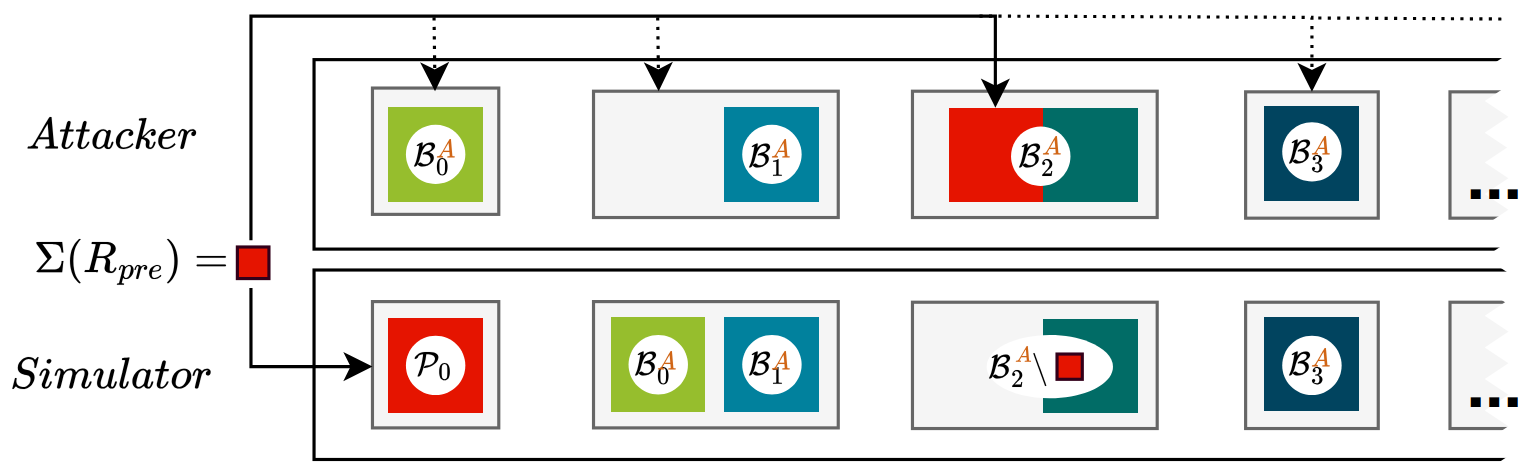}
\caption{Illustration of generic simulator construction: The generic simulator publishes the complete honest user mempool \mbox{\textcolor{red}{\rule{0.6em}{0.6em}}$~= \memcount[0]$} in its first block. Afterwards, the simulator inspects the already published block "\textcolor{red}{\rule{0.6em}{0.6em}}" to catch up (\textcolor{figgreen}{\rule{0.6em}{0.6em}}) with the attacker and to simulate the remaining blocks (\textcolor{figcyan}{\rule{0.6em}{0.6em}}, \textcolor{figteal}{\rule{0.6em}{0.6em}}, \textcolor{figblue}{\rule{0.6em}{0.6em}}) that it generates (by querying the attacker internally).}
\label{fig:proof}
\end{figure}

We show that the interaction conditions generated by Algorithm~\ref{algo:fr} (and the corresponding backrunning algorithm from Appendix~\ref{appendix:tool-soundness}) are sound, meaning that any user strategy $\roundstrat{\rstrat}$ that abides by these conditions ensures that the contract $\contract$ together with $\roundstrat{\rstrat}$ satisfies \sysnameAc{}.
To this end, we construct a generic simulator $\gensstrat$ that relies on the attacker $\astrat$ as a black box, as well as on the contract code $\contract$. 
This simulator will feature two modes of operation:
Based on the interaction conditions (computed from the code of $\contract$, and provided to $\gensstrat$ via $\confPred$), $\gensstrat$ will decide whether there exist transactions that can be safely executed by user $\honuser$ in the current round. 
If this is not the case, then $\gensstrat$ knows that the mempool at this point will be empty (and will stay empty for the whole round), and $\gensstrat$ can simulate the attacker $\astrat$ perfectly. 
If there do exist transactions that can be safely executed by $\honuser$, then $\gensstrat$ will enter a second mode where it first publishes the whole mempool $\mempool$ in its first block and then subsequently reconstructs the attacker's run based on this initial mempool publication, exploiting its access to the previous run. 
An overview of this behavior is illustrated in Figure~\ref{fig:proof}, and we defer the full simulator definition to Appendix~\ref{appendix:tool-soundness}. %
After publishing the initial mempool $\mempool$ in the first block, $\sstrat$ proceeds by invoking $\astrat$ on the empty run with mempool $\memcount[0] = \mempool$ that $\sstrat$ reads from the run.
From this, $\sstrat$ obtains $\acblock[0]$, the last block published by $\astrat$. 
Due to the round-based nature of $\roundstrat{\rstrat}$, $\sstrat$ can also reconstruct the mempool 
$\memcount[1] = \memcount[0] \backslash \acblock[0]$ that $\astrat$ used as input for generating the current block and hence also generate the next attacker block $\acblock[1]$. 
From this point on, $\sstrat$ exactly mimics the actions of $\astrat$ except for excluding mempool transactions from $\mempool$ already published in the first block of the round. 
This construction ensures (due to the required blockchain inclusion time) that at the end of each round both $\arun$ and $\srun$ contain all transactions of the original mempool, as well as all the attacker actions scheduled by $\astrat$ during this run. 

\subheading{Observable Behavior.}
The interaction conditions will enforce a strong observable relation $\simsim$ between the attacker run $\arun$ produced in the presence of $\astrat$ and the simulated run $\srun$ produced in the presence of $\gensstrat$ that covers many practical use cases.
We give a formal definition of $\simsim$ in Appendix~\ref{appendix:tool-soundness}.
Intuitively, $\arun \simsim \srun$ denotes that $\arun$ and $\srun$ executed the same number of rounds and in each completed round all observables (produced by the contract semantics) agree up to reordering. 

\subheading{Soundness.}
Henceforth, we assume contracts $\contract$ where all critical behavior is explicitly marked by an event. 
Concretely, invocations of the function $\procF$ by $\honuser$ will write a specific event variable $\eventvar_{\procF}$, which is assigned all relevant information to be recorded.
We denote the set of all variables that may influence event variables by $\varfixpoint$. 

To state soundness, we make use of the predicate $\frCondR{\conf}{\txv}{\varfixpoint}$ that checks whether executing $\txv$ in a round starting in configuration $\conf$ will have a deterministic effect on variables in $\varfixpoint$ (so either write them to the same value or leave them unchanged).
Intuitively, $\frCondR{\conf}{\txv}{\varfixpoint}$ checks that 
if symbolically executing $\txv$ in $\config$ results in assigning a variable $\seCvar$ from $\varfixpoint$ to an expression $\seExp$
then $\frWriteCond{\procF}{\seExp}{\seCvar}{\roundPred}$ holds. 
Similarly, we use a predicate $\brCondR{\conf}{\txv}{\varfixpoint}$ checking whether for all variables $\seCvar$ from $\varfixpoint$, which will be written when executing $\txv$ in $\conf$, it holds that $\brWriteCond{\procF}{\varfixpoint}{\seCvar}{\roundPred}$. So, intuitively, $\brCondR{\conf}{\txv}{\varfixpoint}$ ensures that an attacker may not rely on variable changes induced by $\txv$ to provoke changes in $\varfixpoint$ (through backrunning).

\begin{theorem}[Soundness of Interaction Conditions]\label{theorem:soundness-simplified}
Let $\contract$ be a contract,
   $\honuser$ be a user and $\varfixpoint$ be the set of variables influencing event variables in contract $\contract$.
   Let $\symExReli{\procG}{}$ be sound and complete symbolic execution relations for functions of $\contract$.
    Let $\cstrat = \roundstrat{\rstrat}$ be the user strategy for a round strategy $\rstrat$ (for $\honuser$ and $\contract$) such that for all%
    \begin{enumerate}[leftmargin=*] %
        \item configurations $\conf, \conf'$ with $\conf \equivfor{\varfixpoint} \conf'$ it holds that $\rstrat(\conf) = \rstrat(\conf')$
        \item configurations $\conf$, $\rstrat(\conf) \neq \emptyset$ implies that $\rstrat(\conf) = \{ \txv \}$ for some $\txv$ 
        and $\frCondR{\conf}{\txv}{\varfixpoint}$, $\brCondR{\conf}{\txv}{\varfixpoint}$
        and $\brCondR{\conf}{\txB}{\varfixpoint}$ hold.     
    \end{enumerate}
    Further, let $\confPred$ be defined as 
    $$\confPred(\conf) :\Leftrightarrow \brCond{\conf}{\txB}{\varfixpoint}~\land~\exists \txv.~ \frCond{\conf}{\txv}{\varfixpoint} ~\land~ \brCond{\conf}{\txv}{\varfixpoint}$$
    Then for all runs $\arun$ such that $\combine{\astrat}{\cstrat} \semantics \arun$
    there exists a run $\srun$ such that $\combine{\gensstrat}{\cstrat} \semantics \srun$
    and $\arun \simsim \srun$.
\end{theorem}
The (slightly simplified) theorem states that for any round-based user strategy that makes scheduling decisions only based on variables in $\varfixpoint$ 
while respecting the interaction conditions, every attacker run can be simulated by our simulator $\gensstrat$. 
In particular, this implies \sysnameAc for the set of all user strategies that satisfy this requirement.
The complete theorem and soundness proof can be found in Appendix~\ref{appendix:tool-soundness}.
Note that the theorem, in addition to respecting the frontrunning and backrunning conditions for the honest user transaction, also requires that $\brCond{\conf}{\txB}{\varfixpoint}$, so that attacker transactions may not impact variables in $\varfixpoint$ by backrunning the block finalizing transaction $\txB$. This requirement excludes an attacker from influencing the values of $\varfixpoint$ through the positioning of attacker transactions within a round.
An attacker may use such power to perform changes to $\varfixpoint$ based on their mempool knowledge, which would break \sysnameAc.
\begin{proof}[Proof Sketch]
At the beginning of each round, the simulator $\gensstrat$ decides whether to simulate the attacker based on the empty mempool (in case that $\confPred$ does not hold) or to publish the mempool in the first block of the round (in case that $\confPred$ holds).
In the first case, by the definition of $\confPred$ it is indeed ensured that $\rstrat$ will not schedule any transaction for the whole round. Consequently, the simulator can perfectly mimic the attacker $\astrat$ since no honest user transactions need to be scheduled.
In the second case, it is still guaranteed that the transactions scheduled by $\gensstrat$ at the end of each round are a permutation of those scheduled by $\astrat$ such that the transaction sequences in both runs can be transformed into each other by pairwise swapping of either (i) the honest user transactions $\txv$ and attacker transactions $\atxv$; or (ii) honest user transactions $\txv$ and block-finalizing transactions $\txB$; or (iii) block-finalizing transactions $\txB$ and attacker transactions $\atxv$. 
In all three cases, we can show that the transactions commute pairwise in that the transactions will still produce the same observables and the variables in $\varfixpoint$ will remain unchanged when swapping the transactions.
This follows from $\frCond{\conf}{\txv}{\varfixpoint}$, $\brCond{\conf}{\txv}{\varfixpoint}$ and $\brCond{\conf}{\tau}{\varfixpoint}$, which enforce that $\txv$ and $\atxv$ only write variables if their critical dependencies are ensured to be left untouched by the other transactions of the round.
In this way, $\varfixpoint$ stays unchanged, ensuring (by its definition) that also observables produced from events are the same independently of the transaction ordering.
\end{proof}

\section{Implementation}
\label{sec:evaluation}

\label{sec:evaluation:impl}
\newcommand{\na}{n/a}

\begin{figure*}
    \footnotesize
    \hfill
    
    \begin{minipage}{0.67\textwidth}
    \centering
    \begin{tabular}{l c c c cc cc ccc}
        \toprule
        Contract &
        Aud. &
        Issue &
        Sev. &
        \multicolumn{2}{c}{LOC} &
        \multicolumn{2}{c}{$\DStwo$} &
        \multicolumn{3}{c}{$\DSthree$} 
        \\
        \cmidrule(lr){5-6}
        \cmidrule(lr){7-8}
        \cmidrule(lr){9-11}
        & & ID & &
        Vuln & \phantom{~}Fix\phantom{~} &
        Sat & Sound &
        Sat & Sound & Func 
        \\
        \midrule
            \texttt{AgentRegistryCore} & CD & \href{https://diligence.security/audits/2022/11/forta-delegated-staking/forta-audit-2022-10.pdf}{5.5} & M & 1394 &1385& unsat & \cmark & sat & \cmark&zero-day \\

            \texttt{BatchedBancorMM} & CD & \href{https://diligence.security/audits/2019/11/aragonblack-fundraising/#fees-can-be-changed-during-the-batch}{6.2} & - &1982 &1998& \na & - & \na & -  & -  \\

            \texttt{ETHRegistrarContr.} & CD & \href{https://diligence.security/audits/2019/03/ens-permanent-registrar/#ethregistrarcontrollerregister-is-vulnerable-to-front-running}{3.3} & C &1313 &1387& unsat & \cmark & sat& \cmark& \cmark\\

            \texttt{Funds} & CD & \href{https://github.com/ConsenSysDiligence/atomic-loans-audit-report-2019-07?tab=readme-ov-file#615-funds-contract-initialization-race-condition}{6.15} & M & 921&935& unsat & \cmark & sat&\cmark& \cmark \\

            \texttt{KeepRandomBeaconOp.} & CD & \href{https://diligence.security/audits/2020/02/thesis-tbtc-and-keep/thesis-tbtc-audit-2020-01.pdf}{5.7} & H & 1997 &2191& unsat & \cmark & unsat &\cmark& \lmark\\
        
            \texttt{Loans} & CD & \href{https://github.com/ConsenSysDiligence/atomic-loans-audit-report-2019-07?tab=readme-ov-file#613-loanspay-allows-anyone-to-force-the-borrower-to-repay-a-loan}{6.16} & M & 932& 937& sat & \cmark & sat& \cmark & \cmark \\

            \texttt{Pool} & CD & \href{https://diligence.security/audits/2023/05/tidal/tidal-audit-2023-04.pdf}{3.9} & M & 1239& 1111& unsat & \cmark &  unsat & \cmark & \xmark  \\

            \texttt{RocketStorage} & CD & \href{https://diligence.security/audits/2021/04/rocketpool/#rocketstorage---anyone-can-setupdate-values-before-the-contract-is-initialized}{6.5} & C & 108&112& sat & \cmark & sat& \cmark & \cmark\\

            \texttt{RocketTokenRPL} & CD & \href{https://diligence.security/audits/2021/04/rocketpool/#rockettokenrpl---inaccurate-inflation-rate-and-potential-for-manipulation-lowering-the-real-apy}{6.9} & H & 377& 407& unsat & \cmark & unsat & \cmark & \xmark \\

            \texttt{TransactionManager} & CD & \href{https://diligence.security/audits/2021/07/connext-nxtp-noncustodial-xchain-transfer-protocol/connext-audit-2021-07.pdf}{4.14} & M & 1241&1248& unsat & \cmark & unsat&\cmark&\lmark\\

            \texttt{Factory} & N & \href{https://github.com/NethermindEth/PublicAuditReports/blob/main/NM0072-FINAL_ARCADIA.pdf}{6.7} & M & 460&460&  unsat & \cmark & unsat&\cmark&\lmark\\

            \texttt{PolygonWorldID} & N & \href{https://github.com/NethermindEth/PublicAuditReports/blob/main/NM0122-FINAL_WORLDCOIN.pdf}{5.4.4} & M & 1517& 1527& unsat & \cmark & sat&\cmark&\cmark\\

            \texttt{PWNSimple} & N & \href{https://github.com/NethermindEth/PublicAuditReports/blob/main/NM0074-FINAL_PWN.pdf}{7.9} & L & 218& 265& sat & \cmark & sat & \cmark & \cmark\\  

            \texttt{StateBridge} & N & \href{https://github.com/NethermindEth/PublicAuditReports/blob/main/NM0122-FINAL_WORLDCOIN.pdf}{5.4.4} & M & 749& 699& unsat & \cmark & sat &\cmark &\cmark \\

            \texttt{ChildRegistrar} & Q & \href{https://certificate.quantstamp.com/full/rara.pdf}{QSP-7} & L & 444&479& unsat & \cmark & sat&\cmark&\cmark\\ 

            \texttt{RootRegistrar} & Q & \href{https://certificate.quantstamp.com/full/rara.pdf}{QSP-7} & L & 689&722& unsat & \cmark & sat&\cmark& \cmark \\ 

            \texttt{GnosisSafeRegistry} & OZ & \href{https://www.openzeppelin.com/news/augur-core-v2-audit-components#critical-severity}{C01} & C & 644& 690& sat & \cmark &sat & \cmark & \cmark\\

            \texttt{OptionsContract} & OZ & \href{https://www.openzeppelin.com/news/opyn-contracts-audit#medium-severity}{M01} & M & 557&927& unsat & \cmark &sat &\cmark&\cmark\\ 

            \texttt{Governance} & OZ & \href{https://www.openzeppelin.com/news/audius-contracts-audit}{Intro} & - & 1017& 1317& sat & \cmark  & sat& \cmark & \cmark \\

            \texttt{ArroToken} & OZ & \href{https://www.openzeppelin.com/news/arrotoken-audit#medium-severity}{M-1} & M & 104&282& unsat & \cmark &unsat&\cmark & \xmark \\

            \texttt{AccountingEngine} & OZ & \href{https://www.openzeppelin.com/news/geb-protocol-audit}{M01} & M & 492&497& sat & \cmark & sat& \cmark & \cmark \\

            \texttt{ACOToken} & OZ & \href{https://www.openzeppelin.com/news/aco-protocol-audit#medium-severity}{M03} & M & 890&912& sat & \cmark & sat&\cmark& zero-day \\

            \texttt{LoanProposalImpl} & ToB & \href{https://github.com/trailofbits/publications/blob/master/reviews/2023-04-mysoloans-securityreview.pdf}{3} & M & 993& 1031&sat & \cmark & unsat & \cmark & \lmark \\

            \texttt{SpoolAccessControl} & ToB & \href{https://github.com/trailofbits/publications/blob/master/reviews/2023-03-spool-platformv2-securityreview.pdf}{2} & H & 658& 661& unsat & \cmark &  \na & - & - \\
        
         \bottomrule
    \end{tabular}
    \end{minipage}
    \begin{minipage}{0.27\textwidth}
    \centering
    \textbf{Auditors (Aud.)}
    \begin{tabular}{ll}
    \midrule
    CD & ConsenSys Diligence \\
    N  & Nethermind \\
    Q  & Quantstamp \\
    OZ & OpenZeppelin \\
    ToB & Trail of Bits \\
    \bottomrule
    \\
    \multicolumn{2}{c}{\textbf{Severity (Sev.)}} \\
    \midrule
    C & Critical \\
    H & High/Major \\
    M & Medium \\
    L & Low/Minor  \\
    - & Unranked \\
    \bottomrule
    \\
    \multicolumn{2}{c}{\textbf{Verdict}} \\
    \midrule
    sat & Satisfiable \\
    unsat & Unsatisfiable \\
    \cmark & Yes (for soundness \& functionality) \\
    \xmark & No (for soundness \& functionality) \\
    \lmark & Overapproximation in symbolic exec. \\
    \na & Timeout/undecidable by Z3 \\
    \bottomrule
    \end{tabular}
    \end{minipage} 
    
    \captionof{table}{\sysnameTool results for datasets $\DStwo$ and $\DSthree$ for satisfiability, soundness and functionality of synthesized condition.
    }
    \label{tab:appendix:tool-evaluation}
\end{figure*}

We implemented the synthesis Algorithm \ref{algo:fr} and refer to this prototype implementation as \sysnameTool (short for \textit{No Deckstacking Synthesizer}).
\sysnameTool synthesizes frontrunning-resistant interaction conditions $\frWriteCond{\procF}{\seExp}{\seCvar}{\roundPred}$ for an honest user $\honuser$ for all events in a smart contract function $\procF$.
The implementation of \sysnameTool builds on top of the existing symbolic execution engine from \cite{zhang2024nyx} and the static analysis framework in~\cite{feist2019slither}.
In particular, we used the symbolic execution engine to implement the path condition generation,
dependency analysis 
and the invariant verification from Section~\ref{section:conditions}.
The implementation is agnostic to the actual symbolic execution engine, closely follows the description of the algorithm and uses Z3~\cite{de2008z3} as its underlying constraint solver.

\sysnameTool utilizes smart strategies during round-invariant strengthening ($\getRoundInv$) and merging ($\mergeRoundInv$):
To avoid over-strengthening of path conditions and losing precision, round-invariant checks and strengthening attempts materialize only on subterms of a path condition.
In detail, we (lazily) compute the disjunctive normal form (DNF) of a path condition and enforce round invariants on every literal of that DNF---literals that are already round-invariant remain unchanged and will thereby not be over-strengthened.
Within bounds, \sysnameTool keeps the DNF structure intact during the merging of conditions to produce more naturally structured output conditions.
For example the output condition in Figure \ref{fig:cond:timlock-fix} for the \lstinline|TimelockedFeeMinted| contract from Figure \ref{fig:sc:timlock-fix} closely resembles the manually derived condition from Section \ref{section:conditions}. 
Note that extracting all but the rightmost bit and checking equivalence with \lstinline|0| is equivalent to checking that the value is less than or equal to \lstinline|1|. 

\begin{figure}[t]
  \raggedleft
  \footnotesize
  \begin{minipage}{0.925\linewidth}
\begin{lstlisting}[basicstyle=\scriptsize\ttfamily, numberstyle=\tiny]
ForAll(ATTACKER_executeFeeChange(uint256)__newFee,
 Or(Extract(31, 1, 
      timestamps[sha256(abi.encodePacked())("newFee",
        ATTACKER_executeFeeChange(uint256)__newFee)]) == 0),
    attacker_sender != owner,
    Not(ULE(timestamps[sha256(abi.encodePacked()("newFee",
              ATTACKER_executeFeeChange(uint256)__newFee))],
           block_number_honest + (k - 1)))))
\end{lstlisting}
  \end{minipage}
      \caption{\sysnameTool output condition for the \lstinline|TimelockedFeeMinted| contract from Figure \ref{fig:sc:timlock-fix}.}
      \label{fig:cond:timlock-fix}
  \end{figure}

Lastly, \sysnameTool has full support for additional user-supplied preconditions on the contract state to compute more situational interaction conditions.
E.g., a realistic precondition that would simplify the above condition even further is \lstinline|honest_sender != owner|.

\subsection{\sysnameTool Evaluation}\label{section:eval}

To validate \sysnameTool, 
we run \sysnameTool on all 48 contracts in $\DStwo$ and $\DSthree$ (cf. Section~\ref{sec:problem-statement}) with a time limit of two hours per contract.
After that, we manually evaluated the computed interaction conditions of all contract pairs, i.e., the vulnerable and fixed versions, for (1) soundness and whether the conditions (2) do not restrict the intended sound functionality of the fixed contract.

Table~\ref{tab:appendix:tool-evaluation} presents the detailed results of our tool evaluation. 
For each contract pair, we report the condition and soundness judgments, as well as the functional correctness of the conditions synthesized for the fixed contract.
For (1) soundness, we concluded that all computed conditions are indeed sound w.r.t. \sysnameAc.
Checking soundness, in particular, included checking that the interaction condition for vulnerable contracts does not allow the described frontrunning attack.
Regarding (2) functionality, for 13 out of 22 fixed contracts, we found the conditions to be functional.
In two of the remaining cases, through an investigation of the interaction condition, we found the contract to remain vulnerable to the described attack.
Therefore, \sysnameTool correctly synthesized non-functional conditions as no sound functional condition for the behavior could exist. 
\sysnameTool was unable to compute a satisfiable condition in seven cases.
Four of those seven cases can be attributed to overapproximations in the symbolic analysis, e.g., due to complex language features such as dynamic contract creations. 
The analysis hit the timeout in the remaining three contract executions.
These timeouts can be partially attributed to Z3 being unable to solve satisfiability queries due to the undecidability of the theories used in the symbolic execution (e.g., quantification over arrays).

\vspace{0.5 em} 
\noindent \textbf{Zero-day Vulnerabilities.}
Both vulnerable contracts that we found in $\DSthree$ were zero-day vulnerabilities.

\begin{figure}[t]
\raggedleft
\footnotesize
\begin{minipage}{0.925\linewidth}
\begin{lstlisting}
contract CollateralProtocol {  // from ACOToken
  mapping (address => uint) members;
  mapping (uint => uint) debt;
  mapping (uint => uint) collateral;
  
  function sellCollaterals(uint payoutAmount,
                    uint salt) onlyOwner {
    uint start = salt % | collateral |;
    for(uint i = start; payoutAmount > 0; ++i)
      if (0 < debt[i]) {
        payoutAmount -= collateral[i];
        collateral[i] = 0;
      } 
  } // ...
\end{lstlisting}
\end{minipage}
    \caption{Simplified \lstinline|CollateralProtocol| vulnerability~\cite{acoprotocol-fix}.}
    \label{fig:sc:collateral}
\end{figure}
The first zero-day vulnerability is present in a collateral protocol DeFi smart contract~\cite{acoprotocol-fix} that provides functionalities to buy options
by depositing a collateral that can be liquidated again when redeeming the option.
However, in the given contract, a malicious user can prevent their own collateral from being liquidated using frontrunning.
The deployed liquidation algorithm liquidates the collaterals of users (whose deposits are not backed by options) sequentially until an accumulated \lstinline|payoutAmount| is reached by iterating through an array of collateral holders starting at some position \lstinline|start|. Prior to the audit, this starting position \lstinline|start| of the liquidation sequence was a hard-coded constant.
This would have enabled malicious users who were monitoring the mempool for liquidation calls to frontrun any liquidation sequence to back users that appear earlier in the liquidation sequence with options, thereby circumventing their liquidation.
The auditor identified this \deckstackingAc vulnerability.

As a consequence, the developers modified the contract to allow the liquidator to dynamically set the \lstinline|start| of the liquidation sequence to a user-defined \lstinline|salt| value.
We sketched the essence of the modified contract in Figure~\ref{fig:sc:collateral}.
However, a \deckstackingAc attacker with access to the mempool can precompute all transactions before their inclusion in the blockchain and hence foresee the effect on the liquidation behavior for any given \lstinline|salt| provided by the liquidator.
Consequently, the malicious user can adapt to the provided \lstinline|salt| and continue frontrunning the liquidation sequence to prevent their own liquidation.

The second zero-day vulnerability is present in a staking contract~\cite{forta-fix} where new participants are registered using a \lstinline|createAgent| function.
The \lstinline|createAgent| function takes an \lstinline|agentId|, an \lstinline|owner|, and additional metadata as parameters and creates a new agent entry in the contract state when the given \lstinline|agentId| is not registered yet.
This poses a series of denial-of-service risks caused by malicious users monitoring the mempool for \lstinline|createAgent| function calls and frontrunning them with duplicated \lstinline|createAgent| calls with altered metadata (cf. Figure~\ref{fig:sc:voting-dos}).
This effectively reverts the user call as their transaction would register a duplicate \lstinline|agentId|.
In particular, the malicious user could register their agent with the same \lstinline|owner| address as the original transaction.
While this strengthens the attack's impact, it could trick users into thinking that their original transaction has been accepted.
This attack vector was pointed out by the audit company.

The contract developers implemented a fix ensuring that the \lstinline|owner| parameter cannot be frontrun anymore, i.e., the fix adds a condition \lstinline|msg.sender == owner| to the \lstinline|createAgent| function.
A simplified implementation of the allegedly fixed version can be found in Figure~\ref{fig:sc:staking}.
However, this does not mitigate the underlying DoS attack on the \lstinline|createAgent| function.
The malicious user can continue frontrunning \lstinline|createAgent| transactions and permanently prevent the registration of new agents.

\begin{figure}
    \raggedleft
    \begin{minipage}{0.925\linewidth}
    \begin{lstlisting}
contract StakingProtocol {  // from AgentRegistryCore
    mapping(uint => address) agents;
    
    function createAgent
            (address owner, uint agentId, /*...*/) {
        require(msg.sender == owner);
        require (agents[agentId] == 0);
        agents[agentId] = owner;
        // ...
    \end{lstlisting}
    \end{minipage}
        \caption{Simplified \lstinline|StakingProtocol| vulnerability \cite{forta-fix}.}
        \label{fig:sc:staking}
\end{figure}

We responsibly disclosed both vulnerabilities to the corresponding contract developers and auditors.

\section{Related Work}

\noindent \textbf{Frontrunning Attacks and Mitigations.}
Different forms of frontrunning attacks and defenses have been widely discussed in the literature and systematized by
Eskandari \etal~\cite{DBLP:conf/fc/EskandariMC19}, Baum \etal~\cite{DBLP:journals/iacr/BaumCDFG21}, and Zhou \etal~\cite{DBLP:journals/iacr/ZhouXECWWQWSG22}. 
In particular, Eskandari \etal~\cite{DBLP:conf/fc/EskandariMC19} provide a taxonomy of frontrunning attack methods as well as techniques to mitigate frontrunning incidents. 
All these works study the concept of frontrunning empirically and do not aim to provide a formal characterization of the underlying problem.

\vspace{0.5 em}\noindent \textbf{Characterizing Frontrunning.}
As discussed in Section~\ref{ps-sota}, the existing approaches for characterizing a smart contract's vulnerability to frontrunning attacks center around the concepts of MEV and TOD.
MEV~\cite{daian2019flash, babel2023clockwork} has been commonly used as a metric to quantify the monetary gains of frontrunning attackers and has been formally defined in~\cite{babel2023clockwork}. 
However, this definition assumes knowledge of the concrete mempool and blockchain state and, hence, cannot be used to statically assess a contract's susceptibility to MEV extraction.
In more recent work, Bartoletti \etal~\cite{bartoletti2023theoretical} study MEV in a formal model of smart contract execution and introduce the notion of \emph{universal MEV}, which characterizes the maximal gain that can be achieved by any arbitrary adversary, and study formal proofs of MEV freedom. The limitations of MEV, discussed in Section~\ref{ps-sota}, still apply to universal MEV.

The notion of TOD has been first introduced in~\cite{luu2016making} and then formally defined in~\cite{grishchenko2018semantic}.
Since TOD coarsely overapproximates a contract's vulnerability to frontrunning attacks, it has been used as a basis of several practical tools (most notably Nyx~\cite{zhang2024nyx} and Sailfish~\cite{bose2022sailfish}) for the static detection of frontrunning vulnerabilities. These tools proceed by localizing the causes of transaction order dependencies (e.g., read-write hazards among variables, which could be accessed in concurrent transactions) and then apply heuristics to identify whether an identified cause gives rise to an exploitable frontrunning attack. 
Since those heuristics are not formally defined, it is unclear which precise security notion the resulting tools are checking.

\vspace{0.5 em}\noindent \textbf{Fair Ordering.}
An orthogonal line of work is mitigating \deckstacking on the system level by enforcing a mining process that provides fair ordering guarantees.
In its generality, it has been shown in~\cite{DBLP:conf/crypto/Kelkar0GJ20} that a strong fair ordering guarantee cannot be achieved in an asynchronous system. 
However, multiple works propose protocols to achieve weaker notions of fair ordering both in the context of a permissioned Byzantine Fault Tolerance (BFT) setting (which assumes a fixed set of miners)~\cite{DBLP:conf/crypto/Kelkar0GJ20,kelkar2023themis} and in the permissionless blockchain setting (where the set of miners can change dynamically)~\cite{DBLP:conf/asiapkc/KelkarDK22,cachin2022quick}. 
These works are orthogonal to the results presented in this paper in that \sysnameAc aims to characterize the sensitivity of a smart contract to the influences of a strong reordering miner.
If a contract satisfies \sysnameAc, then this property would also hold for any miner model that assumes weaker reordering capabilities.

\vspace{0.5 em}\noindent \textbf{Architectural-level mitigations.}
Alternative approaches to mitigate \deckstacking on an architectural level center around proposals for cryptographic protocols that aim to decouple the knowledge of mempool transactions and the attacker's ability to react upon this knowledge.
F3B~\cite{zhang2022f3b} achieves this goal by encrypting transactions under a secret key that is only revealed by a dedicated decentralized committee upon transaction inclusion.
FIRST~\cite{sariboz2025first} makes use of verifiable delay functions to ensure that users can only publish contract transactions after an amount of time that allows all pending mempool transactions to be included in the blockchain.
Both solutions rely on additional (partially) trusted committees and lack rigorous analysis.
In contrast to these architectural-level mitigation attempts, our work targets \deckstacking at the smart-contract level, aiding the design of provably secure contracts within the existing blockchain architecture.

\section{Conclusion}
In this paper, we introduced \sysname, the first formal and precise characterization of a contract's robustness against frontrunning attacks. The definition is motivated by the fundamental shortcomings of existing frontrunning detection techniques, which we empirically demonstrate. More precisely, we showed that existing notions are particularly restricted to cases where: 
(1) attacks need to be immediate, and 
(2) the adversary's goal needs to be \emph{monetary}. Our empirical study showed that these limitations prevent MEV from capturing more than 55\% of audited vulnerabilities with confirmed frontrunning. 

\sysname is a simulation-based notion that contrasts how the executions of a contract interacting with an unrestricted block-generating attacker differ from those executions possible in the presence of a limited attacker that makes scheduling decisions without inspecting the mempool.
In particular, our definition highlights that resistance to frontrunning attacks depends on user interactions. Based on this insight, we propose an algorithm to compute sound interaction conditions from a contract's code that are sufficient to prove resistance to frontrunning attacks, and we provide a prototype implementation, which we validate against a benchmark of real-world smart contracts.

\begin{acks}
    We would like to thank the reviewers for their helpful
    feedback.
    This work has been supported by the Heinz Nixdorf Foundation through a Heinz Nixdorf Research Group (HN-RG) and funded by the Deutsche Forschungsgemeinschaft (DFG, German Research Foundation) under Germany's Excellence Strategy -- EXC 2092 CASA -- 390781972.
\end{acks}

\bibliographystyle{ACM-Reference-Format}        %
\bibliography{bib.bib}

\appendix

\section*{Generative AI Usage}
This paper was edited for grammar using Grammarly and ChatGPT. 

The development of \sysnameTool~was aided by GitHub Copilot, which generated skeleton implementations from high-level descriptions of the toolchain and inferred interface bindings for calls into existing codebases used for symbolic analysis. During development, most of this generated code was replaced with human-written implementations. An exception is a set of self-contained modules for SMT formula transformations, authored entirely by GitHub Copilot; these modules underwent thorough testing and fuzzing, with Z3 serving as an external correctness oracle. We additionally developed a hand-crafted test suite to validate the correctness of \sysnameTool~as a whole.

\section{Open Science}
We provide three different kinds of artifacts with this paper:
\begin{enumerate}
    \item Datasets: As described in Section~\ref{ps-sota}, we leverage 287 publicly available professional smart contract audits to create three datasets (\DSone, \DStwo, \DSthree) containing audits mentioning frontrunning vulnerabilities, the source code of those vulnerabilities, and the source code of their respective fixes. 
    \item Code: The source code of the \sysnameTool~prototype implementing the synthesis algorithm (cf. Section~\ref{section:conditions}) for frontrunning secure interaction conditions as described in Section~\ref{sec:evaluation:impl}.
    \item Evaluation-Results: The results of the evaluation of static and dynamic frontrunning analysis approaches in Section~\ref{ps-sota} and the evaluation results of \sysnameTool~in Section~\ref{section:eval} and Table~\ref{tab:appendix:tool-evaluation}.
    
\end{enumerate}

To ensure reproducibility of our results, our complete datasets (\DSone, \DStwo, \DSthree), the codebase of \sysnameTool, and the evaluation results are available on GitHub here~\cite{deckstacking2026nods} and here~\cite{deckstacking2026mevnyxsailfish}.

\section{Ethical Considerations}

This work studies frontrunning vulnerabilities in Ethereum. As these classes of vulnerabilities are already well documented, our emphasis is on developing a precise definition to capture them and to responsibly advance the fairness and security of blockchain systems.

Guided by our analysis and definition, we identified that two fixes previously implemented in response to audit reports remain susceptible to frontrunning attacks.
We responsibly disclosed these findings to the corresponding Ethereum contract owners and the audit companies. 

We additionally disclosed our results to the Enterprise Ethereum Alliance and have worked with them to adjust their vulnerability classification in the EEA EthTrust Security Levels Specification~\cite{nevile_2023}.

\section{\sysname Model}\label{appendix:model}

In this section, we provide further details and formalizations that accompany the definition of \sysname presented in Section \ref{sec:deckstacking_resistance}.

\subsection{Transaction Semantics}
\label{appendix:inference-rules}

The complete blockchain state is modeled as a configuration $\conf = (\WmvA, \contracts, \blocknumber)$ where $\WmvA$ maps users to their wallet states, $\contracts$ maps all contracts to their states, and $\blocknumber$ is the current block number.
We also write $\conf.\blocknumber$ to access the block number $\blocknumber$ of a configuration $\conf$.

Since \sysnameAc is a general definition that applies to smart contracts in different languages, 
we do not fix the concrete layout of contracts   (i.e., the modeling of program execution, functions, constructors, etc.) or of its storage here.

We model the evolution of blockchain states with a transition system on such configurations $\conf$.
Configurations can be advanced with transactions $\txv$.
Those transactions can either be a contract call $\txT$ 
whose execution impacts the states of users and contracts
or the distinct block-finalization transaction $\txB$ that increments the block number.
To capture the effects of executing a transaction $\txv$ (which may involve multiple asset transfers or the logging of different events), each transition $\conf_0  \trans[\obs]{\txv} \conf_1$ emits a list of observables $\obs$. 
Formally, transitions are defined by inference rules 
such as the following rule of contract calls:

{\small
\begin{mathpar}
    \infer{
    \conf = (\WmvA, \contracts, \blocknumber) \\
    (\WmvA,\contrS)_{\blocknumber} \xmapsto[\obs]{\txT} (\WmvAi,\contrS')_{\blocknumber} \\
    \conf' = (\WmvA', \contracts', \blocknumber) \\
    \contracts(\contrC) = \contrS \\
    \txT = \contrCall{\procF(\overrightarrow{\textit{txarg}})} \\
    \contracts' = \contracts[\contrC :=  \contrS'] \\
    }
    {
    \conf  \trans[\obs]{\txT} \conf'
    } 
\end{mathpar}
}

The rule states that when transaction $\txT$ is a contract call \linebreak $\contrCall{\procF(\overrightarrow{\textit{txarg}})}$ of contract $\contrC$'s function $\procF$ with arguments $\overrightarrow{\textit{txarg}}$, then the contract state $\contrS$ of $\contrC$ and the wallet state $\WmvA$ of all blockchain users are updated according to some smart contract execution semantics $(\WmvA,\contrS)_{\blocknumber} \xmapsto[xs]{\txT} (\WmvAi,\contrS')_{\blocknumber}$.
The updates are applied to $\conf$, resulting in $\conf'$.
The smart contract execution semantics $\xmapsto[xs]{\txT}$ is a formal way of specifying the smart contract execution behavior and can be instantiated for different smart contract languages such as \cite{hildenbrandt2018kevm,jiao2020semantic,marmsoler2021denotational}.
Block finalization is modeled by a second inference rule  \mbox{$(\WmvA, \contracts, \blocknumber) \trans{\txB} (\WmvA, \contracts, \blocknumber')$} that increases the block counter by one ($\blocknumber' = \blocknumber +1$).
We call a sequence of (valid) state transitions
\noindent
{\small
$
\conf_0  \trans[\obs_0]{\txv_0} \conf_1 \trans[\obs_1]{\txv_1} \cdots \trans[\obs_{n-1}]{\txv_{n-1}} \conf_n
$}
a \emph{run} and also write $\conf_0 \trans[\concat{\obs_0}{\concat{\obs_1}{\concat{\dots}{\obs_{n-1}}}}]{[\txv_0, \txv_1, \dots, \txv_{n-1}]}^{*} \conf_n$ for short (where $\concat{\obs_1}{\obs_2}$ denotes list concatenation).

Note that we will assume in the following that the sender of a transaction $\txv$ is always provided in the function arguments and can be accessed via a function $\sender{\txv}$.

\subsection{Symbolic Model of Cryptography}
\label{appendix:crypto}

The definition of \sysname incorporates a symbolic notion of cryptography by instantiating the smart contract semantics with an equational theory $\eqth$.

\vspace{0.5 em}\noindent \textbf{Equational Theory.}
We assume a standard setup for an equational theory $\eqth$ on symbolic terms including functions, equations and constants \cite{dershowitz1990rewrite}.
Functions are defined as congruent rewriting rules ${(\rho,\blocknumber,\sigma)~ \vdash_{\eqth} \alpha \rightarrow \beta}$ with access to an environment consisting of a variable mapping $\rho$, the current block number $\blocknumber$ and the wallet state $\sigma$.
The base set of functions of the equational theory is enriched by blockchain-specific access to this environmental information, e.g.,
\[
    \begin{array}{c}
        (\rho,\blocknumber,\sigma)~ \vdash_{\eqth} \getblocknumber \rightarrow \blocknumber 
        \\[5pt]
        (\rho,\blocknumber,\sigma)~ \vdash_{\eqth} \#\tokT \rightarrow \sigma(\tokT) 
        \\[5pt]
  \end{array}
  \]
where $\getblocknumber$ ($\#\tokT$) is the smart contract language expression to access the current block number (resp. the current token balance of token kind $\tokT$). 
Equational theories needed to model cryptographic primitives used in the context of frontrunning protection in smart contracts include hashing, signing and zero-knowledge verification, which have been studied in detail for different domains~\cite{bursuc2022contingent, meier2013tamarin, backes2008zero}.

\vspace{0.5 em}\noindent \textbf{Secrets.}
Transactions and observables emitted in runs may contain symbolic terms and we write 
$\getTerms{\run}$ to denote the set of all (top-level) terms occurring in run $\run$.
These terms can include symbolic secrets and we use the notation $\sterms_1 \secD \sterms_2$ to indicate that a set of (symbolic) terms $\sterms_2$ can be derived from a set of terms $\sterms_1$ using theory $\eqth$. 

Additionally, attackers and honest users each have a dedicated set of fresh cryptographic secrets $\secrets_{\astrat}$ ($\secrets_{\honusers}$).
The attacker (honest user) may only use secrets in their own transaction $\txv^{\astrat}$ from $\secrets_{\astrat}$ ($\secrets_{\honusers}$) or if they are derivable from public knowledge, i.e., if it is revealed in the mempool or the previous blockchain run: 
\[
\getTerms{\run} \cup~ \getTerms{\mempool} \cup~ \secrets \secD \{ \txv^{\astrat} \}
\]
The notion of secrets enables us to faithfully model knowledge-based \deckstacking-attacks when the smart contract logic relies on cryptographic operations.

\subsection{Real-World Blockchain Execution}
\vspace{0.5 em}\noindent \textbf{Modeling Mempools. } \label{appendix:model:mempool}
We define a mempool $\mempool$ as a list \linebreak \mbox{$[\txv_0, \dots, \txv_i]$} of valid transactions by honest blockchain users $\honusers$.
The choices of all honest users $\honusers$ for a contract $\contract$ are modeled by a set of symbolic user strategies $\ustrat{\contract}{\honusers}{\secrets}$.
Each strategy $\cstrat \in \ustrat{\contract}{\honusers}{\secrets}$ is a function mapping the current blockchain run $\run$ to the next mempool $\mempool_{\cstrat}$ expressing the combined strategy of all honest users while only using their shared set of symbolic secrets $\secrets$.
To ensure correctness of every mempool $\mempool_{\cstrat}$, for each transaction $\txv \in \mempool_{\cstrat}$ it needs to hold that $\txv$ is 

\begin{enumerate}[(i)]
    \item \label{const:semanvalid} valid, i.e., $\exists \conf': \conf_\run \trans{\txv } \conf'$ (with $\conf_\run$ being $\run$'s last configuration),
    \item \label{const:ownmoves} authentic, i.e., $\sender{\txv } \in \honusers$,
    \item \label{const:ownsecrests} $\eqth$-derivable, i.e., 
    $\getTerms{\run} \cup~ \secrets \secD \{ \txv \} $,
    \item \label{const:global} persistent, i.e., if $\txv \in \cstrat_i(\run)$ and $\conf_{\run} \trans{\txv' } \conf_{\run'}$ with $\txv  \neq \txv' $ then $\txv  \in \cstrat_i(\run')~$ ($\run'$ extends $\run$ with $\txv'$),
\end{enumerate}

\vspace{0.5 em}\noindent\textbf{Real-World Scheduling Strategies. } %
Formally, a symbolic scheduling strategy is a function $\astrat(\run, \sseq)_{\secrets}$, parameterized with a set of symbolic secrets $\secrets$ accessible by the function, that produces a valid transaction sequence $\bblock = [\txv_0, \dots, \txv_{n-1}, \txB] $ where each transaction $\txv_i \in \bblock$ 
when sent by an honest user is unique and stems from the mempool ($\txv_i \not \in \sseq \Rightarrow \sender{\txv_i} \not \in \honusers$), and otherwise needs to be constructible using attacker knowledge ($\txv_i \not \in \sseq \Rightarrow \getTerms{\run} \cup~ \getTerms{\mempool} \cup~ \secrets \secD \{ \txv_i \}$).
Further, it is guaranteed that a transaction $\txv \in \mempool$ is included in block $\bblock$ if it was proposed at least $k$ blocks earlier by an honest user.
Formally, this blockchain inclusion guarantee (for a fixed $k$) is ensured by requiring that 
whenever $\txv_i \not \in \bblock =\astrat(\run, \cstrat(\run))$, then there is also no previous sub-run $\run'$ with $\conf_{\run'} \rightarrow^* \conf_{\run}$ that is $k$ blocks behind (denoted by $\lastBlockNumber(\run') + k \leq \lastBlockNumber(\run)$ where $\lastBlockNumber(\run)$ gives the last block number of $\run$) and where the user strategy already proposed $\txv_i$ ($\txv_i \in \cstrat(\run')$).

The inductive transition rule for $\combine{\astrat}{\cstrat} \semantics \run$ is given in the following rule:
\begin{mathpar}
    \infer{
        \combine{\astrat}{\cstrat}  \semantics \run' \\
     \astrat(\run', \cstrat(\run')) = \bblock = [\txv_0, \txv_1, \dots, \txv_j ]\\
     \txv_i = \txB \leftrightarrow i = j \\
     \run = \run' \trans[\obs_0]{\txv_0} \conf^1 \trans[\obs_1]{\txv_1} \dots \trans[\obs_j]{\txB} \conf^j_{\run} 
    }
    {
    \combine{\astrat}{\cstrat}  \semantics \run
    }
\end{mathpar}

\subsection{Substitution and Stability Criterion}
\label{appendix:substitution} 

What complicates the definition of non-adaptivity is the fact that while the $\sstrat$'s scheduling strategy itself should be uniform (w.r.t $\mempool$), this does not hold for the strategy's output. The blocks generated by $\sstrat$ may include honest user transactions, and hence, immediately depend on $\mempool$.
We will capture this by requiring that any modification $\mempool'$ of $\mempool$, i.e., exchanging transactions with different ones, or shortening $\mempool$, should similarly be reflected in the block $\bblock = \sstrat(\run, \mempool')$.

To this end, we consider that $\mempool$ can be modified by a sequence of individual substitutions of the form
\begin{itemize}
    \item $\singlesubst{\txv_1}{\txv_2}$ (indicating that $\txv_1$ gets replaced by $\txv_2$)
    \item $\singlesubst{\txv}{\bot}$ (indicating that $\txv$ in a final position gets removed) 
\end{itemize}
such that these substitutions can also be applied to $\bblock = \sstrat(\run, \mempool)$. 
So, in particular, any transaction $\txv_1$ from $\mempool$ that occurs in $\bblock$ should be replaced in the same way in $\bblock$ as it is replaced in $\mempool$.
Similarly, if $\mempool$ is shortened by transaction $\txv$ then $\txv$ should also be removed from $\bblock$. 

Formally, we define the application $\applysubst{\mempool}{\singlesubst{\txv}{x}}$ of a substitution $\singlesubst{\txv}{x}$ on mempool $\mempool$ recursively as follows: 

{\small
\begin{align*}
    \applysubst{\mempool}{\singlesubst{\txv}{x}}
    := 
    \begin{cases}
        \concat{[\txv_1]}{(\applysubst{\mempool'}{\singlesubst{\txv}{x}})}& \mempool = \concat{[\txv_1]}{\mempool'} ~\land~ \txv \neq \txv_1 \\
        & ~\land~ x \not \in \mempool\\
        \concat{[x]}{\mempool'}& 
        \mempool = \concat{[\txv_1]}{\mempool'} ~\land~ \txv = \txv_1 \\
        & ~\land~ x \neq \bot ~\land x \notin \mempool \\
        [] & \mempool = [\txv] ~\land~ x = \bot 
    \end{cases}
\end{align*}
}%

Note that this function is only defined on mempools $\mempool$ that 
(1) contain the transaction $\txv$ to be substituted at least once; 
(2) contain the transaction $\txv$ to be deleted in end position; 
(3) for which the substitution will not produce collisions (so will not insert a transaction that already exists in $\mempool$).
Further, assuming that all mempools $\mempool$ are free of duplicates also $\mempool' = \applysubst{\mempool}{\singlesubst{\txv}{x}}$ is duplicate-free (since the substitution is only defined given that it introduces a new transaction).
Now, for any two mempools $\mempool_1$ and $\mempool_2$, there exists a sequence of substitutions that transforms $\mempool_1$ into $\mempool_2$ (or vice versa). We write for such a sequential application $\applysubst{((\applysubst{\mempool}{\singlesubst{\txv_1}{x_1}})\dots)}{\singlesubst{\txv_n}{x_n}}$ of a sequence of substitutions also $\applysubst{\mempool}{\multisubst{\simsubst{\txv_1}{x_1}, \dots, \simsubst{\txv_n}{x_n}}}$ for short.

To reflect the effects of mempool changes on blocks, we define the application $\applysubstblock{\bblock}{\singlesubst{\txv}{x}}$ of a substitution $\singlesubst{\txv}{x}$ on a block $\bblock$ as follows:
\begin{align*}
\small
    \applysubstblock{\bblock}{\singlesubst{\txv}{x}}
    := 
    \begin{cases}
       [\tau] & \bblock = [\tau] \\
        \concat{[\txv_1]}{(\applysubstblock{\bblock'}{\singlesubst{\txv}{x}})}& \bblock = \concat{[\txv_1]}{\bblock'} \land \txv \neq \txv_1 \\
        \concat{[x]}{\bblock'}& \bblock = \concat{[\txv_1]}{\bblock'} \\& ~\land~  \txv = \txv_1
        ~\land~ x \neq \bot \\
        \applysubstblock{\bblock'}{\singlesubst{\txv}{x}} & \bblock = \concat{[\txv_1]} {\bblock'} \\
        & ~\land~ \txv = \txv_1 
         ~\land~ x = \bot 
    \end{cases}
\end{align*}
As opposed to substitution on mempools, block substitution allows for deleting transactions that are not in the end position. 
This is needed to reflect that shortening a mempool by transactions $\txv$ results in the deletion of $\txv$ in the block, where it may not necessarily occur as the last element.

With these definitions in place, we can formally define the non-adaptiveness of simulators:

\begin{definition}[Non-adaptiveness]\label{appendix:lemma:sc}
    Let $\honusers$ be a set of honest users. 
    An attacker strategy $\sstrat$ (against $\honusers$) is \emph{non-adaptive} (written $\isSim{\sstrat}$) if and only if the following holds: 
    For all runs $\run$ and all mempools
    $\mempool$ and $\mempool'$ such that 
    there exist $n \in \mathbb{N}$ and 
    substitutions $\singlesubst{\txv_1}{x_1}, \dots, \singlesubst{\txv_n}{x_n}$ with 
    $x_1, \dots, x_n \in \{ \bot\} \cup \{ \txv \in \TxU ~|~ \sender{\txv} \in \honusers\}$
    and 
    $\mempool' = \applysubst{\mempool}{\multisubst{\simsubst{\txv_1}{x_1}, \dots \simsubst{\txv_n}{x_n}}}$, it holds that
    
    {\small
  \[
    \sstrat(\run, \mempool') = \applysubstblock{\sstrat(\run, \mempool)}{\multisubst{\simsubst{\txv_1}{x_1}, \dots \simsubst{\txv_n}{x_n}}}.  
    \] }
\end{definition}

Intuitively, this definition requires that if a mempool $\mempool'$ resulted from applying a sequence of substitutions $\singlesubst{\txv_1}{x_1}, \dots, \singlesubst{\txv_n}{x_n}$ (which only substitutes honest user transactions) to a mempool $\mempool$, then the simulator $\sstrat$ invoked on $\mempool'$ will produce the same block that would result from invoking $\sstrat$ on $\mempool$ and applying substitutions $\singlesubst{\txv_1}{x_1}, \dots, \singlesubst{\txv_n}{x_n}$ afterwards.
Note that even in the cases of untargeted frontrunning or when $\sustrat{\contract}$ consists only of a single honest user strategy, the simulator's capabilities are limited by non-adaptiveness.
In contrast to the attacker $\astrat$, the simulator $\sstrat$ can never rely on mempool knowledge to construct transaction $\txv$ ($\txv \not \in \sseq \rightarrow \getTerms{\run} \cup~ \secrets \secD \{ \txv \}$).
The non-adaptiveness of the simulator, in particular, requires the simulator to produce a strategy that is uniform w.r.t. the one that was created for an empty mempool (where no information about honest user transactions is available).

\section{Soundness Proof}

\label{appendix:tool-soundness}
In this section, we prove the soundness of the proposed algorithm for synthesizing \intConds{}.

\subsection{Preliminaries}

In the following, we are concerned with a concrete contract $\contract$ and will denote transactions invoking this contract with $\ctx{\contract}$. 
Further, we will use $\cstate{\conf}{\contract}$ to refer to the state of contract $\contract$ in configuration $\conf$ (which maps each contract variable to its value) and will use $\cvars{\contract}$ to refer to the set of contract variables of $\contract$ (constituting the domain of $\cstate{\conf}{\contract}$)
and $\cFuncs{\contract}$ to denote the set of functions of $\contract$.
We will use $\cstate{\conf}{\contract} \equivfor{\varset} \cstate{\tick{\conf}}{\contract}$ to denote that the states $\cstate{\conf}{\contract}$ and $\cstate{\tick{\conf}}{\contract}$ coincide on all variables from $\varset \subseteq \cvars{\contract}$.

We will assume that the observables produced by the semantics record the values of a specific \emph{event variable} (at the end of transaction execution) and are of the form $\obsc{v}$ where $\eventvar \in \criticalvars^{\contract}$ is the event variable and $v$ is its value.
For the sake of simplicity, we will assume every transaction emits at most one observable.
Note that, effectively, we are not losing generality with this since we can simply assume that there is a specific (local) event variable $\eventvar_{\procF} \in \criticalvars^{\contract} = \{ \eventvar_{\procF} ~|~ \procF \in \cFuncs{\contract} \}$ for each function, which contains the list of all events emitted during execution of the function.
In particular, this means that
if the execution of a transaction $\ctx{\contract} =  \contrCall{\procF(\fargs)}$ that writes function $\procF$ of contract $\contract$ produces an observable $\obsone{\eventvar_{\procF}}{v}$, then the transaction execution changes the value of the event variable $\eventvar_{\procF}$ to $v$.
Correspondingly, every change of the event variable $\eventvar_{\procF}$ results in an observable $\obsone{\eventvar_{\procF}}{v}$ where $v$ is the new value of $\eventvar_{\procF}$.
We make this assumption more formal in Section~\ref{subsec:conditions-correctness}.

\subsubsection{Symbolic Execution}

In the following, we will assume the existence of sound and complete symbolic execution for smart contracts. 
To this end, we will assume a set $\expset$ of programming expressions over the extended variables $\seCvars{\contract} = \gvars{\contract} \cup \lvars{\funvar}{\indexvar}$ of a contract $\contract$.
The extended variables $\seCvars{\contract}$ consist of the global variables $\gvars{\contract}$ that contain all contract variables $\cvars{\contract}$ and the variable to access the block number ($\blocknumvar$), as well as the contract's local variables $\lvars{\funvar}{\indexvar}$ containing the variables of the function argument and sender for a function $\funvar$.
Note that we assume all variables in $\lvars{\funvar}{\indexvar}$ to be indexed with an index $\indexvar$ to allow for correct scoping during the analysis.
In particular, this will allow us to distinguish between different transactions invoking the same function (and therefore assigning different values to the local variables, such as function arguments).
Note that we also write $\config \equivfor{\varset} \tick{\config}$ for $\varset \subseteq \gvars{\contract}$
provided that $\cstate{\conf}{\contract} \equivfor{\varset \setminus \{ \blocknumvar \}} \cstate{\tick{\conf}}{\contract}$
and if $\blocknumvar \in \varset$ also $\conf.\blocknumber = \tick{\conf}.\blocknumber$.

\begin{remark}
For the sake of simplicity, we consider the block number $\blocknumvar$ as sole global variables. 
Realistic smart contract languages may also provide access to other global variables such as the block timestamp or the state of the wallets of other users. 
It is easy to extend our analysis to such other global variables, which will need to be considered fully attacker-controlled. This can be done by synthesizing the condition $\bot$ when encountering dependencies of critical variables to any such global variable different from $\bot$.
\end{remark}

We will write $\expvarsi{\seExp}{\indexvar}$ (where $\expvars{\seExp} \subseteq \seCvars{\contract}$) to denote the set of variables occurring in expression $\seExp \in \expset$.
Formally, we will assume that for each function $\procF$ of a smart contract $\contract$ that there exists a relation $\symExRel{\procF}$ containing pairs of the form $\seRes{\seSub}{\seCond}$ where $\seSub \in \seCvars{\contract} \to \expset$ is a substitution from variables to expressions and $\seCond$ is a logical condition of the variables in $\seCvars{\contract}$.
We will further assume that there exists a function $\confToAss{\conf}{\txT}$ that maps a concrete configuration $\conf$ and a transaction $\txT$ to a corresponding assignment $\seAss$ over the extended variables $\seCvars{\contract}$.
In particular, we will assume that for the assignment of global variables, $\confToAss{\conf}{\txT}$ is agnostic to the transaction $\txT$. So, more formally: 

\begin{assumption}
Let $\config$ be a configuration, $\txv_1$ and $\txv_2$ be transactions, $\contract$ be a contract
and $x \in \gvars{\contract}$ be a variable. Then for any index $i$
$$\confToAss{\conf}{\txv_1}(x) = \confToAss{\conf}{\txv_2}(x)$$
\end{assumption}

In the following, we write $\evAss{\seAss}{\seExp}$ to denote the evaluation of expression $\seExp \in \expset$ under assignment $\seAss$.

In the remainder of the paper, we will assume the symbolic execution relations to be sound and complete.

\begin{assumption}[Soundness of symbolic execution]
    Let $\contract$ be a contract and $\procF \in \cFuncs{\contract}$.
Then the symbolic execution relation $\symExRel{\procF}$ for function $\procF$ and indexed by some index $i$, satisfies the following soundness property:
\begin{align*}
    \forall \txT ~\procF ~\fargs.&~ \txT = \contrCall{\procF(\fargs)} \nonumber \\
    &\Rightarrow \forall \conf ~\seSub ~ \seCond.~ \seRes{\seSub}{\seCond} \in \symExRel{\procF} ~\land~ \satis{\confToAss{\conf}{\txT}}{\seCond} \nonumber \\
    &\Rightarrow \exists \conf'.~ \conf \trans[]{\txT} \conf'~\land~ 
    \confToAss{\conf'}{\txT} = \assComp{\confToAss{\conf}{\txT}}{\seSub}
\end{align*}
\label{eq:se-soundness}
where $\assComp{\seAss}{\seSub}$ denotes the assignment obtained by first applying the substitution $\seSub$ and then evaluating the resulting expression under $\seAss$ (so $\assComp{\seAss}{\seSub}(x) = \evAss{\seAss}{\seSub(x)}$).
Note that, in particular, it holds that $\evAss{\assComp{\seAss}{\seSub}}{\seExp} = \evAss{\seAss}{\seExp \seSub}$ where $\seExp \seSub$ denotes the expression obtained by applying the substitution $\seSub$ to the expression $e \in \expset$. 
\end{assumption}
This statement expresses the soundness of the symbolic execution relation $\symExRel{\procF}$: for each concrete configuration $\conf$ and transaction $\txT$ invoking function $\procF$, if there exists a symbolic execution result $\seRes{\seSub}{\seCond}$ such that the condition $\seCond$ is satisfied under the assignment corresponding to $\conf$ and $\txT$, then there exists a concrete transition executing $\txT$ from $\conf$ to some configuration $\conf'$ such that the assignment obtained from $\conf'$ corresponds to the application of the substitution $\seSub$ to the assignment corresponding to $\conf$ and $\txT$.

We define completeness of the symbolic execution relation $\symExRel{\procF}$ correspondingly:

\begin{assumption}[Completeness of symbolic execution]
    Let $\contract$ be a contract and $\procF \in \cFuncs{\contract}$.
Then the symbolic execution relation $\symExRel{\procF}$ for function $\procF$ and indexed by some index $i$, satisfies the following completeness property:
\begin{align*}
    \forall \conf~ \conf'~ \txT~ \fargs.&~ \conf \trans[]{\txT} \conf' ~\land~ \txT = \contrCall{\procF(\fargs)} \nonumber \\
    & \Rightarrow \exists \seSub ~ \seCond.~ \seRes{\seSub}{\seCond} \in \symExRel{\procF} ~\land~ \satis{\confToAss{\conf}{\txT}}{\seCond} \nonumber\\ 
    & \qquad \qquad ~\land~\confToAss{\conf'}{\txT} = \assComp{\confToAss{\conf}{\txT}}{\seSub}
\end{align*}
 \label{eq:se-completeness}
\end{assumption}

This statement expresses that for each concrete execution of transaction $\txT$ invoking function $\procF$ from configuration $\conf$ to configuration $\conf'$, there exists a symbolic execution result $\seRes{\seSub}{\seCond}$ such that the condition $\seCond$ is satisfied under the assignment corresponding to $\conf$ and $\txT$, 
and the assignment corresponding to $\conf'$
corresponds to the application of the substitution $\seSub$ to the assignment corresponding to $\conf$ and $\txT$.

Note that we also can assume the following property:

\begin{assumption}[Agreement on condition variables]
    \label{lem:agreement-rel-vars}
    Let $\condVar$ be a condition and $\seAss_1$ and $\seAss_2$ be assignments such that for all $\seCvar \in \condvars{\condVar}$
    it holds that $\seAss_1(\seCvar) = \seAss_2(\seCvar)$.
    Then it holds that 
    \begin{align*}
    \satis{\seAss_1}{\condVar} \Leftrightarrow \satis{\seAss_2}{\condVar}
\end{align*}
\end{assumption}

This property states that if two assignments agree on the variables occurring in a condition, then they equivalently satisfy the condition.
To avoid a full formalization of the concrete formula language with expressions, we assume the correctness of this lemma from now on.

Finally, since we consider transaction execution to be deterministic, we will also require the symbolic execution relation to satisfy a strong determinism property, namely that if two path conditions can be satisfied by the same assignment, then they will also be mapped to the same expression.

\begin{assumption}[Determinism of symbolic execution]
Let $\contract$ be a contract and $\procF \in \cFuncs{\contract}$.
Then the symbolic execution relation $\symExRel{\procF}$ for function $\procF$ and indexed by some index $i$, satisfies the following  property:
\begin{align*}
\forall \seRes{\seSub_1}{\seCond_1}, \seRes{\seSub_2}{\seCond_2} \in \symExRel{\procF}.~ &
(\exists \seAss. \satis{\seAss}{\seCond_1} ~\land~ \satis{\seAss}{\seCond_2}) \\
&\Rightarrow \forall \seCvar \in \seCvars{\contract}. \seSub_1(\seCvar) = \seSub_2(\seCvar)
\end{align*}
\end{assumption}

\begin{remark}
This property is not directly a consequence from the soundness of the symbolic execution relation and the determinism of the smart contract semantics because technically a sound symbolic execution could still map path conditions that can be satisfied to the same assignment to syntactically different but semantically equivalent expressions.
\end{remark}

\subsubsection{Dependency Analysis}

We define the data and control flow dependencies of variables in a contract $\contract$ in terms of the symbolic execution relation $\symExRel{\procF}$. %

\begin{definition}[Data Flow Dependency]
    Let $\contract$ be a smart contract and $\procF$ be a function of $\contract$ and $\symExRel{\procF}$ be a symbolic execution relation of $\procF$.
    We define the set of data flow dependencies $\ddeps{\procF}{\seCvar} \subseteq \seCvars{\contract}$ of variable $\seCvar \in \seCvars{\contract}$ in function $\procF$ of contract $\contract$ as
    \begin{align*}
        \ddeps{\procF}{\seCvar} = \bigcup_{
            \seRes{\seSub}{\seCond} \in \symExRel{\procF}
        } \expvars{\seSub(\seCvar)}
    \end{align*}
\end{definition}

To define control flow dependencies, we first introduce the notions of \emph{path conditions} and \emph{relevant variables}. 

\begin{definition}[Path Condition]
Let $\contract$ be a smart contract and $\procF$ be a function of $\contract$ and $\symExRel{\procF}$ be a symbolic execution relation of $\procF$.
We define the path condition $\pathcond{\procF}{\seExp}{\seCvar}$ for variable $\seCvar \in \seCvars{\contract}$, expression $\seExp$ and function $\procF$ of contract $\contract$ as
\begin{align*}
    \pathcond{\procF}{\seExp}{\seCvar} := \bigvee_{\seRes{\seSub}{\seCond} \in \symExRel{\procF} : \seSub(\seCvar) = \seExp}{\seCond}
\end{align*}
\end{definition}

Intuitively, $\pathcond{\procF}{\seExp}{\seCvar}$ is the condition under which variable $\seCvar$ is assigned the expression $\seExp$ in function $\procF$ of contract $\contract$.

\begin{definition}[Relevant variables]
Let $\contract$ be a smart contract and $\condVar$ be a logical condition over variables in $\seCvars{\contract}$.
We define the relevant variables $\relVars{\condVar}$ of condition $\seCond$ as 
\begin{align*}
    \relVars{\condVar} : \{ \seCvar \in \condvars{\condVar}  ~|~ & \exists \seAss \seAss'.~ \satis{\seAss}{\condVar} 
    ~\land~ \neg \satis{\seAss'}{\condVar} \\
    &~\land~ \forall y \in \condvars{\condVar} \setminus \{ \seCvar\}.~ \seAss(y) = \seAss'(y) \}
\end{align*}
\end{definition}

Intuitively, $\relVars{\condVar}$ denotes all those variables that occur in $\condVar$, which may change the satisfiability of $\condVar$.

We now can define the control dependencies $\cdeps{\procF}{\seCvar}$  of a variable $\seCvar$ in a function $\procF$ 
as the relevant variables of the possible path conditions for assigning the variable $\seCvar$ to different expressions.

\begin{definition}[Control Flow Dependency]
    Let $\contract$ be a smart contract and $\procF$ be a function of $\contract$, and $\symExRel{\procF}$ be a symbolic execution relation of $\procF$.
    We define the set of control flow dependencies $\cdeps{\procF}{\seCvar} \subseteq \seCvars{\contract}$ of variable $\seCvar \in \seCvars{\contract}$ in function $\procF$ of contract $\contract$ as
    \begin{align*}
        \cdeps{\procF}{\seCvar} := \bigcup_{\seExp \in \{ \seExp ~|~  \seRes{\seSub}{\seCond} \in \symExRel{\procF} ~\land~ \seSub(\seCvar) = \seExp \}} \relVars{\pathcond{\procF}{\seExp}{\seCvar}} 
    \end{align*}
\end{definition}

Note that path conditions may have many irrelevant variables since they are simply the disjunctions of the conditions for different execution paths.
For defining the control dependencies, we are only interested in those variables that are pivotal for deciding which expression the variable $\seCvar$ gets assigned.

\newcommand{\depset}{D}

We can show that the dependency analysis defined above is sound with respect to concrete executions, meaning that whenever a transaction $\txT$ is executed in two configurations $\conf_1$ and $\tick{\conf_1}$ that agree on the dependencies of variables $\vSet$, then this will result in the same assignments to $\vSet$. In particular, if $\vSet$ contains the event variable $\eventvar$, then the two executions will also produce the same observables.

\begin{lemma}[Soundness of Dependency Analysis]
    Let $\contract$ be a smart contract and $\procF$ be a function of $\contract$, and $\symExRel{\procF}$ be a symbolic execution relation of $\procF$ for some index $i$.
    Let $\conf_1$ and $\tick{\conf_1}$ be two configurations and $\txT = \contrCall{\procF(\fargs)}$ be a transaction invoking function $\procF$ of contract $\contract$ and $\vSet \subseteq \cvars{\contract}$ be a set of contract variables such that $\eventvar \in \vSet$. 
    Further, let $\depset = \gvars{\contract} \cap \bigcup_{\seCvar \in \vSet}\ddeps{\procF}{\seCvar} \cup \cdeps{\procF}{\seCvar} \cup \vSet$
    and assume 
    $\conf_1 \equivfor{\depset} \tick{\conf_1}$.
    Then for all configurations $\conf_2$ and $\tick{\conf_2}$ and observables $\obs$ and $\tick{\obs}$ such that $\conf_1 \trans[\obs]{\txT} \conf_2$ and $\tick{\conf_1} \trans[\tick{\obs}]{\txT} \tick{\conf_2}$, it holds that $\conf_2 \equivfor{\vSet} \tick{\conf_2}$ and $\obs = \tick{\obs}$. 
\label{lem:dep-soundness}
\end{lemma}
\begin{proof}
This follows directly from the previous definitions as well as the soundness and completeness of symbolic execution.
\end{proof}

Note that the set of dependencies $\depset$ here is defined to (i) also contain the variables $\vSet$ itself to account for the case that $\txT$ does not change the value variables in $\vSet$ and (ii) only considers those dependencies within the global variables $\gvars{\contract}$. This last point accounts for the fact that the dependencies of $\seCvar$ may also contain local variables (denoting the local transaction execution context, like the sender of the transaction, or the arguments provided to the called function). The values of these local variables are fully determined by the transaction $\txT$, so when considering two executions of $\txT$, then the local variables will be assigned to the same values.

\subsubsection{Time Progress}
We will, in the following, assume the existence of a specific function $\procTick$. 
The function $\procTick$ models the specific transaction $\txB$ that increases the block number and is used to model the passage of time in the blockchain.
We model $\procTick$ as a function to enable a unified treatment within the symbolic execution framework.
We can consider the $\procTick$ function simply as a special function, with a well-defined semantics that can be freely scheduled by the attacker and that only modifies the block number variable by incrementing it by one and does not modify any other variable.

The corresponding symbolic execution relation $\symExRel{\procTick}$ is defined as follows: 
\begin{definition}[Symbolic Execution of $\procTick$]
\begin{align*}
    \symExRel{\procTick} := \{ \seRes{\seAss_\procTick}{\top}\}
\end{align*}
with $\seAss_\procTick$ defined by
\begin{align*}
\seAss_\procTick(\seCvar) := 
\begin{cases}
    \blocknumvar + 1 & \seCvar = \blocknumvar \\
    \seCvar & \text{otherwise}
\end{cases}
\end{align*}
\end{definition}

Note that $\blocknumvar$ denotes the variable for the block number and is contained in $\seCvars{\contract}$ for any contract $\contract$.
It is straightforward to show that the symbolic execution for $\procTick$ is sound and complete for the semantics of the transaction $\txB$ that increases the block number.

\begin{lemma}[Completeness of $\symExRel{\procTick}$]
    Let $\contract$ be an arbitrary contract.
\begin{align*}
    \forall \conf~ \conf'.~ \conf \trans[\obs]{\txB} \conf'
    & \Rightarrow \exists \seSub ~ \seCond.~ \seRes{\seSub}{\seCond} \in \symExRel{\procTick} ~\land~ \satis{\confToAss{\conf}{\txB}}{\seCond} \nonumber\\ 
    & \qquad \qquad ~\land~ \cblocknumber{\conf'} = (\assComp{\confToAss{\conf}{\txB}}{\seSub})(\blocknumvar)
\end{align*}
\end{lemma}

\begin{lemma}[Soundness of $\symExRel{\procTick}$]
    Let $\contract$ be an arbitrary contract.
\begin{align*}
    &\forall \conf ~\seSub ~ \seCond.~ \seRes{\seSub}{\seCond} \in \symExRel{\procTick} ~\land~ \satis{\confToAss{\conf}{\txB}}{\seCond} \\
    & \qquad \Rightarrow \exists \conf'.~ \conf \trans[]{\txB} \conf'~\land~ \cblocknumber{\conf'} = (\assComp{\confToAss{\conf}{\txB}}{\seSub})(\blocknumvar)
\end{align*}
\end{lemma}

 \subsubsection{Round-based user strategies}
In the following, we will consider user strategies that proceed in rounds of $k$ blocks. 
More precisely, these user strategies decide upon a single contract transaction $\ctx{\contract}$ to schedule at the beginning of each round and then keep scheduling it for the following $k$ blocks. 
From now on, we will treat the parameter $k$, which denotes the length of a round, as a global parameter.
We will further use notation 
$\lastBlockNumber(\run)$ to denote the block number $\lastconfig{\run}.\blocknumber$ of the last configuration $\lastconfig{\run}$ of a run $\run$. 
We define $\round(\config) := \config.\blocknumber - (\config.\blocknumber \bmod{k})$ to denote the round of a configuration $\config$ and correspondingly 
$\round(\run) := \round(\lastconfig{\run})$. 

The idea behind enforcing strategies to proceed in rounds is that they enable users to ensure that their transactions get processed before they decide upon the next transaction to schedule. 
This, in particular, ensures that an attacker cannot frontrun users with their own transactions. %

\begin{definition}[Round-based user strategies]
    Let $\rstrat \in \WalU \times \ContrU \times \mathbb{N} \rightarrow \listtype{\TxU}$ (see Sec. \ref{subsec:blockchain-model}) be an honest user strategy operating on configurations $\config = (\WmvA, \contracts, \blocknumber)$ and returning a list of transactions $\vec{\txv} = \rstrat(\config)$ such that $\length{\vec{\txv}} \leq 1$ for all configurations $\config$. 
    Then we define the round-based strategy derived from $\rstrat$ (written $\roundstrat{\rstrat}$) as follows: 
    {\small
    \begin{align*}
      \roundstrat{\rstrat} (\run) 
       &:= \{ \txv \in \rstrat(\config_{\run_1}) \backslash  \getTxs(\run_2) \mid \\  
       & \run = \concat{\run_1}{\run_2} 
        ~\land~ 
        \length{\getBlocks(\run_1)} = \round(\run)
        \}
    \end{align*}}
    where 
    $\getBlocks(\run)$ returns the blocks of run $\run$ and 
    $\getTxs(\run)$ the transactions of $\run$.
\end{definition}

Intuitively, a round-based user strategy $\roundstrat{\rstrat}$ simply applies its round strategy $\rstrat$ at the beginning of each round and keeps on scheduling the transactions as mandated by $\rstrat$ (with exception of those that have already been successfully included in the blockchain) until the end of the round. 

Given a set of configuration-based strategies $\{ \rstrat^i \}$, we can canonically lift it to arrive at a set $\{ \roundstrat{{\rstrat^{i}}} \}$ of round-based honest user strategies.

\subsection{Soundness of \intCond{} synthesis}\label{appendix:cond:cond}

We now define the soundness of the \intCond{} synthesis algorithm.
Note that in the following, we may omit proofs of simple lemmas that are direct implications of existing definitions or previously established results.
For stating the algorithm, we will make use of the following helper definitions:

\begin{definition}[No Write Condition]
Let $\contract$ be a smart contract and $\procF$ be a function of $\contract$ and $\symExRel{\procF}$ be a symbolic execution relation of $\procF$.
We define the \emph{no write condition} $\nowritecond{\procF}{\seCvar}$ for variable $\seCvar \in \seCvars{\contract}$ in function $\procF$ of contract $\contract$ as
\begin{align*}
    \nowritecond{\procF}{\seCvar} := \pathcond{\procF}{\seCvar}{\seCvar}
\end{align*}
\end{definition}
Intuitively, $\nowritecond{\procF}{\seCvar}$ is the condition under which variable $\seCvar$ is not modified in function $\procF$ of contract $\contract$.

\paragraph{Round Invariants}

We define the notion of a round invariant.
Intuitively, a pair $(\invpre, \inv)$ denotes a round invariant if the fact that $\invpre$ holds at the beginning of a round implies that $\inv$ holds throughout the whole round.

More formally, we define this property as follows:
\begin{definition}[Round Invariant]
Let $k$ be the length of a round. 
We say that a condition $\inv$ is a round invariant under precondition $\invpre$ (written $\roundinv{\roundvar}{\invpre}{\inv}$) if the following holds
\begin{align*}
    \forall \roundvar \in \mathbb{N}.~\forall \seAss.&~
    \satis{\seAss}{\invpre}
    ~\land~ \seAss(\blocknumvar) = r k \\
&\Rightarrow \forall i \in [0, \dots, k-1].~\satis{\assComp{\seSubst{\seAss(\blocknumvar)}{\blocknumvar+i}}{\seAss}}{\inv}
\end{align*}
Here $\blocknumvar \in \gvars{\contract}$ denotes the variable of the block number.
\label{def:round-invariant}
\end{definition}

We will additionally make use of the notion of a \emph{block-oblivious invariant}.
Intuitively, a logical condition $\inv$ is a block-oblivious invariant w.r.t. a function $\procF$ if $\inv$ still holds after applying the possible effects that could have resulted from executing $\procF$ in an arbitrary block of the current round.
More formally 

\begin{definition}[Block-oblivious Invariant]
We say that a condition $\inv$ is a block-oblivious invariant of a function $\procF$ under a condition $\preCond$ (written $\execinv{\procF}{\preCond}{\inv}$) if the following holds
\begin{align*}
\forall \seAss, \seSub, \seCond, b. 
& \round(b) = \round(\seAss(\blocknumvar)) ~\land~ \seRes{\seSub}{\seCond} \in \symExRel{\procF}
~\land~ \satis{\seAss}{\inv} \\
&~\land~ \satis{\assComp{\seSubst{\blocknumvar}{b}}{\seAss}}{\seCond} \\
&~\land~ \satis{\assComp{\seSubst{\blocknumvar}{b}}{\seAss}}{\preCond} \\
& \Rightarrow \satis{\assComp{\seSubst{\blocknumvar}{\seAss(\blocknumvar)}}{(\assComp{\seSub}{(\assComp{\seSubst{\blocknumvar}{b}}{\seAss})})}}{\inv}
\end{align*}
We use $\round(b) = b - (b \bmod{k})$ to denote the round of a block number.
\label{def:bo-inv}
\end{definition}

Intuitively, this definition states that if an assignment $\seAss$ satisfies the invariant $\inv$ and this assignment would satisfy the prerequisites for execution $\procF$ in some block $b$ of the same round (denoted by satisfaction of the path condition $\seCond$ and the precondition $\preCond$), then performing the same updates to $\seAss$ would leave the invariant intact.

In particular, we will use the previously stated properties to make use of the following lemma: 
\begin{lemma}[Preservation of invariants through rounds]
Let $\inv$ and $\invpre$ be conditions such that $\roundinv{\roundvar}{\invpre}{\inv}$.
Let $\config_1$, $\config_2$, \dots, $\config_n$ be configurations such that 
$\round(\config_j) = \round(\config_l)$ for $j, l \in \{ 1, \dots, n \}$ and $\lastBlockNumber(\config_1) \bmod{k} = 0$.
Further, let 
$\config_1 \trans{\txv_1} \config_2 \trans{\txv_2} \dots \trans{\txv_{n-1}} \config_{n}$
for transaction $\txv_1, \dots \txv_{n-1}$ such that for all $j \in \{1, \dots n-1 \}$, it holds that
either $\txv_j = \txB$ or 
$\txv_j = \contrCall{\procF(\fargs)}$ for some $\procF \in \cFuncs{\contract}$ s.t. 
$\execinv{\procF}{\preCond}{\invpre}$ holds for some index $\indexvar$ and $\preCond$ with $\satis{\confToAss{\config_j}{\txv_j}}{\preCond}$.
Then if $\satis{\confToAss{\config_1}{\txv}}{\invpre}$ for some $\txv$ and index $\indexvar$ it also holds
that $\satis{\confToAss{\config_n}{\txv}}{\inv}$. 
\label{lem:inv-preservation}
\end{lemma}
\begin{proof}
To prove this lemma, we first show by induction over the length of $\vec{\txv}$ that if $\satis{\confToAss{\config_1}{\txv}}{\invpre}$
then also $\satis{\assComp{\seSubst{\blocknumvar}{\confToAss{\config_1}{\txv}(\blocknumvar)}}{\confToAss{\config_n}{\txv}}}{\invpre}$. 
The cases follow immediately by the Definition~\ref{def:bo-inv} (and the soundness and completeness of the symbolic execution). 
Finally, we can use Definition~\ref{def:round-invariant} to conclude from $\satis{\assComp{\seSubst{\blocknumvar}{\confToAss{\config_1}{\txv}(\blocknumvar)}}{\confToAss{\config_n}{\txv}}}{\invpre}$ that also $\satis{\confToAss{\config_n}{\txv}}{\inv}$ (since $\confToAss{\config_n}{\txv} \allowbreak = \assComp{\seSubst{\blocknumvar}{\confToAss{\config_n}{\txv}(\blocknumvar)}}{(\assComp{\seSubst{\blocknumvar}{\confToAss{\config_1}{\txv}(\blocknumvar)}}{\confToAss{\config_n}{\txv}})}$).
\end{proof}

Intuitively, this lemma states that if $(\invpre, \inv)$ is a round invariant and if $\invpre$ holds at the beginning of a round, then $\inv$ holds throughout the whole round as long as only transactions get executed for which $\invpre$ is a block-oblivious invariant.

\paragraph{Constructing Interaction Conditions}
Our algorithms will iteratively build pairs of conditions $(\invpre, \inv)$ that together constitute round invariants. 
To this end, we will define two functions $\getRoundInv$ and $\mergeRoundInv$ such that for a condition $\condVar$, $\getRoundInv(\condVar)$ returns a round invariant $(\invpre, \inv)$ such that $\inv$ is equivalent to $\condVar$ ($\inv \equiv \condVar$) and $\mergeRoundInv((\invpre', \inv'), (\tick{\invpre}, \tick{\inv}))$ returns a round invariant $(\invpre, \inv)$ such that $\inv \equiv \inv' \land \tick{\inv}$.

For the theoretical treatment, we will axiomatize the properties of $\getRoundInv$ and $\mergeRoundInv$ as follows:

\begin{assumption}[Round Invariant Generation]
Let $\condVar$ be a logical condition (over $\seCvars{\contract}$ for some $\procF \in \cFuncs{\contract}$ and index $\indexvar$). 
Further, let $(\invpre, \inv) = \getRoundInv(\condVar)$ for logical conditions $\invpre$, $\inv$ (over $\seCvars{\contract}$). 
Then, the following hold:
\begin{itemize}
\item $\roundinv{\roundvar}{\invpre}{\inv}$
\item $\forall \seAss.~ \satis{\seAss}{\inv} \Leftrightarrow \satis{\seAss}{\condVar}$
\end{itemize}
\end{assumption}

Note that $\getRoundInv$ is guaranteed to exist, since $\getRoundInv(\condVar) := (\bot, \condVar)$ is a trivial implementation that satisfies the above properties. 
In practice, we can implement $\getRoundInv$ by using heuristics to compute more interesting round preconditions $\invpre$, which we then check to indeed satisfy the desired properties.
E.g., in cases where $\condVar$ does not contain the block number $\blocknumvar$, $\getRoundInv(\condVar) := (\condVar, \condVar)$ provides relevant round invariants. 
Simple strengthenings that we attempt in case that $\condVar$ contains the block number $\blocknumvar$ is to substitute $\blocknumvar$ in $\condVar$ with $\blocknumvar + k - 1$.

\begin{assumption}[Round Invariant Composition]
Let $\invpre^1$, $\invpre^2$, $\inv^1$, $\inv^2$  be logical conditions (over $\seCvars{\contract}$ for some $\procF \in \cFuncs{\contract}$ and index $\indexvar$) such that
 $\roundinv{\roundvar}{\invpre^1}{\inv^1}$ and $\roundinv{\roundvar}{\invpre^2}{\inv^2}$.
Further, let $(\invpre, \inv) = \mergeRoundInv((\invpre^1, \inv^1), (\invpre^2, \inv^2))$ 
for logical conditions $\invpre$, $\inv$  (over $\seCvars{\contract}$). 
Then, the following hold:
\begin{itemize}
\item $\roundinv{\roundvar}{\invpre}{\inv}$
\item $\forall \seAss.~ \satis{\seAss}{\inv} \Leftrightarrow \satis{\seAss}{\inv^1 ~\land~ \inv^2}$
\end{itemize}
\end{assumption}

Note that a trivial implementation of $\mergeRoundInv$ would be 
$$\mergeRoundInv((\invpre^1, \inv^1), (\invpre^2, \inv^2)) := (\invpre^1 \land \invpre^2, \inv^1 \land \inv^2)$$
In practice, we perform simplifications to reduce the number of variables occurring in the final round invariant $(\invpre, \inv)$.

We now devise algorithms for computing interaction conditions that cover two different scenarios
\begin{itemize}
\item \emph{Frontrunning conditions} $\frWriteCond{\procF}{\seExp}{\seCvar}{\roundPred}$ that ensure that a call of honest user $\honuser$ to function $\procF$ will deterministically assign the value of variable $\seCvar$ according to expression $\seExp$. In particular, this means that concurrent attacker actions may neither deviate the control flow of the execution so that the execution will assign another expression $\seExp' \neq \seExp$ to $\seCvar$, nor may attacker actions change the variables in $\seExp$ so that the execution of $\procF$ would lead to a different assignment for $\seCvar$.
\item \emph{Backrunning conditions} $\brWriteCond{\procF}{\vSet}{\seCvar}{\roundPred}$ that ensure that when an honest user $\honuser$ writes a variable $\seCvar$ using function $\procF$, then this will not interfere with how the attacker can concurrently assign variables in the set $\vSet$. In particular, this means that the attacker may not perform concurrent actions in which the assignment of variables in $\vSet$ may depend on the variable $\seCvar$, either directly (via data dependencies) or indirectly (via control dependencies).
\end{itemize}

\begin{algorithm}[t]
    \caption{Synthesis of condition $\frWriteCond{\procF}{\seExp}{\seCvar}{\roundPred}$}
    \begin{algorithmic}[1]
\REQUIRE Smart contract $\contract$, user $\honuser$, function $\procF \in \cFuncs{\contract}$, symbolic execution relations 
$\{ \symExRel{\procG} \}_{\procG \in \cFuncs{\contract} \cup \{ \procTick \}, \indexvar \in \{ \userindexvar, \nuserindexvar\} }$ over the respective $\seCvarsi{\contract}{\procG}{\indexvar}$, 
variable $\eventvar \in \gvars{\contract}$, 
expression $\seExp \in \expseti{\procF}{\userindexvar}$
\STATE $\frWriteCond{\procF}{\seExp}{\seCvar}{\roundPred} \gets \getRoundInv(\pathcondi{\procF}{\seExp}{\seCvar}{\userindexvar})$
\STATE $\nowriteset \gets \emptyset$, $\dSet \gets \emptyset$
\IF{$\seExp \neq \seCvar$}
    \STATE $\dSet \gets \ddepsi{\procF}{\seCvar}{\userindexvar} \cup \{ \seCvar \}$
\ENDIF
\FORALL{$y \in \dSet$, $\procG \in \cFuncs{\contract} \cup \{ \procTick \}$} 
            \STATE $\varphi \gets \setforall{\locvar}{\lvars{\procG}{\nuserindexvar}} \sendervari{\nuserindexvar} \neq \honuser \rightarrow \nowritecondi{\procG}{y}{\nuserindexvar}$
            \STATE $\frWriteCond{\procF}{\seExp}{\seCvar}{\roundPred} \gets 
            \mergeRoundInv(
                \frWriteCond{\procF}{\seExp}{\seCvar}{\roundPred},
                \getRoundInv \left(
                    \varphi
                \right)
            )$
            \STATE $\nowriteset \gets \nowriteset \cup \{(y, \procG)\}$
    \ENDFOR
\STATE $\doneflag \gets \false$
    \WHILE{$\condvars{\frWriteCond{\procF}{\seExp}{\seCvar}{\roundPred}} \times \cFuncs{\contract} \not \subseteq \nowriteset \land \doneflag = \false$}
        \STATE $\doneflag \gets \true$
        \FORALL{$\procG \in \cFuncs{\contract}$}
            \STATE $(\invpre, \inv) \gets \frWriteCond{\procF}{\seExp}{\seCvar}{\roundPred}$
            \IF{$\execinvi{\procG}{\predIsNoSender{\honuser}}{\invpre}{\nuserindexvar}$} \label{line:algo-fr:inv-check}
                \STATE \textbf{continue}
            \ENDIF
        \STATE $\doneflag \gets \false$
        \STATE $y \gets \condvars{\frWriteCond{\procF}{\seExp}{\seCvar}{\roundPred}} \setminus \{ z ~|~ (z, \procG) \in \nowriteset \}$
        \STATE $\varphi \gets \setforall{\locvar}{\lvars{\procG}{\nuserindexvar}} \sendervari{\nuserindexvar} \neq \honuser \rightarrow \nowritecondi{\procG}{y}{\nuserindexvar}$
        \STATE $\frWriteCond{\procF}{\seExp}{\seCvar}{\roundPred} \gets 
        \mergeRoundInv (
                \frWriteCond{\procF}{\seExp}{\seCvar}{\roundPred},
                \getRoundInv\left(
                    \varphi
                \right)
            )$
        \STATE $\nowriteset \gets \nowriteset \cup \{(y, \procG)\}$
        \ENDFOR
    \ENDWHILE
    \RETURN $\frWriteCond{\procF}{\seExp}{\seCvar}{\roundPred}$
    \end{algorithmic}
    \label{algo:fr-app}
\end{algorithm}

Algorithm~\ref{algo:fr-app} for generating $\frWriteCond{\procF}{\seExp}{\seCvar}{\roundPred}$ proceeds as follows: 
Starting from the path condition for assigning $\seCvar$ to expression $\seExp$, it refines the condition by 
(1) excluding writes (by the attacker) to all variables on which $\seCvar$ has a data dependency (provided that $\seExp$ indicates a proper assignment to $\seCvar$);
(2) excluding writes (by the attacker) to all variables on which $\seCvar$ has control dependencies until an invariant is found that holds for all possible attacker interactions.
Note that the algorithm is guaranteed to terminate, since if not terminating before (due to successfully finding an invariant), eventually $\nowriteset$ will contain all combinations of contract variables and contract functions, causing a termination of the outer while loop.

\begin{algorithm}[t]
    \caption{Synthesis of condition $\brWriteCond{\procF}{\vSet}{\seCvar}{\roundPred}$}
    \begin{algorithmic}[1]
        \REQUIRE Smart contract $\contract$, user $\honuser$, function $\procF \in \cFuncs{\contract} \cup \{ \procTick \}$, symbolic execution relations 
$\{ \symExRel{\procG} \}_{\procG \in \cFuncs{\contract} \cup \{ \procTick \}, \indexvar \in \{ \userindexvar, \nuserindexvar\} }$ over the respective $\seCvarsi{\contract}{\procG}{\indexvar}$, 
variable $\eventvar \in \gvars{\contract}$, 
set $\vSet \subseteq \gvars{\contract}$ %
\STATE $\brWriteCond{\procF}{\vSet}{\seCvar}{\roundPred} \gets (\top, \top) $
\STATE $\dSet \gets \{ (y, \procG) \in \vSet \times \cFuncs{\contract} ~|~ \seCvar \in \ddepsi{\procG}{y}{\nuserindexvar} ~\lor~ \seCvar = y \}$ %
\FORALL{$(y, \procG) \in \dSet$} 
            \STATE $\varphi \gets \setforall{\locvar}{\lvars{\procG}{\nuserindexvar}} \sendervari{\nuserindexvar} \neq \honuser \rightarrow \nowritecondi{\procG}{y}{\nuserindexvar}$
            \STATE $\brWriteCond{\procF}{\vSet}{\seCvar}{\roundPred} \gets 
            \mergeRoundInv(
                \brWriteCond{\procF}{\vSet}{\seCvar}{\roundPred},
                \getRoundInv\left(
                    \varphi
                \right)
            )$
            \STATE $\nowriteset \gets \nowriteset \cup \{(y, \procG)\}$
    \ENDFOR
\STATE $\cSet \gets \{ (y, \procG) \in \vSet \times \cFuncs{\contract} ~|~ \seCvar \in \cdeps{\procG}{y} \}$
\FORALL{$(y, \procG) \in \cSet$} 
    \STATE $\seExp \gets \select(\{ \seExp ~|~ \seRes{\seSub}{\seCond} \in \symExReli{\procG}{\nuserindexvar} ~\land~ \seSub(y) = \seExp \})$
    \STATE $\varphi \gets \setforall{\locvar}{\lvars{\procG}{\nuserindexvar}} \sendervari{\nuserindexvar} \neq \honuser \rightarrow \pathcondi{\procG}{\seExp}{y}{\nuserindexvar}$
    \STATE $\brWriteCond{\procF}{\vSet}{\seCvar}{\roundPred} \gets 
                \mergeRoundInv(
                    \brWriteCond{\procF}{\vSet}{\seCvar}{\roundPred},
                    \getRoundInv\left(
                        \varphi
                    \right)
                )$ 
    \ENDFOR
\STATE $\doneflag \gets \false$
\WHILE{$\condvars{\brWriteCond{\procF}{\vSet}{\seCvar}{\roundPred}} \times \cFuncs{\contract} \not \subseteq \nowriteset \land \doneflag = \false$}
    \STATE $\doneflag \gets \true$
    \STATE $(\invpre, \inv) \gets \brWriteCond{\procF}{\vSet}{\seCvar}{\roundPred}$
        \IF{$\neg \execinvi{\procF}{\predIsSender{\honuser}}{\invpre}{\userindexvar}$}
            \STATE $\doneflag \gets \false$
            \STATE $y \gets \condvars{\brWriteCond{\procF}{\vSet}{\seCvar}{\roundPred}} \setminus \{ z ~|~ (z, \procF) \in \nowriteset \}$
            \STATE $\brWriteCond{\procF}{\vSet}{\seCvar}{\roundPred} \gets 
                \mergeRoundInv(
                    \brWriteCond{\procF}{\vSet}{\seCvar}{\roundPred},$ \\
                    \hspace{4\algorithmicindent}$\getRoundInv(\nowritecondi{\procF}{y}{\userindexvar})
                )$
            \STATE $\nowriteset \gets \nowriteset \cup \{(y, \procF)\}$
        \ENDIF
        \FORALL{$\procG \in \cFuncs{\contract}$}
        \STATE $(\invpre, \inv) \gets \brWriteCond{\procF}{\vSet}{\seCvar}{\roundPred}$
            \IF{$\execinvi{\procG}{\predIsNoSender{\honuser}}{\invpre}{\nuserindexvar}$}
                \STATE \textbf{continue}
            \ENDIF
        \STATE $\doneflag \gets \false$
        \STATE $y \gets \condvars{\brWriteCond{\procF}{\vSet}{\seCvar}{\roundPred}} \setminus \{ z ~|~ (z, \procG) \in \nowriteset \}$
        \STATE $\varphi \gets \setforall{\locvar}{\lvars{\procG}{\nuserindexvar}} \sendervari{\nuserindexvar} \neq \honuser \rightarrow \nowritecondi{\procG}{y}{\nuserindexvar}$
        \STATE $\brWriteCond{\procF}{\vSet}{\seCvar}{\roundPred} \gets 
            \mergeRoundInv(
                \brWriteCond{\procF}{\vSet}{\seCvar}{\roundPred},
                \getRoundInv\left(
                    \varphi
                \right)
            )$
        \STATE $\nowriteset \gets \nowriteset \cup \{(y, \procG)\}$
        \ENDFOR
\ENDWHILE
    \RETURN $\brWriteCond{\procF}{\vSet}{\seCvar}{\roundPred}$
    \end{algorithmic}
    \label{algo:br}
    \end{algorithm}

Algorithm~\ref{algo:br} for computing $\brWriteCond{\procF}{\vSet}{\seCvar}{\roundPred}$ proceeds similarly to the previous one.
It starts with the condition $\top$ and first excludes writes to all variables $y$ that have a data dependency on $\seCvar$.
For variables $y$ that have a control dependency on $\seCvar$, it is ensured that the condition implies a path condition of a specific expression $\seExp$, ensuring that writing $\seCvar$ (using $\procF$) can not change the way how a(nother) function $\procG$ writes $y$.
Finally, writes to all variables of the conditions are prohibited until establishing an invariant both for adversarial function executions as well as the execution of $\procF$ by the honest user.

\subsubsection{Characteristic Properties}

We state the characteristic properties of the frontrunning and backrunning conditions.

We first provide the characteristic properties of the frontrunning conditions:

\begin{lemma}[Characteristic properties of the frontrunning condition]
    Let $\contract$ be a smart contract, $\honuser \in \honusers$, \break
     $\procF \in \cFuncs{\contract}$, 
$\{ \symExRel{\procG} \}_{\procG \in \cFuncs{\contract} \cup \{ \procTick \}, \indexvar \in \{ \userindexvar, \nuserindexvar\} }$ symbolic execution relations over the respective $\seCvarsi{\contract}{\procG}{\indexvar}$, 
 $\eventvar \in \gvars{\contract}$, 
$\seExp \in \expseti{\procF}{\userindexvar}$. 
 Let $\frWriteCond{\procF}{\seExp}{\seCvar}{\roundPred} = (\frCondPre, \frCondInv)$ be the result of Algorithm~\ref{algo:fr}.
The following properties hold
\begingroup%
\allowdisplaybreaks%
\begin{align}
    &\frCondInv \sementails \pathcondi{\procF}{\seExp}{\seCvar}{\userindexvar} \label{eq:fr-inv-pathcond} \\
    \nonumber \\
    &\seCvar \neq \seExp 
     \Rightarrow \forall y \in \ddepsi{\procF}{\seCvar}{\userindexvar} \cup \{ \seCvar\}, \procG \in \cFuncs{\contract} \cup \{\procTick\}. \nonumber \\ 
    & \qquad \qquad \quad \frCondInv  \sementails \left ( \setforall{\locvar}{\lvars{\procG}{\nuserindexvar}} \sendervari{\nuserindexvar} \neq \honuser \rightarrow  \nowritecondi{\procG}{y}{\nuserindexvar} \right ) \label{eq:fr-inv-ddeps} \\
     \nonumber \\
        &\forall\procG \in \cFuncs{\contract}. \execinvi{\procG}{\sendervari{\nuserindexvar} \neq \honuser}{\frCondPre}{\nuserindexvar} \label{eq:fr-pre-cdeps} \\
        \nonumber \\
        & \roundinvi{\roundvar}{\frCondPre}{\frCondInv}{\nuserindexvar} \\
        \nonumber \\
        & \condvars{\frCondPre} \subseteq \seCvarsi{\contract}{\procF}{\userindexvar} \\
        \nonumber \\
        & \condvars{\frCondInv} \subseteq \seCvarsi{\contract}{\procF}{\userindexvar}
\end{align}
\endgroup
\end{lemma}

Correspondingly, we can define the characteristic properties of the backrunning conditions:

\begin{lemma}[Characteristic properties of the backrunning condition]
    Let $\contract$ be a smart contract, $\honuser \in \honusers$, \break
     $\procF \in \cFuncs{\contract}$, 
     $\vSet \subseteq \seCvarsi{\contract}{\procF}{\userindexvar}$,
$\{ \symExRel{\procG} \}_{\procG \in \cFuncs{\contract} \cup \{ \procTick \}, \indexvar \in \{ \userindexvar, \nuserindexvar\} }$ symbolic execution relations over the respective $\seCvarsi{\contract}{\procG}{\indexvar}$, 
 $\eventvar \in \gvars{\contract}$, 
$\seExp \in \expseti{\procF}{\userindexvar}$. 
Let $\brWriteCond{\procF}{\vSet}{\seCvar}{\roundPred} = (\brCondPre, \brCondInv)$ be the result of Algorithm~\ref{algo:br}.
The following properties hold
\begin{align}
&\forall y \in \vSet, \procG \in \cFuncs{\contract}.~ (\seCvar \in \ddepsi{\procG}{y}{\nuserindexvar}  \lor y = \seCvar) \nonumber\\
&\qquad \Rightarrow \brCondInv \sementails \left ( \setforall{\locvar}{\lvars{\procG}{\nuserindexvar}} \sendervari{\nuserindexvar} \neq \honuser \rightarrow \nowritecondi{\procG}{y}{\nuserindexvar} \right ) \\
\nonumber \\
&\forall y \in \vSet, \procG \in \cFuncs{\contract}.~ \seCvar \in \cdepsi{\procG}{y}{\nuserindexvar} \Rightarrow \exists \seExp \in \expseti{\procG}{\nuserindexvar}. \nonumber \\
&\qquad \qquad~ \brCondInv \sementails \left (  \setforall{\locvar}{\lvars{\procG}{\nuserindexvar}} \sendervari{\nuserindexvar} \neq \honuser \rightarrow \pathcondi{\procG}{\seExp}{y}{\nuserindexvar} \right ) \\
\nonumber \\
&\execinvi{\procF}{\sendervari{\userindexvar} = \honuser}{\brCondPre}{\userindexvar} \\
\nonumber \\
&\forall\procG \in \cFuncs{\contract}. \execinvi{\procG}{\sendervari{\nuserindexvar} \neq \honuser}{\brCondPre}{\nuserindexvar} \\
\nonumber \\
        & \roundinvi{\roundvar}{\brCondPre}{\brCondInv}{\nuserindexvar} \\
        \nonumber \\
        & \condvars{\brCondPre} \subseteq \seCvarsi{\contract}{\procF}{\userindexvar} \\
        \nonumber \\
        & \condvars{\brCondInv} \subseteq \seCvarsi{\contract}{\procF}{\userindexvar}
\end{align}
\end{lemma}

Note that both lemmas follow immediately from the construction of the algorithm provided that $\getRoundInv$ provides round invariant and $\mergeRoundInv$ preserves round invariants. 

\subsubsection{Semantic notions of frontrunning and backrunning conditions}

The frontrunning and backrunning conditions produced by the algorithms are specific to individual variables $\seCvar$, which may be changed by an honest user transaction. 
For formulating requirements for the kind of transactions $\txv$ that an honest user may submit at the beginning of a round, we want to ensure that those transactions keep a specific relevant set of variables (denoted as $\vSet$ in the following) to be unaffected by the possible deckstacking attacks.
To this end, on the one side, it needs to be ensured that the effects that $\txv$ has on $\vSet$ are deterministicht within the round (so scheduling attacker transactions before $\txv$) may not change how $\txv$ assigns variables in $\vSet$. 
This can be realized by assuring that frontrunning conditions hold for all variables $\seCvar$ in $\vSet$. 
On the other side, it also needs to be ensured that if $\txv$ changes variables (that could also lie outside of $\vSet$), this cannot affect how attacker transactions again write variables in $\vSet$. 
To ensure this, it needs to be established that the backrunning conditions (for $\vSet$) hold for all variables $\seCvar$ that are written by $\txv$.

Following these intuitions, we establish predicates $\frCond{\conf}{\txv}{\vSet}$ and $\brCond{\conf}{\txv}{\vSet}$, which capture this intuition. Namely $\frCond{\conf}{\txv}{\vSet}$ will express that submitting transaction $\txv$ in configuration $\conf$ will ensure that writes to variables in $\vSet$ cannot be affected by frontrunning, and $\brCond{\conf}{\txv}{\vSet}$ will express that by backrunning the transaction $\txv$ an attacker cannot change the variables in $\vSet$.

In the following, we consider frontrunning conditions $\frWriteCond{\procF}{\seExp}{\seCvar}{\roundPred}$ to be of the form $(\frCondPre, \frCondInv)$ where $\frCondPre \sementails \frCondInv$ and $\frCondPre$ denotes the precondition that needs to hold at the beginning of the round in order to ensure that $\frCondInv$ holds within the whole round. 

We define the following helper predicates operating on assignments that we will then use to define the predicates $\frCond{\cdot}{\cdot}{\cdot}$ and $\brCond{\cdot}{\cdot}{\cdot}$.

\begin{definition}
    Let $\contract$ be a smart contract, $\honuser \in \honusers$, $\procF \in \cFuncs{\contract}$, 
    $\vSet \subseteq \gvars{\contract}$, 
    $\seAss$ with an assignment over $\seCvarsi{\contract}{\procF}{\userindexvar}$, and
$\symExReli{\procF}{\userindexvar}$ a symbolic execution relation over $\seCvarsi{\contract}{\procF}{\userindexvar}$.
We define 
\begin{align*}
\frAssSet{\procF}{\seAss}{\vSet} := \bigwedge_{(\frCondPre, \frCondInv) \in S} \frCondInv
\end{align*}
where 
\begin{align*}
S = \{ \frWriteCond{\procF}{\seExp}{\seCvar}{\roundPred} &~|~ 
\seCvar \in \vSet 
~\land~
\seRes{\seSub}{\seCond} \in \symExReli{\procF}{\userindexvar}
~\land~
\seSub(\seCvar) = \seExp
~\land~ 
\satis{\seAss}{\seCond}
\}
\end{align*}
\end{definition}

Intuitively, $\frAssSet{\procF}{\seAss}{\vSet}$ denotes the conjunction for all frontrunning conditions that enforce that variables in $\vSet$ are assigned as mandated by $\seAss$ (which in particular contains the assignment of the argument variables for the function call to $\procF$).

\begin{remark}
Note that we implicitly assume the existence of a smart contract $\contract$ and corresponding symbolic execution relations for all of its functions and do not always make this assumption explicit in the lemmas of this section.
\end{remark}

We provide a definition that caters to the transaction setting (instead of operating on assignments): 
\begin{definition}[Frontrunning condition]
    Let $\contract$ be a smart contract, $\honuser \in \honusers$, $\procF \in \cFuncs{\contract}$ and $\symExReli{\procF}{\userindexvar}$ a symbolic execution relation over $\seCvarsi{\contract}{\procF}{\userindexvar}$. 
    Further let $\txv = \contrCall{\procF(\fargs)}$ for some arguments $\fargs$ and $\sender{\txv} = \honuser$ and $\conf$ be a configuration. We define
    \begin{align*}
\frCond{\conf}{\txv}{\vSet} := 
\satis{\confToAssi{\conf}{\txv}{\userindexvar}}{\frAssSet{\procF}{\confToAssi{\conf}{\txv}{\userindexvar}}{\vSet}}
    \end{align*}
\end{definition}

So intuitively, $\frCond{\conf}{\txv}{\vSet} $ denotes the condition on a configuration $\conf$ and transaction $\txv$, which ensures that concurrent attacker transactions cannot interfere with how $\txv$ writes the variables in $\vSet$.

We can provide similar definitions for the backrunning condition:

\begin{definition}
Let $\contract$ be a smart contract, $\honuser \in \honusers$, $\procF \in \cFuncs{\contract}$, 
    $\vSet \subseteq \gvars{\contract}$, 
    $\seAss$ with an assignment over $\seCvarsi{\contract}{\procF}{\userindexvar}$, and
$\symExReli{\procF}{\userindexvar}$ a symbolic execution relation over $\seCvarsi{\contract}{\procF}{\userindexvar}$.
We define 
\begin{align*}
    \brAssWrite{\procF}{\vSet}{\seAss} := \bigwedge_{(\brCondPre, \brCondInv) \in S} \brCondInv
\end{align*}
where 
\begin{align*}
S = \{ \brWriteCond{\procF}{\vSet}{\seCvar}{\roundPred} &~|~ 
\seCvar \in \gvars{\contract}
~\land~
\seRes{\seSub}{\seCond} \in \symExReli{\procF}{\userindexvar} \\
&~\land~ 
\seSub(\seCvar) = \seExp
~\land~ 
\satis{\seAss}{\seCond}
~\land~ 
\seExp \neq \seCvar
\}
\end{align*}
\end{definition}

Intuitively, $\brAssWrite{\procF}{\vSet}{\seAss}$ denotes the conjunction for all backrunning conditions that enforce that variables in $\vSet$ may not be influenced by variables that are written when executing $\procF$ following the assignment $\seAss$ (which in particular contains the assignment of the argument variables for the function call to $\procF$).
Note that this definition considers that $\seAss$ may induce changes in any variable $\seCvar$ of the contract (when executing $\procF$), but it ensures that for all such variables $\seCvar$, $\brWriteCond{\procF}{\vSet}{\seCvar}{\roundPred}$ holds, ensuring that changes to such variables leave the variables in $\vSet$ unaffected.

Again, we provide a definition that caters to the transaction setting (instead of operating on assignments): 
\begin{definition}[Backrunning condition]
    Let $\contract$ be a smart contract, $\honuser \in \honusers$, $\procF \in \cFuncs{\contract}$ and $\symExReli{\procF}{\userindexvar}$ a symbolic execution relation over $\seCvarsi{\contract}{\procF}{\userindexvar}$. 
    Further let $\txv = \contrCall{\procF(\fargs)}$ for some arguments $\fargs$ and $\sender{\txv} = \honuser$ and $\conf$ be a configuration. We define
    \begin{align*}
\brCond{\conf}{\txv}{\vSet} := \satis{\confToAssi{\conf}{\txv}{\userindexvar}}{\brAssWrite{\procF}{\confToAssi{\conf}{\txv}{\userindexvar}}{\vSet}}
    \end{align*}
\end{definition}

So intuitively, $\brCond{\conf}{\txv}{\vSet} $ denotes the condition on a configuration $\conf$ and transaction $\txv$, which ensures that the way that concurrent attacker transactions write variables in $\vSet$ does not depend on how the transaction $\txv$ changes any (!) variables.

We provide analogous definitions to the previous ones for the round preconditions for frontrunning and backrunning.
\begin{definition}
    Let $\contract$ be a smart contract, $\honuser \in \honusers$, $\procF \in \cFuncs{\contract}$, 
    $\vSet \subseteq \gvars{\contract}$, 
    $\seAss$ with an assignment over $\seCvarsi{\contract}{\procF}{\userindexvar}$, and
$\symExReli{\procF}{\userindexvar}$ a symbolic execution relation over $\seCvarsi{\contract}{\procF}{\userindexvar}$.
We define 
\begin{align*}
\frAssSetPre{\procF}{\seAss}{\vSet} := \bigwedge_{(\frCondPre, \frCondInv) \in S} \frCondPre
\end{align*}
where 
\begin{align*}
S = \{ \frWriteCond{\procF}{\seExp}{\seCvar}{\roundPred} &~|~ 
\seCvar \in \vSet 
~\land~
\seRes{\seSub}{\seCond} \in \symExReli{\procF}{\userindexvar}
~\land~
\seSub(\seCvar) = \seExp
~\land~ 
\satis{\seAss}{\seCond}
\}
\end{align*}
\end{definition}

\begin{definition}[Frontrunning round condition]
    Let $\contract$ be a smart contract, $\honuser \in \honusers$, $\procF \in \cFuncs{\contract}$ and $\symExReli{\procF}{\userindexvar}$ a symbolic execution relation over $\seCvarsi{\contract}{\procF}{\userindexvar}$. 
    Further let $\txv = \contrCall{\procF(\fargs)}$ for some arguments $\fargs$ and $\sender{\txv} = \honuser$ and $\conf$ be a configuration. We define
    \begin{align*}
\frCondR{\conf}{\txv}{\vSet} := \satis{\confToAssi{\conf}{\txv}{\userindexvar}}{\frAssSetPre{\procF}{\confToAssi{\conf}{\txv}{\userindexvar}}{\vSet}}
    \end{align*}
\end{definition}

\begin{definition}
Let $\contract$ be a smart contract, $\honuser \in \honusers$, $\procF \in \cFuncs{\contract}$, 
    $\vSet \subseteq \gvars{\contract}$, 
    $\seAss$ an assignment over $\seCvarsi{\contract}{\procF}{\userindexvar}$, and
$\symExReli{\procF}{\userindexvar}$ a symbolic execution relation over $\seCvarsi{\contract}{\procF}{\userindexvar}$.
We define 
\begin{align*}
    \brAssWritePre{\procF}{\vSet}{\seAss} := \bigwedge_{(\brCondPre, \brCondInv) \in S} \brCondPre
\end{align*}
where 
\begin{align*}
S = \{ \brWriteCond{\procF}{\vSet}{\seCvar}{\roundPred} &~|~ 
\seRes{\seSub}{\seCond} \in \symExReli{\procF}{\userindexvar}
~\land~
\seSub(\seCvar) = \seExp
~\land~ 
\satis{\seAss}{\seCond}
~\land~ 
\seExp \neq \seCvar
\}
\end{align*}
\end{definition}

\begin{definition}[Backrunning condition]
    Let $\contract$ be a smart contract, $\honuser \in \honusers$, $\procF \in \cFuncs{\contract}$ and $\symExReli{\procF}{\userindexvar}$ a symbolic execution relation over $\seCvarsi{\contract}{\procF}{\userindexvar}$. 
    Further let $\txv = \contrCall{\procF(\fargs)}$ for some arguments $\fargs$ and $\sender{\txv} = \honuser$ and $\conf$ be a configuration. We define
    \begin{align*}
\brCondR{\conf}{\txv}{\vSet} := \satis{\confToAssi{\conf}{\txv}{\userindexvar}}{\brAssWritePre{\procF}{\confToAssi{\conf}{\txv}{\userindexvar}}{\vSet}}
    \end{align*}
\end{definition}

We can establish the following properties for the combined frontrunning and backrunning conditions, which are immediately inherited from the definitions of the individual frontrunning and backrunning conditions.

\begin{lemma}
    \label{lem:fr-br-big-roundinv}
    Let $\contract$ be a smart contract, $\honuser \in \honusers$, $\procF \in \cFuncs{\contract}$, 
    $\vSet \subseteq \gvars{\contract}$, 
    $\seAss$ an assignment over $\seCvarsi{\contract}{\procF}{\userindexvar}$, and
$\{ \symExRel{\procG} \}_{\procG \in \cFuncs{\contract} \cup \{ \procTick \}, \indexvar \in \{ \userindexvar, \nuserindexvar\} }$ symbolic execution relations over the respective $\seCvarsi{\contract}{\procG}{\indexvar}$.
Then it holds that 
$$\roundinvi{}{\frAssSetPre{\procF}{\seAss}{\vSet}}{\frAssSet{\procF}{\seAss}{\vSet}}{}$$
and 
$$\execinvi{\procG}{\sendervari{\nuserindexvar} \neq \honuser}{\frAssSetPre{\procF}{\seAss}{\vSet}}{\nuserindexvar}$$
 \label{lem:fr-lift-inv}
\end{lemma}
\begin{proof}
Follows immediately from the definitions of $\frAssSetPre{\procF}{\seAss}{\vSet}$, $\frAssSet{\procF}{\seAss}{\vSet}$, $\brAssWritePre{\procF}{\vSet}{\seAss}$, $\brAssWrite{\procF}{\vSet}{\seAss}$, and the characteristic properties of the frontrunning and backrunning conditions.
\end{proof}

\begin{lemma}
Let $\contract$ be a smart contract, $\honuser \in \honusers$, $\procF \in \cFuncs{\contract} \cup \{ \procTick \}$, 
    $\vSet \subseteq \gvars{\contract}$, 
    $\seAss$ an assignment over $\seCvarsi{\contract}{\procF}{\userindexvar}$, and
    $\{ \symExRel{\procG} \}_{\procG \in \cFuncs{\contract} \cup \{ \procTick \}, \indexvar \in \{ \userindexvar, \nuserindexvar\} }$ symbolic execution relations over the respective $\seCvarsi{\contract}{\procG}{\indexvar}$.
Then it holds that 
 $$\roundinvi{}{\brAssWritePre{\procF}{\vSet}{\seAss}}{\brAssWrite{\procF}{\vSet}{\seAss}}{}$$
 and 
$$\execinvi{\procF}{\sendervari{\nuserindexvar} \neq \honuser}{\brAssWritePre{\procF}{\seAss}{\vSet}}{\nuserindexvar}$$
and for all $\procG \in \cFuncs{\contract}$ that
$$\execinvi{\procG}{\sendervari{\nuserindexvar} \neq \honuser}{\brAssWritePre{\procF}{\seAss}{\vSet}}{\nuserindexvar}$$
\label{lem:br-lift-inv}
\end{lemma}

We further provide the following helper notion of a write set that will come in handy.

\begin{definition}[Write sets]
    Let $\contract$ be a contract with variables $\cvars{\contract}$.
Let $\config$ be a configuration, $\txv$ a transaction such that $\txv = \contrCall{\procF(\fargs)}$, $\indexvar$ be an index variable and $\symExRel{\procF}$ be a symbolic execution relation over $\seCvars{\contract}$ and $\vSet \subseteq \cvars{\contract}$.
\begin{align*}
    \writeseti{\procF}{\config}{\txv}{\vSet}{\indexvar} := \{ \seCvar \in \vSet ~|~ & \seRes{\seSub}{\seCond} \in \symExReli{\procF}{\userindexvar} ~\land~ \seSub(\seCvar) = \seExp \\
    & ~\land~  \satis{\confToAssi{\config}{\txv}{\indexvar}}{\seCond} 
    ~\land~ \seExp \neq \seCvar \}
\end{align*}
\end{definition}

Intuitively, $\writeseti{\procF}{\config}{\txv}{\vSet}{\indexvar}$ denotes the set of variables from $\vSet$, which are written by the transaction $\txv$ when executed in configuration $\config$.

We establish some useful properties of write sets. 

First, we can show that all variables from $\vSet$, which are not contained in the write set $\writeset{\config}{\txv}{\vSet}$ stay unchanged when executing $\txv$ in $\config$:
\begin{lemma}[Preservation of variables outside write set]
Let $\contract$ be a contract and $\vSet \subseteq \gvars{\contract}$ a set of variables.
Let $\config_1$ and $\config_2$ be configurations and $\txv$ a transaction such that 
$\config_1 \trans{\txv} \config_2$. 
Then it holds that 
$$\config_1 \equivfor{\vSet \setminus \writeset{\config_1}{\txv}{\vSet}} \config_2$$
\label{lem:nowrite-outside-writeset}
\end{lemma}

Next, we can show that two configurations that agree on the control dependencies of some set of variables $\vSet$ also write the same variables from $\vSet$:
\begin{lemma}[Agreement of write sets]
    Let $\contract$ be a contract and $\vSet \subseteq \gvars{\contract}$ a set of variables.
Let $\config$ and $\tick{\config}$ be configurations and $\txv = \contrCall{\procF(\fargs)}$ be a transaction invoking function $\procF$ of contract $\contract$.
If for $C = \gvars{\contract} \cap \bigcup_{\seCvar \in \vSet} \cdeps{\procF}{\seCvar}$ (for some index $i$) it holds that
$$\config \equivfor{C} \tick{\config}$$ 
then 
also 
$$\writeset{\config}{\txv}{\vSet} = \writeset{\tick{\config}}{\txv}{\vSet}$$
\label{lem:writeset-agreement}
\end{lemma}

\subsubsection{Correctness of Frontrunning and Backrunning Conditions}
\label{subsec:conditions-correctness}

We next establish the semantic correctness properties for the semantic frontrunning and backrunning predicates that we established before. 

First, we make the assumption concerning the connection between written variables and produced observables explicit:

\begin{assumption}[Unique observables]
Let $\config_1$ and $\config_2$ be configurations and $\txv = \contrCall{\procF(\fargs)}$ be a transaction calling function $\procF \in \cFuncs{\contract}$. 
Let $\eventvar_\procF \in \criticalvars{}$ be the event variable of $\procF$
and $\vSet \subseteq \gvars{\contract}$ such that $\eventvar_\procF \in vSet$.
If $\config_1 \trans[\obs]{\txv} \config_2$ for some observable $\obs$ then the following holds
\begin{itemize}
\item $\length{\obs} \leq 1$
\item $\eventvar_\procF \in \writeset{\config_1}{\txv}{\vSet} \Rightarrow \obs = [\obscV{\eventvar_\procF}{\cstate{\conf_2}{\contract}(\eventvar_\procF)}]$
\item $\forall \eventvar \, v.~\obs = [\obscV{\eventvar}{v}] \Rightarrow \eventvar = \eventvar_\procF ~\land~ \eventvar_\procF \in \writeset{\config_1}{\txv}{\vSet}$
\end{itemize}
\label{asm:observables}
\end{assumption}

Note that this assumption requires that every transaction produces at most one observable (first condition) and that it produces an observable exactly when its corresponding event variable $\eventvar_\procF$ is written, which then contains the value assigned to this variable (the last two conditions).
We formulate this condition in terms of write sets since this provides us with a convenient way to characterize when a transaction execution writes a variable. A purely semantic notions (which would consider changes to the values assigned to the event variable in $\config_1$ and $\config_2$ would fall short of adequately characterizing when the same event is recorded two times in a row.)

We first state the core correctness lemma for the semantic frontrunning conditions, which intuitively says that if $\frCond{\config}{\htxv}{\vSet}$ holds for a transaction $\htxv$ then placing an adversarial transaction $\atxv$ before $\htxv$ will not affect the behavior of $\htxv$ (on the variables $\vSet$). 
More specifically, this means that if $\htxv$ writes variables then those transactions will be written to the same values as they would be when executing $\htxv$ without $\atxv$.
Similarly, if executing $\htxv$ leaves variables unchanged then this should be also the case when executing $\atxv$ before.
Note that if the event variables $\criticalvars$ are contained in $\vSet$ then this additionally implies that executing $\htxv$ will produce the same observables (independently of whether $\atxv$ was executed before) because observables are produced exactly if the corresponding event variables are written.

Formally, we capture this property with the following lemma.

\begin{lemma}[Correctness of $\frCond{\cdot}{\cdot}{\cdot}$ for attacker transactions]
Let $\config_1$, $\config_2$ and $\config_3$ be configurations and $\htxv$ be a transaction such that $\sender{\htxv} = \honuser$ for $\honuser \in \honusers$, and $\vSet$ be a set of variables such that $\criticalvars \subseteq \vSet$.
Further, let $\atxv$ be a transaction such that $\sender{\atxv} \neq \honuser$. 
Assume that $\frCond{\config_1}{\htxv}{\vSet}$ hold and that 
$$
    \conf_1 \trans[\aobs]{\atxv} \conf_2 \trans[\hobs]{\htxv} \conf_3 
$$
Then it holds that 
\begin{itemize}
\item For all configurations $\tick{\config_2}$ such that $\conf_1 \trans[\tick{\hobs}]{\htxv} \tick{\conf_2}$, it holds that 
$$\config_3 \equivfor{\writeset{\config_1}{\htxv}{\vSet}} \tick{\conf_2} ~\land~ \hobs = \tick{\hobs}$$
\item It holds that $$\config_2 \equivfor{\vSet \setminus \writeset{\config_1}{\htxv}{\vSet}} \config_3$$
\end{itemize}
\label{lem:fr-correctness-user-attacker}
\end{lemma}
\begin{proof}
This property follows directly from the soundness and correctness of symbolic execution, as well as from the characteristic properties of frontrunning conditions.
\end{proof}

We can show a similar correctness lemma for $\txB$ transactions, which shows that the frontrunning condition rules out interference with $\txB$.

\begin{lemma}[Correctness of $\frCond{\cdot}{\cdot}{\cdot}$ for $\txB$ transactions]
Let $\config_1$, $\config_2$ and $\config_3$ be configurations and $\htxv$ be a transaction such that $\sender{\htxv} = \honuser$ for $\honuser \in \honusers$, and $\vSet$ be a set of variables such that $\criticalvars \subseteq \vSet$.
Assume that $\frCond{\config_1}{\htxv}{\vSet}$ and $\round(\config_1) = \round(\config_2)$ hold and that 
$$
    \conf_1 \trans[]{\txB} \conf_2 \trans[\hobs]{\htxv} \conf_3 
$$
Then it holds that 
\begin{itemize}
\item For all configurations $\tick{\config_2}$ such that $\conf_1 \trans[\tick{\hobs}]{\htxv} \tick{\conf_2}$, it holds that 
$$\config_3 \equivfor{\writeset{\config_1}{\htxv}{\vSet}} \tick{\conf_2} ~\land~ \hobs = \tick{\hobs}$$
\item It holds that $$\config_2 \equivfor{\vSet \setminus \writeset{\config_1}{\htxv}{\vSet}} \config_3$$
\end{itemize}
\label{lem:fr-correctness-user-tau}
\end{lemma}
\begin{proof}
This property follows directly from the soundness and correctness of symbolic execution, as well as from the characteristic properties of backrunning conditions.
\end{proof}

Similar correctness lemmas hold for the semantic backrunning conditions. 
This lemma states in a similar fashion that if $\brCond{\config}{\htxv}{\vSet}$ holds in a configuration then placing the transaction $\htxv$ before an (adversarial) transaction $\atxv$ does not change the effect of $\atxv$ (on $\vSet$).

\begin{lemma}[Correctness of $\brCond{\cdot}{\cdot}{\cdot}$ for honest user transactions]
Let $\config_1$, $\config_2$ and $\config_3$ be configurations and $\htxv$ be a transaction such that $\sender{\htxv} = \honuser$ for $\honuser \in \honusers$, and $\vSet$ be a set of variables such that $\criticalvars \subseteq \vSet$.
Further, let $\atxv$ be a transaction such that $\sender{\atxv} \neq \honuser$. 
Assume that $\brCond{\config_1}{\htxv}{\vSet}$ hold and that 
$$
    \conf_1 \trans[\hobs]{\htxv} \conf_2 \trans[\aobs]{\atxv} \conf_3 
$$
Then it holds that 
\begin{itemize}
\item For all configurations $\tick{\config_2}$ such that $\conf_1 \trans[\tick{\aobs}]{\atxv} \tick{\conf_2}$, it holds that 
$$\config_3 \equivfor{\writeset{\config_1}{\atxv}{\vSet}} \tick{\conf_2} ~\land~ \aobs = \tick{\aobs}$$
\item It holds that $$\config_2 \equivfor{\vSet \setminus \writeset{\config_1}{\atxv}{\vSet}} \config_3$$
\end{itemize}
\label{lem:br-correctness-user-attacker}
\end{lemma}

A similar correctness lemma holds for backrunning conditions with respect to $\txB$ transactions (which similarly to honest user transactions, may not impact the behavior of attacker transactions).

\begin{lemma}[Correctness of $\brCond{\cdot}{\cdot}{\cdot}$ for $\txB$ transactions]
Let $\config_1$, $\config_2$ and $\config_3$ be configurations and $\vSet$ be a set of variables such that $\criticalvars \subseteq \vSet$ and $\honuser \in \honusers$ be a user.
Further, let $\atxv$ be a transaction such that $\sender{\atxv} \neq \honuser$. 
Assume that $\brCond{\config_1}{\txB}{\vSet}$ and $\round(\config_1) = \round(\config_2)$ hold and that 
$$
    \conf_1 \trans[]{\txB} \conf_2 \trans[\aobs]{\atxv} \conf_3 
$$
Then it holds that 
\begin{itemize}
\item For all configurations $\tick{\config_2}$ such that $\conf_1 \trans[\tick{\aobs}]{\atxv} \tick{\conf_2}$, it holds that 
$$\config_3 \equivfor{\writeset{\config_1}{\atxv}{\vSet}} \tick{\conf_2} ~\land~ \aobs = \tick{\aobs}$$
\item It holds that $$\config_2 \equivfor{\vSet \setminus \writeset{\config_1}{\atxv}{\vSet}} \config_3$$
\end{itemize}
\label{lem:br-correctness-tau-attacker}
\end{lemma}
\begin{proof}
This property follows directly from the soundness and correctness of symbolic execution, as well as from the characteristic properties of backrunning conditions.
\end{proof}

Additionally, we can show the following lemma, which states that provided that $\frCond{\conf}{\htxv}{\vSet}$ holds, it is ensured that the variables from $\vSet$ written by $\htxv$ in $\conf$ can never overlap with variables that would be written by attacker transactions $\atxv$ when executed in $\conf$.

\begin{lemma}[Disjointness of write sets]
Let $\config$ be a configuration and $\htxv$ be a transaction such that $\sender{\htxv} = \honuser$ for $\honuser \in \honusers$, and $\vSet \subseteq \gvars{\contract}$ be a set of variables. 
Assume that $\frCond{\config}{\htxv}{\vSet}$ holds.
Then it holds for all transactions $\atxv$ with $\sender{\atxv} \neq \honuser$ that 
$$
\writeset{\config}{\htxv}{\vSet} \subseteq \vSet \setminus \writeset{\config}{\atxv}{\vSet}
$$
\label{lem:disjointness-writesets}
\end{lemma}

Additionally, we can provide semantic versions of the notions of invariants and round invariants, which hold for the frontrunning and backrunning conditions.

We first can show that if $\frCondR{\cdot}{\cdot}{\cdot}$ holds at the beginning of the round, then $\frCond{\cdot}{\cdot}{\cdot}$ holds throughout the round as long as only $\txB$ or attacker transactions get executed. 

\begin{lemma}[Frontrunning Conditions throughout rounds]
Let $\htxv$ be a transaction with $\sender{\htxv} = \honuser$ for some $\honuser \in \honusers$ and $\vSet \subseteq \gvars{\contract}$.
Let $\config_1$, $\config_2$, \dots, $\config_n$ be configurations such that 
$\round(\config_j) = \round(\config_l)$ for $j, l \in \{ 1, \dots, n \}$ and $\lastBlockNumber(\config_1) \bmod{k} = 0$.
Further, let 
$\config_1 \trans{\txv_1} \config_2 \trans{\txv_2} \dots \trans{\txv_{n-1}} \config_{n}$
for transaction $\txv_1, \dots \txv_{n-1}$ such that for all $j \in \{1, \dots n-1 \}$, it holds that
either $\txv_j = \txB$ or 
$\txv_j = \contrCall{\procG(\fargs)}$ for some $\procG \in \cFuncs{\contract}$
and $\sender{\txv_j} \neq \honuser$.
Then if $\frCondR{\config_1}{\htxv}{\vSet}$ then also $\frCond{\config_n}{\htxv}{\vSet}$.
\label{lem:frcond-through-rounds}
\end{lemma}
\begin{proof}
Follows immediately from Lemma~\ref{lem:inv-preservation} and Lemma~\ref{lem:fr-lift-inv}. 
\end{proof}

We can establish a similar lemma for backrunning conditions. 
Note that backrunning conditions are also invariant w.r.t. the honest user transaction, resulting in a slightly stronger lemma:

\begin{lemma}[Backrunning Conditions throughout rounds]
Let $\htxv$ be a transaction with $\sender{\htxv} = \honuser$ for some $\honuser \in \honusers$ and $\vSet \subseteq \gvars{\contract}$.
Let $\config_1$, $\config_2$, \dots, $\config_n$ be configurations such that 
$\round(\config_j) = \round(\config_l)$ for $j, l \in \{ 1, \dots, n \}$ and $\lastBlockNumber(\config_1) \bmod{k} = 0$.
Further, let 
$\config_1 \trans{\txv_1} \config_2 \trans{\txv_2} \dots \trans{\txv_{n-1}} \config_{n}$
for transaction $\txv_1, \dots \txv_{n-1}$ such that for all $j \in \{1, \dots n-1 \}$, it holds that
either $\txv_j = \txB$ or 
$\txv_j = \htxv$ or
$\txv_j = \contrCall{\procG(\fargs)}$ for some $\procG \in \cFuncs{\contract}$
and $\sender{\txv_j} \neq \honuser$.
Then if $\brCondR{\config_1}{\htxv}{\vSet}$ then also $\brCond{\config_n}{\htxv}{\vSet}$.
\label{lem:brcond-through-rounds}
\end{lemma}

\begin{proof}
Follows immediately from Lemma~\ref{lem:inv-preservation} and Lemma~\ref{lem:br-lift-inv}. 
\end{proof}

\begin{remark}
Throughout this paper, we assume that the attacker also only schedules transactions of contract $\contract$. 
This assumption could easily be lifted since transactions of other contracts $\contract' \neq \contract$ let the contract variables of $\contract$ unchanged.
\end{remark}

\subsubsection{Commutativity}

Using the semantic correctness conditions, we can next establish a commutativity lemma, which shows that front- and backrunning conditions together ensure that honest user transactions and attacker transactions commute.

We will consider commutativity with respect to a set $\varfixpoint \subseteq \gvars{\contract}$ of variables, which contains the event variables $\criticalvars$ and which is closed under data and control dependencies of the contract $\contract$. 
Formally, we define closedness as follows:

\begin{definition}[Closedness under dependencies]
Let $\varfixpoint \subseteq \gvars{\contract}$ be a set of variables and $\indexvar$ be some index. 
We say that $\varfixpoint$ is \emph{closed under dependencies} (written $\depclosed{\varfixpoint}$) if it holds that 
\begin{align*}
\gvars{\contract} \cap \bigcup_{y \in \varfixpoint, \procG \in \cFuncs{\contract}}{\ddeps{\procG}{y} \cup \cdeps{\procG}{y} }
    \subseteq \varfixpoint 
\end{align*}
\end{definition}

Note that the data or control dependencies of the different contract functions may contain local variables (e.g., referring to the transaction parameters, such as the transaction sender or the value sent along with the transaction). 
Those variables, however, should not be considered in the set $\varfixpoint$, which ranges over the global variables $\gvars{\contract}$ (since it will denote the set of variables, which will stay uninfluenced by reordering of transactions).

\begin{lemma}[Commutativity (Between user and attacker transactions)]
    Let $\config_1$, $\tick{\config_1}$, $\config_2$, $\tick{\config_2}$, $\config_3$ and $\tick{\config_3}$ be configurations and $\htxv$ be a transaction such that $\sender{\htxv} = \honuser$ for $\honuser \in \honusers$. 
    Let $\varfixpoint \subseteq \gvars{\contract}$ be a set of variables such that $\criticalvars \subseteq \varfixpoint$ and $\depclosed{\varfixpoint}$.

Further, let $\atxv$ be a transaction such that $\sender{\atxv} \neq \honuser$. 
Assume that $\config_1 \equivfor{\varfixpoint} \tick{\config_1}$ and $\frCond{\config_1}{\htxv}{\varfixpoint}$,
$\brCond{\config_1}{\htxv}{\varfixpoint}$, 
$\frCond{\tick{\config_1}}{\htxv}{\varfixpoint}$,
and
$\brCond{\tick{\config_1}}{\htxv}{\varfixpoint}$,
hold. 

Then
\begin{align*}
  \conf_1 \trans[\hobs]{\htxv} \conf_2 \trans[\aobs]{\atxv} \conf_3
  ~\land~ 
  \tick{\conf_1} \trans[\tick{\aobs}]{\atxv} \tick{\conf_2} \trans[\tick{\hobs}]{\htxv} \tick{\conf_3} \\
  \Rightarrow 
  \aobs = \tick{\aobs}
  ~\land~ 
  \hobs = \tick{\hobs}
  ~\land~ 
  \conf_3 \equivfor{\varfixpoint} \tick{\conf_3}
  \end{align*}
for observables $\hobs$, $\tick{\hobs}$, $\aobs$, $\tick{\aobs}$.
\label{lem:com-user-attacker}
\end{lemma}
\begin{proof}
In the following, we will use the shorthands
$\WFset := \writeset{\tick{\config_1}}{\htxv}{\varfixpoint}$
and
$\WGset := \writeset{\config_1}{\atxv}{\varfixpoint}$.
Note that this gives us immediately (using Lemma~\ref{lem:nowrite-outside-writeset}) for $\config_2'$ with $\config_1 \trans{\atxv} \config_2'$
that $\config_1 \equivfor{\NWGset} \config_2'$.
Consequently, using Lemma~\ref{lem:dep-soundness},
from $\config_1 \equivfor{\varfixpoint} \tick{\config_1}$ and $\depclosed{\varfixpoint}$, 
we can conclude that $\config_2' \equivfor{\NWGset} \tick{\config_2}$, 
and so also $\config_1 \equivfor{\NWGset} \tick{\config_2}$
and since $\config_1 \equivfor{\varfixpoint} \tick{\config_1}$
also $\tick{\config_1} \equivfor{\NWGset} \tick{\config_2}$ (NWu1).
Similarly, we know by Lemma~\ref{lem:nowrite-outside-writeset} that for $\tick{\config_2}'$ with 
$\tick{\config_1} \trans{\htxv} \tick{\config_2}'$ we get 
$\tick{\config_1} \equivfor{\NWFset} \tick{\config_2}'$.
So, using Lemma~\ref{lem:dep-soundness},
from $\config_1 \equivfor{\varfixpoint} \tick{\config_1}$ and $\depclosed{\varfixpoint}$, 
we can conclude that $\tick{\config_2'} \equivfor{\NWFset} \config_2$, 
and so also $\tick{\config_1} \equivfor{\NWFset} \config_2$
and since $\tick{\config_1} \equivfor{\varfixpoint} \config_1$
also $\config_1 \equivfor{\NWFset} \config_2$ (NWa1).

Further, by Lemma~\ref{lem:fr-correctness-user-attacker}, we know that 
$\tick{\config_2} \equivfor{\NWFset} \tick{\config_3}$ (NWa2)
and that for $\tick{\config_2}'$ and $\tick{\hobs}'$ with 
$\tick{\config_1} \trans[\tick{\hobs}']{\htxv} \tick{\config_2}'$ it holds 
that $\tick{\config_2}' \equivfor{\WFset} \tick{\config_3}$ and $\tick{\hobs} = \tick{\hobs}'$. 
Using Lemma~\ref{lem:dep-soundness}, from
$\config_1 \equivfor{\varfixpoint} \tick{\config_1}$ and $\depclosed{\varfixpoint}$
we can also conclude that $\tick{\config_2}' \equivfor{\WFset} \config_2$
and $\tick{\hobs}' = \hobs$
and consequently $\tick{\hobs} = \hobs$ and
$\tick{\config_3} \equivfor{\WFset} \config_2$ (Wa).
Similarly, from Lemma~\ref{lem:br-correctness-user-attacker}, we know that 
$\config_2 \equivfor{\NWGset} \config_3$ (NWu2)
and that for $\config_2'$ and $\aobs'$ with 
$\config_1 \trans[\aobs']{\atxv} \config_2'$ it holds 
that $\config_2' \equivfor{\WGset} \config_3$ and $\aobs = \aobs'$. 
Using Lemma~\ref{lem:dep-soundness}, from
$\config_1 \equivfor{\varfixpoint} \tick{\config_1}$ and $\depclosed{\varfixpoint}$,
we can also conclude that
$\config_2' \equivfor{\WGset} \tick{\config_2}$
and $\aobs' = \tick{\aobs}$
and consequently
$\aobs = \tick{\aobs}$ and
$\config_3 \equivfor{\WGset} \tick{\config_2}$ (Wu).
Further, from Lemma~\ref{lem:disjointness-writesets} (combined with Lemma~\ref{lem:writeset-agreement}), we also know that
$\WGset \subseteq \NWFset$ (I1) and consequently also 
$\WFset \subseteq \NWGset$ (I2).
We now show $\config_3 \equivfor{\varfixpoint} \tick{\config_3}$ by showing
$\config_3 \equivfor{\WGset} \tick{\config_3}$, 
$\config_3 \equivfor{\WFset} \tick{\config_3}$, and
$\config_3 \equivfor{(\NWGset) \cap (\NWFset)} \tick{\config_3}$ separately.
\begin{itemize}
\item %
By (NWa2), we have $\tick{\config_2} \equivfor{\NWFset} \tick{\config_3}$
and consequently by (I1) also $\tick{\config_2} \equivfor{\WGset} \tick{\config_3}$.
Together with (Wu) ($\config_3 \equivfor{\WGset} \tick{\config_2}$), this gives us 
$\config_3 \equivfor{\WGset} \tick{\config_3}$.
\item  %
By (NWu2), we have $\config_2 \equivfor{\NWGset} \config_3$ and consequently, by (I2) also 
$\config_2 \equivfor{\WFset} \config_3$. 
Together with (Wa) ($\tick{\config_3} \equivfor{\WFset} \config_2$), this gives us 
$\config_3 \equivfor{\WFset} \tick{\config_3}$.
\item %
Let in the following be $\NWcap := (\NWGset) \cap (\NWFset)$.
We have that $\config_2 \equivfor{\NWcap} \config_3$ (by (NWu2) and $\NWcap \subseteq \NWGset$), 
$\config_2 \equivfor{\NWcap} \config_1$ (by (NWa2) and $\NWcap \subseteq \NWFset$),
$\config_1 \equivfor{\NWcap} \tick{\config_1}$ (by assumption and $\NWcap \subseteq \varfixpoint$), 
$\tick{\config_1} \equivfor{\NWcap} \tick{\config_2}$ (by (NWu1) and $\NWcap \subseteq \NWGset$),
and $\tick{\config_2} \equivfor{\NWcap} \tick{\config_3}$ (by (NWa2) and $\NWcap \subseteq \NWFset$).
Together, this gives us $\config_3 \equivfor{\NWcap} \tick{\config_3}$ (by transitivity of $\equivfor{}$).
\end{itemize}
Note that $\config_3 \equivfor{\WGset} \tick{\config_3}$, 
$\config_3 \equivfor{\WFset} \tick{\config_3}$, and
$\config_3 \equivfor{(\NWGset) \cap (\NWFset)} \tick{\config_3}$
give us $\config_3 \equivfor{\varfixpoint} \tick{\config_3}$ because
\begin{align*}
\WFset \cup \WGset \cup ((\NWGset) \cap (\NWFset)) \\
= (\WFset \cup \WGset) \cup (\varfixpoint \setminus (\WFset \cup \WGset)) = \varfixpoint
\end{align*}
\end{proof}

In addition to the commutativity between the user and the attacker transaction, we also need to establish that also 
(i) user transactions commute with $\txB$ within one round and
(ii) attacker transactions commute with $\txB$ within one round. 
This is crucial since the attacker may freely move transactions between blocks. 

We state the corresponding lemmas analogously to the previous commutativity lemma.

For the commutativity between user transactions and $\txB$ transactions, it is sufficient of the initial configuration satisfies the frontrunning condition, since this condition already ensures that the behavior of the transaction $\htxv$ (i) does not depend directly on the block number and that (ii) indirect dependencies to the block number cannot impact the behavior of $\htxv$ within one round (due to the round invariant property of the frontrunning condition).

\begin{lemma}[Commutativity (Between user and $\txB$ transactions)]
    Let $\config_1$, $\tick{\config_1}$, $\config_2$, $\tick{\config_2}$, $\config_3$ and $\tick{\config_3}$ be configurations and $\htxv$ be a transaction such that $\sender{\htxv} = \honuser$ for $\honuser \in \honusers$. 
    Let $\varfixpoint \subseteq \gvars{\contract}$ be a set of variables such that $\criticalvars \subseteq \varfixpoint$ and $\depclosed{\varfixpoint}$.

Assume that $\config_1 \equivfor{\varfixpoint} \tick{\config_1}$ and $\frCond{\config_1}{\htxv}{\varfixpoint}$,
$\round(\config_1) = \round(\config_3)$ and 
$\round(\tick{\config_1}) = \round(\tick{\config_3})$
and
$\frCond{\tick{\config_1}}{\htxv}{\varfixpoint}$,
hold. 

Then
\begin{align*}
  \conf_1 \trans[\hobs]{\htxv} \conf_2 \trans[]{\txB} \conf_3
  ~\land~ 
  \tick{\conf_1} \trans[]{\txB} \tick{\conf_2} \trans[\tick{\hobs}]{\htxv} \tick{\conf_3} \\
  \Rightarrow 
  \hobs = \tick{\hobs}
  ~\land~ 
  \conf_3 \equivfor{\varfixpoint} \tick{\conf_3}
  \end{align*}
for observables $\hobs$, $\tick{\hobs}$.
\label{lem:com-user-tau}
\end{lemma}
\begin{proof}
Proof is analogous to that of Lemma~\ref{lem:com-user-attacker}, using Lemma~\ref{lem:fr-correctness-user-tau}.
\end{proof}

\begin{lemma}[Commutativity (Between $\txB$ and attacker transactions)]
    Let $\config_1$, $\tick{\config_1}$, $\config_2$, $\tick{\config_2}$, $\config_3$ and $\tick{\config_3}$ be configurations and $\honuser \in \honusers$ a user.
    Let $\varfixpoint \subseteq \gvars{\contract}$ be a set of variables such that $\criticalvars \subseteq \varfixpoint$ and $\depclosed{\varfixpoint}$.

Further, let $\atxv$ be a transaction such that $\sender{\atxv} \neq \honuser$. 
Assume that $\config_1 \equivfor{\varfixpoint} \tick{\config_1}$ and 
$\round(\config_1) = \round(\config_3)$ and 
$\round(\tick{\config_1}) = \round(\tick{\config_3})$ and
$\brCond{\config_1}{\txB}{\varfixpoint}$, 
$\brCond{\tick{\config_1}}{\txB}{\varfixpoint}$
hold. 

Then
\begin{align*}
  \conf_1 \trans[]{\txB} \conf_2 \trans[\aobs]{\atxv} \conf_3
  ~\land~ 
  \tick{\conf_1} \trans[\tick{\aobs}]{\atxv} \tick{\conf_2} \trans[]{\txB} \tick{\conf_3} \\
  \Rightarrow 
  \aobs = \tick{\aobs}
  ~\land~ 
  \conf_3 \equivfor{\varfixpoint} \tick{\conf_3}
  \end{align*}
for observables $\aobs$, $\tick{\aobs}$.
\label{lem:com-tau-attacker}
\end{lemma}
\begin{proof}
Proof is analogous to that of Lemma~\ref{lem:com-user-attacker}, using Lemma~\ref{lem:br-correctness-tau-attacker}.
\end{proof}

Based on the commutativity lemmas, we can now establish results for permutations between transaction sequences where 
(i) a (single) user transaction may be arbitrarily placed and
(ii) where an attacker may shift the block boundaries.

In the following, we will consider permutations $\perm$ on lists to be functions that apply finite sequences of pairwise swap operations among neighboring elements to a list. 

We will use the notation $\txlist \equaluptotx{x_1, \dots, x_n} \txlist'$ 
to denote that two lists $\txlist$, $\txlist'$ are equal up to the set $S = \{ x_1, \dots, x_n \}$, meaning that 
$\filter{\txlist}{\lambda x.~ x \not \in S} = \filter{\txlist'}{\lambda x.~ x \not \in S}$
where $\filter{\txlist}{P}$ denotes the list $\txlist$ with all elements not satisfying a predicate $P$ removed.

Note that we generally assume that all (faithfully signed) transactions can execute, but may indicate a potential failure with a dedicated observable. So, in particular, the prior execution of another transaction may not impact the executability of a transaction but may only alter its effect. 
This implies in particular that if a transaction sequence $\txlist$ is executable, then also every permutation $\perm(\txlist)$ of the same transaction sequence is executable.

The following lemma states that permuting a transaction sequence $\txlist$ of at most $k$ blocks and containing a single transaction $\htxv$ of the honest user, in a way such that the order of attacker transactions in $\perm(\txlist)$ stays unaffected, will not change the effects of executing this sequence in configurations that satisfy the frontrunning and backrunning conditions (provided that they agree on a set $\varfixpoint$ of variables that contain the event variables and that is closed under dependencies). 

\begin{lemma}[Permutation within rounds (with attacker transactions)]
    \label{lem:round-permutation-tx}
Let $\contract$ be a contract with variables $\cvars{\contract}$ and $\honuser \in \honusers$ be a user.
Let $\varfixpoint \subseteq \gvars{\contract}$ be a set of variables such that $\criticalvars \subseteq \varfixpoint$ and $\depclosed{\varfixpoint}$.
    Let $\conf, \tick{\conf}$ be configurations with $\conf \equivfor{\varfixpoint} \tick{\conf}$ 
    and $\lastBlockNumber(\conf) = \lastBlockNumber(\tick{\conf}) = r * \round(\conf)$ for some $r \in \mathbb{N}$.
    Let $\htxv$ be a transaction such that $\sender{\htxv} = \honuser$
    and $\frCondR{\conf}{\htxv}{\varfixpoint}$,
     $\frCondR{\tick{\conf}}{\htxv}{\varfixpoint}$,
    $\brCondR{\conf}{\htxv}{\varfixpoint}$, 
    $\brCondR{\tick{\conf}}{\htxv}{\varfixpoint}$, \break
    $\brCondR{\conf}{\txB}{\varfixpoint}$, 
    and  $\brCondR{\tick{\conf}}{\txB}{\varfixpoint}$.
   Then it holds that if 
$\conf \ssteps{\obslist}{\txlist} \conf'$
for some lists of observables $\obslist$, configurations $\conf'$ 
and transaction list $\txlist$ with $\length{\filter{\txlist}{\lambda x.~ x = \txB}} < k$ and $\filter{\txlist}{\lambda x.~ \sender{x} = \honuser} = [\htxv]$,
then also
for all permutations $\perm$ 
such that $\perm(\txlist) \equaluptotx{\htxv, \txB} \txlist$
there exist $\obslist'$, $\tick{\conf}'$ such that 
$\tick{\conf} \ssteps{\obslist'}{\perm(\txlist)} \tick{\conf}'$ and 
$\conf' \equivfor{\varfixpoint} \tick{\conf}'$
and $\obslist' = \perm(\obslist)$.
\end{lemma}
\begin{proof}
We note that requiring $\perm(\txlist) \equaluptotx{\htxv, \txB} \txlist$ ensures that such permutation $\perm$ only needs to perform pairwise swaps involving $\htxv$ or $\txB$.
Consequently, we prove the lemma by induction on the list of such swap operations performed by $\perm$. 
For each of the pairwise swaps involving $\ctx{\contract}$ or $\txB$, we can easily conclude the existence of the permuted transaction sequence with the desired properties from Lemmas~\ref{lem:com-user-attacker},~\ref{lem:com-user-tau},~\ref{lem:com-tau-attacker}, and relying on the fact that backrunning and frontrunning conditions persist through rounds as given by Lemmas~\ref{lem:frcond-through-rounds} and~\ref{lem:brcond-through-rounds}. 
\end{proof}

\begin{lemma}[Permutation within rounds (with $\txB$)]
    \label{lem:round-permutation-tau}
    Let $\contract$ be a contract with variables $\cvars{\contract}$ and $\honuser \in \honusers$ be a user.
Let $\varfixpoint \subseteq \gvars{\contract}$ be a set of variables such that $\criticalvars \subseteq \varfixpoint$ and $\depclosed{\varfixpoint}$.
    Let $\conf, \tick{\conf}$ be configurations with $\conf \equivfor{\varfixpoint} \tick{\conf}$ 
    and $\lastBlockNumber(\conf) = \lastBlockNumber(\tick{\conf}) = r * \round(\conf)$ for some $r \in \mathbb{N}$.
    Assume that $\brCondR{\conf}{\txB}{\varfixpoint}$, 
    and  $\brCondR{\tick{\conf}}{\txB}{\varfixpoint}$ hold.
   Then it holds that if 
$\conf \ssteps{\obslist}{\txlist} \conf'$
for some lists of observables $\obslist$, configurations $\conf'$ 
and transaction list $\txlist$ with $\length{\txlist} = k$,
then also
for all permutations $\perm$ 
such that $\perm(\txlist) \equaluptotx{\txB} \txlist$
there exist $\obslist'$, $\tick{\conf}'$ such that 
$\tick{\conf} \ssteps{\obslist'}{\perm(\txlist)} \tick{\conf}'$ and 
$\conf' \equivfor{\varfixpoint} \tick{\conf}'$
and $\obslist' = \perm(\obslist)$.
\end{lemma}
\begin{proof}
We note that requiring $\perm(\txlist) \equaluptotx{\txB} \txlist$ ensures that such permutation $\perm$ only needs to perform pairwise swaps involving $\txB$.
Consequently, we prove the lemma by induction on the list of such swap operations performed by $\perm$. 
For each of the pairwise swaps involving $\txB$, we can easily conclude the existence of the permuted transaction sequence with the desired properties from Lemma~\ref{lem:com-tau-attacker}, and relying on the fact that backrunning aconditions persist through rounds as given by Lemma~\ref{lem:brcond-through-rounds}. 
\end{proof}

\subsubsection{Simulator}
We formally define a generic simulator used for our soundness proof. 
The general idea behind the simulator is that at the beginning of every round, it publishes the mempool as the first block and then uses this knowledge to reconstruct the attacker's behavior for the respective round. 
To define our simulator, we first define the following function:

$\itattacker{\astrat}{\run}{\mempool}{n} := $ %
\begin{align*}
    \hspace{1.5cm}
    \begin{cases}
        \run & , n = 0 \\
        \itattacker{\astrat}{\run \xrightarrow[]{\bblock}^* \conf}{\mempool'}{n-1} & , n >  0 \\
        &\land~ \bblock = \astrat(\run, \mempool) \\
        & \land~ \conf_{\run} \xrightarrow[]{\bblock}^* \conf \\
        & \land~ \mempool' = \mempool \backslash \bblocks
    \end{cases}
\end{align*}

Intuitively, $\itattacker{\astrat}{\run}{\mempool}{n}$ invokes the attacker strategy $\astrat$ $n$-times with mempool $\mempool$ on run $\run$. 

\begin{lemma}[$\itattackerN$ Correctness]
    Let $\astrat$ be an attacker strategy, $\cstrat = \roundstrat{\rstrat}$ the user strategy for a round strategy $\rstrat$ 
    and $\arun$ a run such that 
    $\length{\getBlocks(\arun)} = r * k$ for some $r \in \mathbb{N}$ and 
    $\combine{\astrat}{\cstrat} \semantics \arun$ and $n \in \mathbb{N}$ such that $n \leq k$.
    Then there exists a run $\tick{\arun} = \itattacker{\astrat}{\run}{\cstrat(\arun)}{n}$ 
    such that $\length{\getBlocks(\tick{\arun})} = n$ and 
    $\combine{\astrat}{\cstrat} \semantics \concat{\arun}{\tick{\arun}}$.
    \label{lem:step-correctness}
\end{lemma}
\begin{proof}
    Follows by induction on $n$, exploiting the property of round-based user strategy $\cstrat$ to keep on scheduling the same transactions within the same round.
\end{proof}

Next, we define a function that will allow us to reconstruct an attacker run, from a simulator run, when assuming that the simulator published the mempool as the first block in each round. 
{\begin{align*}\small
    &\getAttackrun{\astrat}{\arun}{\srun}{i} := \\
    &
    \begin{cases}
        \arun & ,\length{\getBlocks(\srun)} = 0 \\
        \itattacker{\astrat}{\arun}{\bblocks[0]}{i} & ,\bblocks = \getBlocks(\srun) \\
        & \land \length{\bblocks} > 0 \land \length{\bblocks} < k \\ 
        \getAttackrun{\astrat}{\itattacker{\astrat}{\arun}{\bblocks_1[0]}{k}}{\run_2}{i} & ,\srun = \concat{\run_1}{\run_2} \\
        &  \land \bblocks_1 = \getBlocks(\run_1) \\
        & \land \bblocks_2 = \getBlocks(\run_2) \\
        & \land \length{\bblocks_1} = k \\
        & \confPred(\config_{\run_1}) \\
        \getAttackrun{\astrat}{\itattacker{\astrat}{\arun}{\nil}{k}}{\run_2}{i} & ,\srun = \concat{\run_1}{\run_2} \\
        &  \land \bblocks_1 = \getBlocks(\run_1) \\
        & \land \bblocks_2 = \getBlocks(\run_2) \\
        & \land \length{\bblocks_1} = k \\
        & \neg \confPred(\config_{\run_1})
    \end{cases}
\end{align*}}%
Intuitively, this function processes the simulator run $\srun$ from the left, taking chunks of blocks of size $k$ and produces the corresponding attacker run by using $\itattacker{\cdot}{\cdot}{\cdot}{k}$ to invoke the attacker $k$-times on the mempool extracted from the first block of the simulator run (if the condition $\confPred$ holds) or the empty mempool (if $\confPred$ does not hold) in each chunk.

\begin{lemma}[$\getAttackrunN$ Correctness]
    Let $\astrat$, $\sstrat$ be attacker strategies and $\cstrat = \roundstrat{\rstrat}$ be the user strategy for a round strategy $\rstrat$
    and $\srun$ be a run such that $\combine{\sstrat}{\cstrat} \semantics \srun$
    and $n \in \mathbb{N}$ such that $n \leq k$.
    Let $\confPred$ be a condition on configurations.
    Further, let $\length{\getBlocks(\srun)} = r * k + i$ for some $r, i \in \mathbb{N}$ and 
    $i < k$.
    Then there exists $\arun = \getAttackrun{\astrat}{\arun}{\srun}{n}$
    such that $\combine{\astrat}{\cstrat} \semantics \arun$
    and $\length{\getBlocks(\srun)} = r * k + n$.
    \label{lem:arun-correctness}
\end{lemma}
\begin{proof}
    Follows by size induction on the length of $\length{\getBlocks(\srun)}$, distinguishing the cases 
    $\length{\getBlocks(\srun)} < k$ and $\length{\getBlocks(\srun)} = k * r + j$ for $r,j \in \mathbb{N}$ for $j < k$.
    For the individual cases, we use Lemma~\ref{lem:step-correctness} in addition to the inductive hypothesis to conclude the proof.
\end{proof}

With these functions in place, we can now define the simulator: 

{
\begin{align*}
    \gensstrat(\srun, \mempool)
    := \begin{cases}
    \mempool & , \length{\getBlocks(\srun)} \bmod{k} = 0 ~\land~ \confPred(\srun) \\
    \bblock & 
    , \srun = \concat{\srunpre}{\srunpost} \land \confPred(\srunpre) \\
    & \land \exists r \in \mathbb{N}. \length{\getBlocks(\srunpre)} = k*r \\
    & \land \getBlocks(\srunpost) = [\amempool{0}]\\
    & \land \arun = \getAttackrun{\astrat}{\initconf}{\srun}{2} \\
    & \land \concat{\ablockspre}{\concat{[\ablock{0}]}{[\ablock{1}]}} =  \getBlocks(\arun) \\
    & \land \bblock = \ablock{0} \backslash \amempool{0} \cdot \ablock{1} \backslash \amempool{0} \\
    \bblock & 
    , \srun = \concat{\srunpre}{\srunpost} \land  ~\land~ \confPred(\srunpre)\\
    & \land \exists r \in \mathbb{N}. \length{\getBlocks(\srunpre)} = k*r \\
    & \land \length{\getBlocks(\srunpost)} = i \land i < k \land i > 1 \\  
    & \land \concat{[\amempool{0}]}{\blocks} = \getBlocks(\srunpost) \\
    & \land \arun = \getAttackrun{\astrat}{\initconf}{\srun}{i+1} \\
    & \land \concat{\ablockspre}{[\ablock{0}]} = \getBlocks(\arun) \\
    & \land \bblock = \ablock{0} \backslash \amempool{0} \\
    \bblock&, 
    \srun = \concat{\srunpre}{\srunpost} \land  ~\land~ \neg \confPred(\srunpre)\\
    & \land \exists r \in \mathbb{N}. \length{\getBlocks(\srunpre)} = k*r \\
    & \land \length{\getBlocks(\srunpost)} = i \land i < k  \\ 
    & \land \arun = \getAttackrun{\astrat}{\initconf}{\srun}{i+1} \\
    & \land \concat{\ablockspre}{[\bblock]} = \getBlocks(\arun) \\
    \end{cases}
    \end{align*}
}

Here $\initconf$ denotes the initial configuration of all runs (representing a fresh state).
Note that $\getAttackrun{\astrat}{\initconf}{\srun}{j}$ simulates the attacker behavior for all complete rounds and additional $j$ steps. 
For this reason, the function is invoked with $j=2$ in the second case (to produce two attacker blocks for the newly started round) and $j = i+1$ in the third case (to simulate the attacker block once ahead of the current simulated run $\srun$).

Also, note that it is evident that the simulator satisfies the substitution criterion, since its definition does not make any adaptive use of the mempool.

The simulator satisfies the key property that for a round-based user strategy $\cstrat$, within one round, it executes the same transactions as the attacker, with the only difference that the (single) user transaction scheduled by $\cstrat$ may occur in a different position. 
All other transactions are executed in the same order. 

We capture this behavior with the following core property:

\begin{lemma}[Correctness of $\gensstrat$ (in $\confPred$ rounds)]
    Let $\honuser$ be an honest user and $\cstrat = \roundstrat{\rstrat}$ be the user strategy for a round strategy $\rstrat$ (for $\honuser$) and let $\astrat$ be an attacker strategy. 
    Let $\confPred$ be a condition on configurations.
    Let $k > 1$, $\txv$ be a transaction and $\arunpre$ and $\arun$ be an attacker run
    such that $\combine{\astrat}{\cstrat} \semantics \concat{\arunpre}{\arun}$
    and $\length{\getBlocks(\arunpre)} = r* k$ for some $r \in \mathbb{N}$
    and $\length{\getBlocks(\arun)} = i$ for $i \in \mathbb{N}$ with $i \leq k$.
    Let $\arun = \conf_{\arunpre} \ssteps{}{\txlist} \aconf$ for some configuration $\aconf$ and transaction list $\txlist$.
    Further, let $\srunpre$ be a run such that 
    $\combine{\gensstrat}{\cstrat} \semantics \srunpre$,
    and $\length{\getBlocks(\srunpre)} = r * k$
    and $\arunpre = \getAttackrun{\astrat}{\initconf}{\srunpre}{0}$
    and $\rstrat(\config_{\arunpre}) = \rstrat(\config_{\srunpre}) = \{ \txv \}$
    and $\confPred(\srunpre)$.
    Then there exists a run 
    \[
        \srun = \conf_{\srunpre} \ssteps{}{\txlist'} \sconf 
    \]
    for some $\sconf$
    and transaction list $\txlist'$
    such that $\combine{\gensstrat}{\cstrat} \semantics \concat{\srunpre}{\srun}$, 
    $\length{\getBlocks(\srun)} = i$, 
    $\txlist \equaluptotx{\txv, \txB} \txlist'$
    and $\arun = \getAttackrun{\astrat}{\initconf}{\srun}{i \bmod k}$ and
    if $i > 0$ then $\getBlocks(\srun)[0] = [\txv]$.
    \label{lem:sim-correctness}
\end{lemma}
\begin{proof}
This proof is carried out by induction on $i \in \mathbb{N}$.
The case $i = 0$ follows trivially from the assumption by choosing $\srun = \conf_{\srunpre}$.
In the case $i = 1$, we can choose $\srun = \conf_{\srunpre} \ssteps{}{[\txv]} \sconf$ for a corresponding $\sconf$ (which is guaranteed to exist since $\txv$ is the output of $\cstrat$) to derive the desired properties immediately.
In the case $i = 2$, we can choose $\srun = \conf_{\srunpre} \ssteps{}{[\txv]} \sconf \ssteps{}{\bblock} \sconf'$ for $\bblock$ as defined in the simulator definition. The claim follows then immediately from the definition of the simulator. 
In the case $i > 2$ the reasoning follows immediately from the inductive hypothesis, and the definitions of $\gensstrat$ and $\getAttackrunN$.
\end{proof}

\begin{lemma}[Correctness of $\gensstrat$ (in $\neg \confPred$ rounds)]
    Let $\honuser$ be an honest user and $\cstrat = \roundstrat{\rstrat}$ be the user strategy for a round strategy $\rstrat$ (for $\honuser$) and let $\astrat$ be an attacker strategy. 
    Let $\confPred$ be a condition on configurations.
    Let $k > 1$, $\txv$ be a transaction and $\arunpre$ and $\arun$ be an attacker run
    such that $\combine{\astrat}{\cstrat} \semantics \concat{\arunpre}{\arun}$
    and $\length{\getBlocks(\arunpre)} = r* k$ for some $r \in \mathbb{N}$
    and $\length{\getBlocks(\arun)} = i$ for $i \in \mathbb{N}$ with $i \leq k$.
    Let $\arun = \conf_{\arunpre} \ssteps{}{\txlist} \aconf$ for some configuration $\aconf$ and transaction list $\txlist$.
    Further, let $\srunpre$ be a run such that 
    $\combine{\gensstrat}{\cstrat} \semantics \srunpre$,
    and $\length{\getBlocks(\srunpre)} = r * k$
    and $\arunpre = \getAttackrun{\astrat}{\initconf}{\srunpre}{0}$
    and $\rstrat(\config_{\arunpre}) = \rstrat(\config_{\srunpre}) = \emptyset$
    and $\neg \confPred(\srunpre)$.
    Then there exists a run 
    \[
        \srun = \conf_{\srunpre} \ssteps{}{\txlist} \sconf 
    \]
    for some $\sconf$
    such that $\combine{\gensstrat}{\cstrat} \semantics \concat{\srunpre}{\srun}$, 
    $\length{\getBlocks(\srun)} = i$, 
    and 
    \break $\arun = \getAttackrun{\astrat}{\initconf}{\srun}{i \bmod k}$.
    \label{lem:sim-correctness-not-conf-pred}
\end{lemma}
\begin{proof}
This proof is carried out by induction on $i \in \mathbb{N}$.
The case $i = 0$ follows trivially from the assumption by choosing $\srun = \conf_{\srunpre}$.
For $i > 0$, the claim follows immediately from the inductive hypothesis, and the definitions of $\gensstrat$ and $\getAttackrunN$.
\end{proof}

\begin{lemma}[Correctness of $\gensstrat$ (in $\confPred$ rounds, no transaction)]
    Let $\honuser$ be an honest user and $\cstrat = \roundstrat{\rstrat}$ be the user strategy for a round strategy $\rstrat$ (for $\honuser$) and let $\astrat$ be an attacker strategy. 
    Let $\confPred$ be a condition on configurations.
    Let $k > 1$, $\txv$ be a transaction and $\arunpre$ and $\arun$ be an attacker run
    such that $\combine{\astrat}{\cstrat} \semantics \concat{\arunpre}{\arun}$
    and $\length{\getBlocks(\arunpre)} = r* k$ for some $r \in \mathbb{N}$
    and $\length{\getBlocks(\arun)} = i$ for $i \in \mathbb{N}$ with $i \leq k$.
    Let $\arun = \conf_{\arunpre} \ssteps{}{\txlist} \aconf$ for some configuration $\aconf$ and transaction list $\txlist$.
    Further, let $\srunpre$ be a run such that 
    $\combine{\gensstrat}{\cstrat} \semantics \srunpre$,
    and $\length{\getBlocks(\srunpre)} = r * k$
    and $\arunpre = \getAttackrun{\astrat}{\initconf}{\srunpre}{0}$
    and $\rstrat(\config_{\arunpre}) = \rstrat(\config_{\srunpre}) = \emptyset$
    and $\confPred(\srunpre)$.
    Then there exists a run 
    \[
        \srun = \conf_{\srunpre} \ssteps{}{\txlist'} \sconf 
    \]
    for some $\sconf$
    and transaction list $\txlist'$
    such that $\combine{\gensstrat}{\cstrat} \semantics \concat{\srunpre}{\srun}$, 
     $\length{\getBlocks(\srun)} = i$, 
    $\txlist \equaluptotx{\txB} \txlist'$
    and 
    \break $\arun = \getAttackrun{\astrat}{\initconf}{\srun}{i \bmod k}$.
    \label{lem:sim-correctness-not-conf-pred-no-tx}
\end{lemma}
\begin{proof}
This proof is carried out by induction on $i \in \mathbb{N}$.
The case $i = 0$ follows trivially from the assumption by choosing $\srun = \conf_{\srunpre}$.
For $i > 0$ the claim follows immediately from the inductive hypothesis, and the definitions of $\gensstrat$ and $\getAttackrunN$.
\end{proof}

\subsubsection{Similarity Relation}

We define the similarity relation that our simulator $\sstrat_{\astrat}$ can uphold. 
Intuitively, the simulator can be guaranteed to perform all events in the same order, with the only exception of user-triggered events whose position could be shifted within one round. 

Note that this similarity is very strong in the sense that for the case where relevant events can only be triggered by user transactions (which is a very common case to be considered), the observables of the traces are fully identical. 
Even when considering attacker transactions to trigger observables, the provided property still yields useful notions when considering observations, whose concrete ordering within one round is irrelevant (e.g., when tracking money transfers).
Since we do not distinguish on the observation level how an observation was created, we will show our result for the general simulation relation, which does not fix the order of observables within a round:

\vspace{-10pt}
{
\begin{align*}
&\arun \simsim \srun
 :\Leftrightarrow \\
&\forall m, m', r, r'.~
\length{\getBlocks(\arun)} = r*k + m 
~\land~ m < k \\
& \quad \land \length{\getBlocks(\srun)} = r'*k + m ' 
~\land~ m' < k \\
& \qquad \Rightarrow r = r' 
~\land 
\forall j.~ j < r \\
&\qquad \land~ \arun = \concat{\arunpre}{\concat{\arunround{j}}{\arunpost}}
\land \srun = \concat{\srunpre}{\concat{\srunround{j}}{\srunpost}} \\
& \qquad \land~ \length{\getBlocks(\arunpre)} = j*k
~\land~ \length{\getBlocks(\srunpre)} = j*k \\
& \qquad \land~ \length{\getBlocks(\arunround{j})} = k
~\land~ \length{\getBlocks(\srunround{j})} = k \\
& \qquad \land~ \arunround{j} = \aconf \ssteps{\obslist}{\txlist} \tick{\aconf}
~\land~ \srunround{j} = \sconf \ssteps{\obslist'}{\txlist'} \tick{\sconf} \\
& \qquad \Rightarrow (\forall \obs.~\obs \in \obslist \Leftrightarrow \obs \in \obslist')
\end{align*}}%
Note that here all unbound variables can be implicitly assumed to be universally quantified.

\subsubsection{Soundness Proof}

We are now in the position to show the soundness of the generated interaction conditions. 

We first show the following slightly strengthened version of the soundness lemma, which will allow us to conclude the final soundness theorem.

\begin{lemma}[Soundness]
    \label{lem:soundness}
     Let $\contract$ be a contract with contract variables $\cvars{\contract}$ and $\criticalvars \subseteq \cvars{\contract}$ the set of critical variables in $\contract$ and $\honuser$ be a user. 
   Let $\{ \symExRel{\procF} \}_{\procF \in \cFuncs{\contract}, \indexvar 
   \in \{ \userindexvar, \nuserindexvar \} }$ be sound and complete symbolic execution relations over $\seCvars{\contract}$ with $\indexvar 
   \in \{ \userindexvar, \nuserindexvar \}$. 
    Let $\varfixpoint \subseteq \gvars{\contract}$ s.t. 
    $\depclosed{\varfixpoint}$ and $\criticalvars \subseteq \varfixpoint$.
    Let $\cstrat = \roundstrat{\rstrat}$ be the user strategy for a round strategy $\rstrat$ (for $\honuser$ and $\contract$) such that
    \begin{enumerate} 
        \item for all configurations $\conf, \conf'$ with $\conf \equivfor{\varfixpoint} \conf'$ it holds that $\rstrat(\conf) = \rstrat(\conf')$
        \item for all configurations $\conf$, $\rstrat(\conf) \neq \emptyset$ implies that $\rstrat(\conf) = \{ \txv \}$ for some $\txv$ 
        such that 
        $\frCondR{\config}{\txv}{\varfixpoint}$ and 
        $\brCondR{\config}{\txv}{\varfixpoint}$, and 
        $\brCondR{\config}{\txB}{\varfixpoint}$ hold.
    \end{enumerate}
    Let $\confPred$ be defined as follows:
    \begin{align*}
        \confPred(\conf) :\Leftrightarrow \brCondR{\conf}{\txB}{\varfixpoint} ~\land~ \exists \txv.~ 
        \frCondR{\conf}{\txv}{\varfixpoint}
        ~\land~ \brCondR{\conf}{\txv}{\varfixpoint}
    \end{align*}
    Let $\arun$ be an attacker run such that $\combine{\astrat}{\cstrat} \semantics \arun$ 
    and $\length{\getBlocks(\arun)} = k * r$ for $r \in \mathbb{N}$.
    Then there exists a run $\srun$ such that $\combine{\gensstrat[\confPred]}{\cstrat} \semantics \srun$
    and $\length{\getBlocks(\srun)} = k * r$
    and $\config_{\arun} \equivfor{\varfixpoint} \conf_{\srun}$
    and $\arun \simsim \srun$
    and $\arun = \getAttackrun[\confPred]{\astrat}{\initconf}{\srun}{0}$.
\end{lemma}

\begin{proof}
    The proof proceeds by induction on $r \in \mathbb{N}$.
    For $r = 0$, we know that $\arun = \initconf$. 
    By choosing $\srun = \initconf$ the claims trivially hold since $\arun = \srun$.
    For $r > 0$, we know that $\arun = \concat{\arunpre}{\arunpost}$ such that $\length{\getBlocks(\arunpre)} = (r-1)*k$ and 
     $\length{\getBlocks(\arunpost)} = k$. 
     From the inductive hypothesis, we hence get a valid simulator run $\srunpre$ such that  
     $\length{\getBlocks(\srunpre)} = (r-1)*k$,
     $\config_{\arunpre} \equivfor{\varfixpoint} \conf_{\srunpre}$ and 
     $\arunpre \simsim \srunpre$.
     We distinguish the cases of whether $\confPred(\conf_{\srunpre})$ holds. %
     \begin{itemize}
        \item If $\confPred(\conf_{\srunpre})$ does not hold, we know by assumption that $\cstrat(\conf_{\srunpre}) = \emptyset$. Since $\rstrat$ agrees on all configurations with the same values for variables in $\varfixpoint$, we have that also $\cstrat(\conf_{\arunpre}) = \emptyset$. 
        We can now use Lemma~\ref{lem:sim-correctness-not-conf-pred} to deduce the existence of a valid simulator run $\srunpost$
     such that $\length{\getBlocks(\srunpost)} = k$.
      Further, this gives us that for $\arunpost = \conf_{\arunpre} \ssteps{\obs}{\txlist} \aconf$
     also $\srunpost = \conf_{\srunpre} \ssteps{\obs'}{\txlist}$. %
     We hence, in particular, know that $\aconf \equivfor{\varfixpoint} \sconf$ and $\obs = \obs'$ (since the initial configuration agreed on $\varfixpoint$).
      Together with $\arunpre \simsim \srunpre$ by the definition of $\simsim$, this also gives us that $\arun \simsim \srunpre$. 
      Also, $\arun = \getAttackrun{\astrat}{\initconf}{\srun}{0}$ follows immediately from Lemma~\ref{lem:arun-correctness} and the definition of $\getAttackrunN$ (note that the last argument to $\getAttackrunN$ denotes the additional invocations of the attacker after a complete round).
        \item 
        If $\confPred(\conf_{\srunpre})$ holds, we know that either $\cstrat(\conf_{\srunpre}) = \emptyset$ or $\cstrat(\conf_{\srunpre}) = \{ \txv \}$. 
        We distinguish these two cases:
        \begin{itemize}
            \item If $\cstrat(\conf_{\srunpre}) = \emptyset$ by assumption, we know in this case that $\cstrat(\conf_{\arunpre}) = \cstrat(\conf_{\srunpre})$. We can, hence, use Lemma~\ref{lem:sim-correctness-not-conf-pred-no-tx} to deduce the existence of a valid simulator run $\srunpost$
     such that $\length{\getBlocks(\srunpost)} = k$. 
     Further, this gives us that for $\arunpost = \conf_{\arunpre} \ssteps{\obs}{\txlist} \aconf$
     also $\srunpost = \conf_{\srunpre} \ssteps{\obs'}{\txlist'} \sconf$ for some $\txlist'$ with $\txlist \equaluptotx{\txB} \txlist'$.
     Since $\length{\getBlocks(\srunpost)} = k = \length{\getBlocks(\arunpost)}$, we know that $\txlist$ and $\txlist'$ must contain the same number of $\txB$ transactions and consequently that there needs to exist a permutation $\perm$ such that $\txlist' = \perm(\txlist)$. 
     Using Lemma~\ref{lem:round-permutation-tau}, we can hence show that 
     also $\aconf \equivfor{\varfixpoint} \sconf$ and $\obs' = \perm(\obs)$.
     Together with $\arunpre \simsim \srunpre$ by the definition of $\simsim$, this also gives us that $\arun \simsim \srunpre$.
     Also, $\arun = \getAttackrun{\astrat}{\initconf}{\srun}{0}$ follows immediately from Lemma~\ref{lem:arun-correctness} and the definition of $\getAttackrunN$.
            \item If $\cstrat(\conf_{\srunpre}) = \{ \txv \}$ by assumption, we know in this case that $\cstrat(\conf_{\arunpre}) = \cstrat(\conf_{\srunpre}) = \{ \txv \}$ and $\txv = \contrCall{\procF(\fargs)}$ for some $\procF \in \cFuncs{\contract}$ and that the interaction conditions hold in $\conf_{\srunpre}$ and $\conf_{\arunpre}$. 
We can, hence, use Lemma~\ref{lem:sim-correctness} to deduce the existence of a valid simulator run $\srunpost$
     such that $\length{\getBlocks(\srunpost)} = k$. 
     Further, this gives us that for $\arunpost = \conf_{\arunpre} \ssteps{\obs}{\txlist} \aconf$
     also $\srunpost = \conf_{\srunpre} \ssteps{\obs'}{\txlist'} \sconf$ for some $\txlist'$ with $\txlist \equaluptotx{\txv, \txB} \txlist'$.
     We know from the transaction inclusion guarantee that since $\txv \in \cstrat(\conf_{\arunpre})$ also $\txv \in \txlist$.
     We also know (from Definition~\ref{lem:sim-correctness}) that $\getBlocks(\srunpost)[0] = [\txv]$ and so $\txv \in \txlist'$.
     Since transactions only appear once in a given run, we get from $\txlist \equaluptotx{\txv} \txlist'$ that there needs to exist a permutation $\perm$ such that $\txlist' = \perm(\txlist)$.
     Using Lemma~\ref{lem:round-permutation-tx}, we can hence show that 
     also $\aconf \equivfor{\varfixpoint} \sconf$ and $\obs' = \perm(\obs)$.
     Together with $\arunpre \simsim \srunpre$ by the definition of $\simsim$, this also gives us that $\arun \simsim \srunpre$.
     Also, $\arun = \getAttackrun{\astrat}{\initconf}{\srun}{0}$ follows immediately from Lemma~\ref{lem:arun-correctness} and the definition of $\getAttackrunN$.
        \end{itemize}
     \end{itemize}
\end{proof}

With this, we show our final soundness theorem, which establishes that a smart contract and together with any round-based user strategy that respects the synthesized interaction conditions for the contract, satisfies \sysnameAc.

\begin{theorem}[Soundness of Interaction Conditions]
Let $\contract$ be a contract with contract variables $\cvars{\contract}$ and $\criticalvars \subseteq \cvars{\contract}$ the set of critical variables in $\contract$ and $\honuser$ be a user. 
   Let $\{ \symExRel{\procF} \}_{\procF \in \cFuncs{\contract}, \indexvar 
   \in \{ \userindexvar, \nuserindexvar \} }$ be sound and complete symbolic execution relations over $\seCvars{\contract}$ with $\indexvar 
   \in \{ \userindexvar, \nuserindexvar \}$. 
   Let $\varfixpoint \subseteq \gvars{\contract}$ s.t. 
    $\depclosed{\varfixpoint}$ and $\criticalvars \subseteq \varfixpoint$.
    Let $\cstrat = \roundstrat{\rstrat}$ be the user strategy for a round strategy $\rstrat$ (for $\honuser$ and $\contract$) such that
    \begin{enumerate} 
        \item for all configurations $\conf, \conf'$ with $\conf \equivfor{\varfixpoint} \conf'$ it holds that $\rstrat(\conf) = \rstrat(\conf')$
        \item for all configurations $\conf$, $\rstrat(\conf) \neq \emptyset$ implies that $\rstrat(\conf) = \{ \txv \}$ for some $\txv$ 
        such that 
        $\frCondR{\config}{\txv}{\varfixpoint}$ and 
        $\brCondR{\config}{\txv}{\varfixpoint}$ and
        $\brCondR{\conf}{\txB}{\varfixpoint}$ hold.
    \end{enumerate}
    Let $\confPred$ be defined as follows:
    \begin{align*}
        \confPred(\conf) :\Leftrightarrow \brCondR{\conf}{\txB}{\varfixpoint} ~\land~ \exists \txv.~ 
        \frCondR{\conf}{\txv}{\varfixpoint}
        ~\land~ \brCondR{\conf}{\txv}{\varfixpoint}
    \end{align*}
    Let $\arun$ be an attacker run such that $\combine{\astrat}{\cstrat} \semantics \arun$.
    Then there exists a run $\srun$ such that $\combine{\gensstrat}{\cstrat} \semantics \srun$
    and $\arun \simsim \srun$.
\end{theorem}
\begin{proof}
The statement follows immediately from Lemma~\ref{lem:soundness} and the definition of $\simsim$ since we know that for a run $\arun$ it needs to hold that $\arun = \concat{\arunpre}{\arunpost}$ such that $\length{\getBlocks(\arunpre)} = r *k$ and $\length{\getBlocks(\arunpost)} = i$ for some $i,k \in \mathbb{N}$ with $i < k$. 
Consequently, from Lemma~\ref{lem:soundness}, we can deduce the existence of a run $\srunpre$
such that $\length{\getBlocks(\arunpre)} = r*k$ and $\arunpre \simsim \srunpre$.
By the definition of $\simsim$, we know that then it also holds that $\arun \simsim \srunpre$, which concludes the proof.
\end{proof}

Note that this theorem indeed proves \sysnameAc since the simulator $\gensstrat$ only depends on the attacker $\astrat$ and the contract $\contract$ (via the condition $\confPred$ that is derived from the contract code using the algorithm for synthesizing interaction conditions) but the user strategy $\cstrat$ remains arbitrary (satisfying the required conditions). 
So, in particular, the theorem implies \sysnameAc for all sets of user strategies $\ustrat{\contract}{\honusers}{}$ such that all $\cstrat \in \ustrat{\contract}{\honusers}{}$ satisfy the required conditions.

\subsection{Relation to TOD}
\label{appendix:tod-counterexample}

Contracts with non-round-based user strategies can be free of TOD and still not satisfy \sysnameAc. 
For example, consider the following two functions which are trivially free of TOD, since none of them change the contract state:
\begin{lstlisting}
function test(uint b) { /* Noop  */ }
function trigger() {
   emit Triggered(msg.sender);
}
\end{lstlisting}
Still, for the following set of user strategies, the contract does not satisfy \sysnameAc:
\begin{align*}
    \sustrat{\text{\lstinline|noTOD|}} :=
    \{ 
        &\cstrat_b ~|~
        b \in \{0,1\}
        ~\land~
        \cstrat_b(\config) \neq [] \\
        & ~\Rightarrow
         \cstrat_b(\run) = [\testFunc(\honuser:0\token,b)] \\
        & \qquad \land~ \testFunc(\honuser:0\token,b) \not \in \run \\
         & \quad \lor \cstrat_b(\run) = [\triggerFunc(\honuser:0\token)] \\
         & \qquad \land~  \run = \run_\textit{pre}\concat{\trans{\testFunc(\honuser:0\token,b)} \config'}{\run_\textit{post}} \\
        & \qquad \land \length{\run_\textit{pre}} \bmod 2 = b 
    \} 
\end{align*}
$\sustrat{\text{\lstinline|noTOD|}}$ contains two user strategies, $\cstrat_0$ and $\cstrat_1$. 
Both strategies initially submit a transaction to invoke the \lstinline|test| function with argument \lstinline|b|.
Then, $\cstrat_0$ only submits a \lstinline|trigger| transaction, exactly if their \lstinline|test| transaction was included in an even block and $\cstrat_0$ submits a \lstinline|trigger| transaction, exactly if their \lstinline|test| transaction was included in an odd block. 
Now, an adaptive attacker $\astrat$ can, when observing the argument of the \lstinline|test| transaction, deduce in which block they need to include it to make the user strategy submit the \lstinline|trigger| transaction next. Like this, $\astrat$ can reliably produce the event $\triggeredEvent(\honuser)$ for the honest user $\honuser$ for both $\cstrat_0$ and $\cstrat_1$.
Instead, a non-adaptive attacker always needs to place the user transaction in the same block and hence will not be able to produce the event $\triggeredEvent(\honuser)$ for both $\cstrat_0$ and $\cstrat_1$. This example, while artificial, underlines how a user strategy's dependence can easily enable \deckstacking attacks.

\end{document}